\documentclass[10pt,letterpaper]{article}

\usepackage[utf8]{inputenc}
\usepackage[T1]{fontenc}
\usepackage[english]{babel}
\usepackage{csquotes}
\usepackage{CJKutf8}
\newcommand{\chinese}[1]{\begin{CJK*}{UTF8}{gbsn}#1\end{CJK*}}

\usepackage[margin=0.75in,columnsep=0.22in]{geometry}

\usepackage{setspace}
\usepackage{graphicx}

\usepackage{tikz,subcaption,mathtools}
\usetikzlibrary{arrows.meta}

\usepackage{amsmath, amssymb, amsfonts}
\usepackage{amsthm, mathrsfs, mathtools}
\usepackage{algorithm}
\usepackage{algorithmic}

\usepackage{booktabs}

\usepackage{hyperref}
\hypersetup{
  colorlinks=false,
  linktoc=page,
  pdfborder={0 0 0.3},
  linkbordercolor={0.5 0.5 0.5},
  citebordercolor={0.5 0.5 0.5},
  urlbordercolor={0.5 0.5 0.5}
}
\usepackage[capitalise]{cleveref}

\crefname{appendix}{Section}{Sections}
\Crefname{appendix}{Section}{Sections}
\crefname{subappendix}{Section}{Sections}
\Crefname{subappendix}{Section}{Sections}

\usepackage[backend=biber,style=alphabetic, minbibnames=99,sorting=nty,sortcites=true,maxbibnames=99]{biblatex}
\usepackage{titling}

\usepackage{authblk}

\usepackage{enumitem} % Customized lists
\usepackage{caption} % Customizing captions
\usepackage{subcaption} % Sub-figures
\usepackage{booktabs} % Professional tables
\usepackage{microtype} % Micro-typographic enhancements
\usepackage{dsfont} % Double-struck fonts
\usepackage{faktor} % Quotient structures

\usepackage[baseline]{euflag}

\setlist[itemize]{leftmargin=*,topsep=2pt,itemsep=1pt}
\setlist[enumerate]{leftmargin=*,topsep=2pt,itemsep=1pt}

\usepackage{titlesec}

\titleformat{\paragraph}[block]
    {\normalfont\normalsize\bfseries}
    {}
    {0pt}
    {}

\titlespacing*{\paragraph}
    {0pt}{1.0\baselineskip}{0.35\baselineskip}

\setlist[itemize]{leftmargin=*,topsep=4pt,itemsep=3pt}
\setlist[enumerate]{leftmargin=*,topsep=4pt,itemsep=3pt}

\makeatletter
\renewcommand\tableofcontents{%
  \par\addvspace{2em}%
  {\large\bfseries\contentsname}\par\vspace{0.8em}%
  \@starttoc{toc}%
}
\makeatother

\let\epsilon\varepsilon

\let\phi\varphi

\theoremstyle{plain}
\newtheorem{theorem}{Theorem}[section]
\newtheorem{lemma}[theorem]{Lemma}
\newtheorem{corollary}[theorem]{Corollary}
\newtheorem{proposition}[theorem]{Proposition}

\newtheorem{definition}[theorem]{Definition}
\newtheorem*{theorem*}{Theorem}
\newtheorem*{corollary*}{Corollary}

\theoremstyle{definition}
\newtheorem{example}[theorem]{Example}
\newtheorem{remark}[theorem]{Remark}

\newcommand*{\Tr}[1]{\mathop{}\!\mathrm{Tr}{\left(#1\right)}}

\newcommand{\Span}[1]{\mathrm{Span}\left\lbrace #1 \right\rbrace}
\newcommand{\SpanNoBracket}[1]{\mathrm{Span} #1 }
\newcommand{\Spectrum}[1]{\mathrm{Spec}\left( #1 \right)}
\newcommand{\SpectrumNoBracket}[1]{\mathrm{Spec} #1 }

\def\id{\mathds{1}}

\newcommand{\bra}[1]{\ensuremath{\left\langle#1\right|}}
\newcommand{\ket}[1]{\ensuremath{\left|#1\right\rangle}}
\newcommand{\braket}[2]{\langle #1  |#2\rangle}
\newcommand{\ketbra}[2]{|#1\rangle\langle #2|  }
\newcommand{\sandwich}[3]{\langle #1|#2 |#3\rangle  }
\newcommand{\norm}[1]{\left\lVert#1\right\rVert}
\newcommand{\abs}[1]{\left\lvert#1\right\rvert}

\newcommand{\cH}{\mathcal{H}}
\newcommand{\cHA}{\mathcal{H}_{A}}
\newcommand{\cHB}{\mathcal{H}_{B}}
\newcommand{\cA}{\mathcal{A}}

\newcommand{\cK}{\mathcal{K}}
\newcommand{\cG}{\mathcal{G}}

\newcommand{\cQ}{\mathcal{Q}}
\newcommand{\thinS}{\mathscr{S}}
\newcommand{\fatS}{\pmb{\mathscr{S}}}

\newcommand{\familyH}{\mathcal{F}}
\newcommand{\latticeL}{\Lambda(L)}
\newcommand{\OBC}{\mathrm{OBC}}
\newcommand{\PBC}{\mathrm{PBC}}

\newcommand{\proofstep}[2]{%
  \par
  \addvspace{\medskipamount}%
  \noindent\emph{Step #1: #2}\par\nobreak
  \addvspace{\smallskipamount}%
}

\usepackage{xcolor}

\hypersetup{
  pdftitle = {The bulk spectral gap is certifiable from above but uncomputable from below},
  pdfauthor = {Xiangling Xu, Matthias Sch\"{o}tz, Jie Wang, Victor Magron, Igor Klep, Omar Fawzi, Marc-Olivier Renou},
  pdfsubject = {Quantum many-body physics; spectral gaps in the thermodynamic limit; undecidability; semidefinite programming},
  pdfkeywords = {bulk spectral gap, spectral gap, thermodynamic limit, quantum many-body systems, quantum spin chains, undecidability, halting problem, RE-completeness, semidecidability, certified upper bounds, semidefinite programming, state polynomial optimization, operator algebras, KMS ground state, kagome lattice Heisenberg antiferromagnet},
  pdflang = {en},
  pdfcreator = {LaTeX with hyperref},
  pdfdisplaydoctitle = true
}

\title{\large\bfseries
The bulk spectral gap is certifiable from above but uncomputable from below
}

\author[1,$*$]{Xiangling Xu (\chinese{许湘灵})}
\author[2]{Matthias Sch\"{o}tz}
\author[3,$\dagger$]{Jie Wang (\chinese{王杰})}
\author[4]{\authorcr Victor Magron}
\author[5,6]{Igor Klep}
\author[7]{Omar Fawzi}
\author[1]{Marc-Olivier Renou}

\affil[1]{Inria, CPHT, LIX, CNRS, \'Ecole Polytechnique, Institut Polytechnique de Paris, Palaiseau, France}
\affil[2]{Universität W\"{u}rzburg, Institute of Mathematics, Emil-Fischer-Stra{\ss}e 31, 97074 W\"{u}rzburg, Germany}
\affil[3]{State Key Laboratory of Mathematical Sciences, Academy of Mathematics and Systems Science,\protect\\ Chinese Academy of Sciences, Beijing, China}
\affil[4]{Universit\'e de Toulouse; LAAS-CNRS, 7 avenue du colonel Roche, F-31400 Toulouse, France}
\affil[5]{University of Ljubljana, Faculty of Mathematics and Physics, Jadranska 21, 1000 Ljubljana, Slovenia}
\affil[6]{University of Primorska, Faculty of Mathematics, Natural Sciences and Information Technologies,\protect\\ Glagolja\v{s}ka 8, 6000 Koper, Slovenia}
\affil[7]{Inria, ENS Lyon, UCBL, LIP, F-69342 Lyon Cedex 07, France}

\date{\vspace{-3.00em}}

\begin{document}
% \maketitle
%---------------------------------%
% Corresponding emails
\twocolumn[
\maketitle
\vspace{-2.00em}
\begin{center}
{\scriptsize
\(*\)\,\texttt{xu.xiangling@inria.fr}
\qquad
\(\dagger\)\,\texttt{wangjie212@amss.ac.cn}
}
\end{center}
\vspace{0.0em}

%---------------------------------%
% Abstract
\begin{center}
    \begin{minipage}{0.85\textwidth}
        \small
        \textbf{Abstract.}
        Determining spectral gaps in the thermodynamic limit is a central challenge in quantum many-body physics.
        When a fixed boundary condition is imposed, this problem is undecidable: no algorithm, even with unlimited resources, can decide for every Hamiltonian whether it is gapped.
        Yet the thermodynamic limit is meant to capture the intrinsic properties of very large systems, which should not depend on how their boundaries are chosen.
        We therefore consider the bulk spectral gap, which concerns only the intrinsic excitations of the infinite system, and ask what can and cannot be computed about it.
        First, we give a general algorithm producing a hierarchy of certified upper bounds on the bulk gap, arbitrarily tight at the cost of more computation.
        Second, we prove that the spectral gap problem remains undecidable in the bulk, already for translationally invariant nearest-neighbour chains of fixed local dimension.
        Hence no general algorithm produces convergent lower bounds, and no convergent upper-bounding algorithm, ours included, admits a computable convergence rate.
        The bulk spectral gap is thus certifiable from above but uncomputable from below; more precisely, the bulk-gaplessness promise problem is $\mathsf{RE}$-complete, exactly as hard as the halting problem.
        The same construction shows that no general algorithm can decide whether a given product state is a ground state, and that convergent bounds on ground-state expectation values admit no computable convergence rate in general.
        Finally, we obtain the first nontrivial certified upper bounds on the bulk gap of the spin-$\frac{1}{2}$ kagome Heisenberg antiferromagnet.
    \end{minipage}
\end{center}

% ---------------------------------%
% Popular Summary
% \begin{center}
%     \begin{minipage}{0.85\textwidth}
%         \small
%         \textbf{Popular Summary.}
%         The spectral gap, the smallest energy needed to excite a quantum material, shapes much of its behaviour. It is notoriously hard to determine. Celebrated results proved that, in general, no algorithm can decide whether it is nonzero. But they concern a gap that can depend on how the boundary is treated. This leads us to ask: is this limit to computation an artefact of the boundary, or a genuine feature of the bulk, the material's interior?

%         The answer is split: the bulk gap is certifiable from above, uncomputable from below. We give an algorithm whose rigorous upper bounds approach the true gap with more computation, applying it to the much-debated kagome antiferromagnet. We prove that, already in one dimension, no general algorithm can approach the gap from below, regardless of computing power. Thus, some limits to computation are intrinsic to the interior of quantum materials, not artefacts of their boundaries.
%     \end{minipage}
% \end{center}
\vspace{1em}

]
%---------------------------------%
% Introduction
The spectral gap, the energy difference between the ground state and the first excited state, is fundamental in quantum many-body physics.
A nonzero gap often implies exponential decay of correlations~\cite{hastings2006spectral} and phase stability~\cite{sachdev1999quantum,bravyi2010topological}, while a vanishing gap is associated with criticality.
It also governs the runtime of adiabatic quantum algorithms~\cite{albash2018adiabatic}.
Many long-standing open questions are about the gap: the Haldane conjecture for antiferromagnetic integer-spin chains~\cite{haldane1983nonlinear}, whether the spin-$\frac{1}{2}$ kagome Heisenberg antiferromagnet is gapped~\cite{zhu2025quantum}, and the Yang--Mills mass gap~\cite{jaffe2006quantum}.

Each of these asks about the spectral gap in the thermodynamic limit, that is, in the limit of infinitely large systems. Determining it is notoriously difficult.
Rigorous methods for bounding spectral gaps mainly apply to special settings, most notably frustration-free systems, with many techniques originating in the study of models such as the AKLT chain~\cite{affleck1988valence,knabe1988energy,nachtergaele1996spectral} and more recently developed into semidefinite relaxations~\cite{rai2026hierarchyofspectral}.
Outside such structured settings, the ground states and low-energy excitations are not known in general.
One often relies on variational numerical approaches, including tensor-network methods~\cite{schollwock2011density,orus2014practical} and variational Monte Carlo~\cite{becca2017quantum}, and extrapolates results from finite systems to the thermodynamic limit.
These approaches have achieved remarkable success, but their conclusions about the thermodynamic limit rest on extrapolation rather than proof. 

There is, moreover, a fundamental computational obstruction.
The spectral gap problem is undecidable \cite{cubitt2022undecidability,bausch2020undecidability}.
More precisely, no general algorithm, even with unlimited resources, can determine whether a system is gapped or gapless, already for translationally invariant nearest-neighbour spin chains.
Even when variational methods accurately reproduce finite-system spectra, the results may give a misleading picture of the thermodynamic limit: a system may appear gapped over a wide range of finite sizes yet become gapless at larger sizes, or vice versa~\cite{cubitt2022undecidability,bausch2020undecidability}.
The spin-$\frac{1}{2}$ kagome Heisenberg antiferromagnet illustrates this difficulty: large-scale DMRG first gave strong evidence for a gapped spin liquid~\cite{yan2011spin}, but later simulations on infinitely long cylinders found signatures of a gapless Dirac spin liquid~\cite{he2017signatures}, and the question remains open~\cite{zhu2025quantum}.
What was missing was not numerical power but a way to turn finite-size evidence into guaranteed statements about the thermodynamic limit.

These undecidability results, however, concern finite-size Hamiltonians with prescribed boundary conditions, and their low-energy spectra can depend on the boundary condition chosen.
Even when the bulk remains gapped, the gap of a finite system can close because of excitations living at its edge, as in the AKLT chain~\cite{affleck1988valence}. 
Hence the finite-size gap is a property of the model together with a choice of boundary, not of the infinite system alone.
To describe the infinite system itself, free of boundary effects, we consider the \emph{bulk spectral gap} $\Delta_{\mathrm{bulk}}$, defined through the energy cost of intrinsic local excitations above its ground states, without prescribing a boundary condition~\cite{BratteliRobinsonVol2}.
Undecidability of the finite-size gap does not by itself establish undecidability of this bulk quantity.
What, then, can be computed about the bulk spectral gap, and what cannot?

\begin{figure}[!t]
    \centering
    \begingroup
    \definecolor{gapblue}{HTML}{0072B2}
    \definecolor{gaporange}{HTML}{D55E00}
    \tikzset{
        gap axis/.style={draw=black!75,line width=0.55pt,-{Stealth[length=1.8mm]}},
        gap small/.style={font=\sffamily\fontsize{8}{9.5}\selectfont}
    }
    \resizebox{\linewidth}{!}{%
    \begin{tikzpicture}[x=1cm,y=1cm,font=\sffamily\fontsize{9}{10.8}\selectfont,line cap=round,line join=round]
            \path[use as bounding box] (-0.12,0.30) rectangle (8.68,5.95);
            \draw[gap axis] (0.95,1.32) -- (8.15,1.32);
            \draw[gap axis] (0.95,1.32) -- (0.95,5.62);
            \node[rotate=90] at (0.18,3.48) {Energy};
            \node at (4.60,0.94) {More computation};
    
            \draw[black!60,densely dashed,line width=0.65pt] (0.95,3.04) -- (8.05,3.04);
            \node[anchor=north west,fill=white,inner sep=2pt] at (1.10,3.01) {$\Delta_{\mathrm{bulk}}$};
            \draw[gapblue,line width=1.35pt]
                (1.18,5.20) -- (1.90,5.20) -- (1.90,4.67) -- (2.72,4.67)
                -- (2.72,4.40) -- (3.47,4.40) -- (3.47,3.93) -- (4.32,3.93)
                -- (4.32,3.76) -- (5.25,3.76) -- (5.25,3.53) -- (6.20,3.53)
                -- (6.20,3.39) -- (7.02,3.39) -- (7.02,3.30) -- (7.85,3.30);
            \foreach \x/\y in {1.18/5.20,1.90/4.67,2.72/4.40,3.47/3.93,4.32/3.76,5.25/3.53,6.20/3.39,7.02/3.30} {
                \fill[gapblue] (\x,\y) circle (1.65pt);
            }
            \node[gapblue,anchor=west] at (3.14,5.23) {Certified upper bounds};
            \node[gapblue] at (6.08,4.58) {$\gamma_{L,d}\downarrow\Delta_{\mathrm{bulk}}$};
    
            \draw[gaporange,dashed,line width=1pt,-{Stealth[length=2mm]}]
                (1.35,1.63) .. controls (2.60,1.70) and (3.15,2.04) .. (4.20,2.36)
                .. controls (5.20,2.65) and (6.45,2.79) .. (7.55,2.86);
            \draw[white,line width=4.7pt] (4.02,2.17) -- (4.38,2.55);
            \draw[white,line width=4.7pt] (4.02,2.55) -- (4.38,2.17);
            \draw[gaporange,line width=1.65pt] (4.02,2.17) -- (4.38,2.55);
            \draw[gaporange,line width=1.65pt] (4.02,2.55) -- (4.38,2.17);
            \node[gaporange,gap small] at (5.10,1.78) {No general algorithm from below};
        \end{tikzpicture}%
    }
    \endgroup
    \caption{\emph{The bulk gap is certifiable from above but uncomputable from below.}
    Our hierarchy of semidefinite programs returns certified upper bounds $\gamma_{L,d}$ that decrease to $\Delta_{\mathrm{bulk}}$ as the relaxation level grows.
    No general algorithm produces convergent lower bounds, and consequently no general upper-bounding algorithm, ours included, admits a computable convergence rate.}
    \label{fig:BulkGapComputability}
    \label{fig:BulkGapComputabilityUpper}
\end{figure}

Here we settle this question: \emph{the bulk spectral gap can be certified from above, and from above only} (\cref{fig:BulkGapComputability}).
We introduce the first general algorithm for rigorously upper bounding the bulk spectral gap, a convergent hierarchy of semidefinite programs whose bounds become arbitrarily tight at the cost of more computation.
As a proof of principle, our algorithm certifies for the first time that no ground state of the spin-$\frac{1}{2}$ kagome Heisenberg antiferromagnet in the thermodynamic limit has a bulk gap above $0.94J$.

The bulk spectral gap is uncomputable from below.
We prove that deciding whether the bulk gap is positive is undecidable.
Hence no general algorithm can produce lower bounds converging to the bulk gap, and no upper-bounding algorithm, ours included, admits a computable convergence rate.
This obstruction already arises in the simplest setting.
We construct one-dimensional chains, translationally invariant and nearest-neighbour with fixed local dimension, whose bulk gap is positive when a universal Turing machine does not halt on a given input, and gapless in every ground state of the infinite system when it halts (\cref{fig:BulkGapComputabilitySpectra}).
Together with our upper-bounding hierarchy, this makes deciding bulk gaplessness $\mathsf{RE}$-complete under the promise that $\Delta_{\mathrm{bulk}}=0$ or $\Delta_{\mathrm{bulk}}\geq 1$.
Our constructions also show that no general algorithm can verify whether an explicitly given state is a ground state of the infinite system, and that certified bounds on ground-state expectation values admit no computable convergence rate.

%---------------------------------%
% Main Body
\paragraph{From finite-size gap to the bulk gap}
For clarity, in the main text we consider a spin chain on $\mathds{Z}$ with single-site Hilbert space $\mathds{C}^q$, an onsite term $h^{(1)}_i$ at each site, and a nearest-neighbour term $h^{(2)}_{i,i+1}$ on each bond, all uniformly bounded.
The system in the thermodynamic limit is then formally described by the Hamiltonian
\begin{align}\label{eq:FormalHamiltonian}
    H=\sum_{i\in\mathds{Z}}h^{(1)}_i+\sum_{i\in\mathds{Z}}h^{(2)}_{i,i+1}.
\end{align}
The framework extends to finite-range interactions on lattices of bounded degree, including the kagome lattice.

The formal sum in \cref{eq:FormalHamiltonian} is not necessarily well-defined, so the spectral gap cannot be read off from it directly.
In the finite-size approach, one instead considers finite chains $\latticeL \coloneqq \{-L,\dots,L\} \subset \mathds{Z}$, where the local Hamiltonian $H^{\latticeL}=\sum_{i\in\latticeL}h^{(1)}_i+\sum_{\{i,i+1\}\subset\latticeL}h^{(2)}_{i,i+1}$ is a hermitian matrix, together with a chosen boundary term acting near the two ends.
Under a prescribed boundary condition, the spectral gap is the difference $E_1(H^{\latticeL})-E_0(H^{\latticeL})$ between its two lowest eigenvalues, and the finite-size approach studies this difference as $L\to\infty$.
The existing undecidability results concern this finite-size formulation~\cite{cubitt2022undecidability,bausch2020undecidability}, whose limit can depend on the boundary term chosen, as in the AKLT chain~\cite{affleck1988valence}.

We therefore define the gap directly in the bulk of the infinite system, taking the thermodynamic limit through local observables rather than through a sequence of finite matrices~\cite{BratteliRobinsonVol2}.
The operators acting at site $i$ form the matrix algebra $\cA_i\simeq M_q(\mathds{C})$, and those supported in $\latticeL$ form $\cA_{\latticeL}=\bigotimes_{i\in\latticeL}\cA_i$.
The union $\cA_{\mathrm{loc}}=\bigcup_L\cA_{\latticeL}$ consists of the local excitations of the infinite chain, and its norm closure $\cA$, the \emph{quasi-local algebra}, adds those that can be approximated by local ones.

On a finite chain, a density operator $\rho$ is equivalently described by the positive, normalized linear functional $\omega(a)=\Tr{\rho a}$ assigning expectation values to observables.
On the infinite chain, this gives a natural, more general definition: a \emph{state} is a linear functional $\omega\colon\cA\to\mathds{C}$ with
\begin{align}\label{eq:StateConditions}
    \omega(\id)=1,\qquad \omega(a^*a)\geq0\quad\text{for all }a\in\cA.
\end{align}

The formal sum $H$ is not a well-defined observable, but the energy cost of a local excitation makes sense.
In particular, the commutator $[H,a]$ is a well-defined operator for each local excitation $a$, since only finitely many $h^{(1)}_i$ and $h^{(2)}_{i,i+1}$ fail to commute with $a$.
The zero-temperature Kubo--Martin--Schwinger (KMS) ground-state condition is
\begin{align}\label{eq:MainTextKMSGroundState}
    \omega(a^*[H,a])\geq0\qquad\text{for all }a\in\cA_{\mathrm{loc}},
\end{align}
and a state satisfying it is a \emph{KMS ground state}~\cite[Chapter~5.3]{BratteliRobinsonVol2}.
The left-hand side corresponds to the energy increase caused by the excitation $a$.
Thus the inequality expresses that no local excitation can lower the energy of the state $\omega$.
On a finite system this coincides with the standard spectral definition.
Moreover, under any boundary condition, every limit of finite-volume ground states is a KMS ground state; see Supplementary Information \cref{rem:RelationLimitGroundAndKMS}.
The class of KMS ground states therefore already contains every state reachable from a finite-size calculation.

For a KMS ground state $\omega$, we say that $\omega$ has a \emph{locally non-degenerate bulk spectral gap of size at least $\gamma>0$} if in addition
\begin{align}\label{eq:MainTextNonDegenSpectralGap}
    \omega(a^*[H,a]) \geq \gamma\big(\omega(a^*a)-|\omega(a)|^2\big)
\end{align}
for all local observables $a \in \cA_{\mathrm{loc}}$~\cite{nachtergaele2024stability}.
Since the variance $\omega(a^*a)-\abs{\omega(a)}^2$ is nonnegative by the Cauchy--Schwarz inequality, \cref{eq:MainTextNonDegenSpectralGap} implies \cref{eq:MainTextKMSGroundState}, and reduces to it at $\gamma=0$.
For every state, this variance is the squared norm of the part of the excitation $a$ orthogonal to $\omega$ itself.
Thus \cref{eq:MainTextNonDegenSpectralGap} says that every local excitation, normalized so that its orthogonal part has unit norm, costs at least $\gamma$ in energy.
The name \emph{locally non-degenerate} records what this excludes: if some local excitation could move $\omega$ to an orthogonal state of zero energy, the left-hand side of \cref{eq:MainTextNonDegenSpectralGap} would vanish while the right-hand side stays positive, so the inequality could hold only with $\gamma=0$.
Global degeneracy is still allowed: distinct symmetry-broken or topological ground states that no local observable connects can each have a locally non-degenerate gap, as the two ordered states of the Ising chain do (Supplementary Information \cref{rem:NonDegeneracyDifferent}).

To obtain upper bounds valid for every KMS ground state, we define the \emph{bulk spectral gap} $\Delta_{\mathrm{bulk}}$ as the supremum of thresholds $\gamma\geq0$ for which at least one such state satisfies \cref{eq:MainTextNonDegenSpectralGap}.
If $\Delta_{\mathrm{bulk}}=0$, it means that every KMS ground state of the infinite system is either locally degenerate or admits local excitations of arbitrarily small positive energy.
We record the strongest case when $\Delta_{\mathrm{bulk}}=0$.
A system is \emph{strongly bulk-gapless} if there is a constant $c>0$ such that every KMS ground state admits local excitations covering a continuous range of excited energies $[0,c]$ (see Supplementary Information \cref{def:GaplessAlgebraic} for precise definition with bulk Hamiltonians).
Our undecidability result concerns this strong bulk-gaplessness.

\paragraph{A general upper-bounding algorithm}
\cref{eq:MainTextNonDegenSpectralGap} is the starting point of our complete semidefinite program (SDP) hierarchy; see Supplementary Information \cref{sec:CompleteSDPHierarchyMainSec} for the full formulation and proof.

Suppose there exists a KMS ground state with a locally non-degenerate bulk spectral gap at least $\gamma$.
Then \cref{eq:MainTextNonDegenSpectralGap} holds for every local excitation, in particular for every excitation supported inside a fixed finite interval $\latticeL$.
So every $a \in \cA_{\latticeL}$ gives a necessary consistency check.

This leads to a hierarchy of tests: for each $L$, we ask whether there exists a consistent assignment of expectation values to operators supported on $\latticeL$ satisfying \cref{eq:MainTextNonDegenSpectralGap} for excitations supported on $\Lambda(L-1)$.
If not for some $L$, then by contraposition, no KMS ground state with a locally non-degenerate bulk gap at least $\gamma$ can exist.
Thus, the hierarchy is a nested family of necessary conditions for the existence of a KMS ground state with a locally non-degenerate bulk gap of at least $\gamma$.

Although each test is formulated on a finite interval $\latticeL$, it is \emph{not} a finite-size calculation.
We do not diagonalize a truncated Hamiltonian on $\latticeL$, nor do we extrapolate properties of a finite system with a prescribed boundary condition.
Instead, we test only excitations $a$ supported on $\Lambda(L-1)$, so that $[H,a]$ is supported in $\latticeL$ and equals $[H^{\latticeL},a]$.
Thus, the finite and infinite dynamics have the same generator on these excitations, and the gap test uses the exact bulk energy cost without selecting a boundary condition.
This is why infeasibility at some finite $\latticeL$ gives a certified upper bound in the thermodynamic limit, rather than a mere finite-size estimation.

For each fixed interval $\latticeL$, this test becomes a finite-dimensional feasibility problem.
Existing works~\cite{araujo2023first,wang2024certifying,fawzi2024certified} already use a similar local-consistency-check idea to, e.g., certify KMS ground states with \cref{eq:MainTextKMSGroundState}.
In those settings the constraints are linear in the state $\omega$, so the problem casts naturally into standard SDP formulations.
Here, by contrast, \cref{eq:MainTextNonDegenSpectralGap} contains a nonlinear term $\abs{\omega(a)}^2$, which is the major obstacle when one seeks to treat it with SDP methods.

To overcome this difficulty, we build on state polynomial optimization~\cite{klep2024state}, further extending it so that \cref{eq:MainTextNonDegenSpectralGap} becomes an SDP constraint without changing the underlying thermodynamic-limit problem.
Since checking \cref{eq:MainTextNonDegenSpectralGap} against all local excitations is not computationally viable, the algorithm is indexed by two parameters: $L$, which controls the size of the local test window, and $d$, which controls the class of local excitations tested within that window; see Supplementary Information \cref{def:HierarchyOfRelaxations,def:GappedRelaxHierarchyStatePoly}.

\begin{theorem*}[Complete upper-bounding SDP hierarchy]
    For each level $(L,d)$, let $\gamma_{L,d}$ be the largest threshold $\gamma$ for which the level-$(L,d)$ relaxation is feasible.
    Then $\gamma_{L,d}$ is a certified upper bound on the bulk spectral gap $\Delta_{\mathrm{bulk}}$ and
    \begin{align}\label{eq:HierarchyConvergence}
        \gamma_{L,d}\searrow\Delta_{\mathrm{bulk}}\qquad(L,d\to\infty).
    \end{align}
\end{theorem*}

In particular, every threshold $\gamma>\Delta_{\mathrm{bulk}}$ is excluded at some finite level, while a threshold at or below $\Delta_{\mathrm{bulk}}$ is never excluded (\cref{fig:BulkGapComputabilityUpper}); see Supplementary Information \cref{thm:CompletenessSDPHierarchy}.
The same hierarchy bounds expectation values of local observables over gapped KMS ground states, and can be restricted to a symmetry sector, yielding a symmetry-restricted bulk gap; see Supplementary Information \cref{rem:SymmetryRestrictedHierarchy}.

\paragraph{Uncomputability from below}
Despite the convergent upper bounds, deciding whether $\Delta_{\mathrm{bulk}}>0$ is undecidable.

\begin{theorem*}[Undecidability of the bulk spectral gap]
    Fix a universal Turing machine $M$.
    There is a computable family of one-dimensional, translationally invariant, nearest-neighbour systems $\widetilde{H}(m)$ with fixed local dimension and uniformly bounded interactions.
    Each admits a KMS ground state, and (\cref{fig:BulkGapComputabilitySpectra}):
    \begin{enumerate}[label=(\roman*)]
        \item if $M(m)$ does not halt, some KMS ground state has a locally non-degenerate bulk spectral gap $\geq 4$;
        \item if $M(m)$ halts, every KMS ground state has bulk spectrum $[0,\infty)$.
    \end{enumerate}
\end{theorem*}

\begin{figure}[!b]
    \centering
    \begingroup
    \definecolor{gapblue}{HTML}{0072B2}
    \definecolor{gaporange}{HTML}{D55E00}
    \tikzset{
        gap axis/.style={draw=black!75,line width=0.55pt,-{Stealth[length=1.8mm]}},
        gap small/.style={font=\sffamily\fontsize{8}{9.5}\selectfont}
    }
    \resizebox{\linewidth}{!}{%
    \begin{tikzpicture}[x=1cm,y=1cm,font=\sffamily\fontsize{9}{10.8}\selectfont,line cap=round,line join=round]
            \path[use as bounding box] (0,0.30) rectangle (8.80,5.95);
            \node at (2.30,5.54) {$M(m)$ does not halt};
            \node at (6.70,5.54) {$M(m)$ halts};
    
            % A locally non-degenerate gapped ground state exists.
            \draw[gap axis] (1.00,1.72) -- (1.00,4.98);
            \node[anchor=east] at (0.84,4.87) {$E$};
            \node[anchor=east,gap small] at (0.90,1.92) {$0$};
            \draw[gapblue,line width=1.5pt] (1.40,1.92) -- (3.20,1.92);
            \draw[gapblue,line width=0.8pt] (1.40,3.15) -- (3.20,3.15);
            \draw[gapblue,line width=0.65pt] (1.40,3.43) -- (3.20,3.43);
            \fill[gapblue!25] (1.40,3.66) rectangle (3.20,4.04);
            \draw[gapblue,line width=0.65pt] (1.40,4.27) -- (3.20,4.27);
            \draw[gapblue,line width=0.65pt] (1.40,4.51) -- (3.20,4.51);
            \node[gapblue] at (2.30,4.99) {$\vdots$};
            \draw[black!70,line width=0.6pt,{Stealth[length=1.4mm]}-{Stealth[length=1.4mm]}]
                (3.45,1.96) -- (3.45,3.11);
            \node[anchor=west,gap small] at (3.54,2.535) {$\geq 4$};
            \node[gapblue,font=\sffamily\bfseries\fontsize{8.5}{10}\selectfont] at (2.30,1.44) {Bulk-gapped};
            \node[gap small] at (2.30,1.02) {$\exists$ ground state};
            \node at (2.30,0.60) {$\Delta_{\mathrm{bulk}}\geq 4$};
    
            % Every ground state has all nonnegative energies in its bulk spectrum.
            \draw[gap axis] (5.40,1.72) -- (5.40,4.98);
            \node[anchor=east] at (5.24,4.87) {$E$};
            \node[anchor=east,gap small] at (5.30,1.92) {$0$};
            \fill[gaporange!30] (5.80,1.92) rectangle (7.60,4.51);
            \draw[gaporange,line width=1.15pt] (5.80,1.92) -- (7.60,1.92);
            \draw[gaporange,line width=0.75pt,-{Stealth[length=1.6mm]}] (6.70,4.66) -- (6.70,5.12);
            \node[gaporange,font=\sffamily\bfseries\fontsize{8.5}{10}\selectfont] at (6.70,1.44) {Strongly bulk-gapless};
            \node[gap small] at (6.70,1.02) {$\forall$ ground states};
            \node at (6.70,0.60) {spectrum $[0,\infty)$};
        \end{tikzpicture}%
    }
    \endgroup
    \caption{\emph{Bulk spectra in the halting reduction.}
    Fix a universal Turing machine $M$.
    For any input $m$, our construction gives a one-dimensional, translationally invariant, nearest-neighbour chain $\widetilde{H}(m)$ with fixed local dimension and uniformly bounded interactions.
    If $M(m)$ does not halt, some ground state of the infinite chain is locally non-degenerate with bulk gap at least $4$.
    If $M(m)$ halts, every ground state of the infinite chain has bulk spectrum $[0,\infty)$, including every limit of finite-volume ground states under arbitrary boundary conditions.}
    \label{fig:BulkGapComputabilitySpectra}
\end{figure}

Together, the two theorems imply:
\begin{corollary*}[Computability and $\mathsf{RE}$-completeness of the bulk spectral gap]
    No general algorithm produces convergent lower bounds on $\Delta_{\mathrm{bulk}}$, and no convergent upper-bounding algorithm, ours included, admits a computable convergence rate.
    The promise problem $\Delta_{\mathrm{bulk}}=0$ (\textnormal{\textsc{yes}}) versus $\Delta_{\mathrm{bulk}}\geq1$ (\textnormal{\textsc{no}}) is $\mathsf{RE}$-complete.
    This holds already for one-dimensional, translationally invariant, nearest-neighbour systems with fixed local dimension and uniformly bounded interactions, even when every KMS ground state of each gapless instance has bulk spectrum $[0,\infty)$.
\end{corollary*}

The same construction also settles a related question.
Deciding whether a fixed, explicitly given product state is a KMS ground state of the infinite chain is $\mathsf{coRE}$-complete: when the state fails to be one, a finite computation always detects this, but when it is one, no computation can ever confirm it.
In particular, verifying \cref{eq:MainTextKMSGroundState} for every excitation in a finite region never certifies that the state is a KMS ground state, and no computation determines in advance how large a region would be enough.
Convergent certified bounds on ground-state expectation values admit no computable convergence rate either, already for a one-site projection; this applies to our hierarchy at $\gamma = 0$ and to those of~\cite{araujo2023first,fawzi2024certified}.
Refer to Supplementary Information \cref{thm:BulkGapUndecidability,cor:BulkGapNoConvergenceRate,cor:KMSGroundStateComputability} for details.

\begin{figure*}[!htbp]
    \centering
    \begingroup
    \definecolor{sectorA}{HTML}{0072B2}
    \definecolor{sectorB}{HTML}{D55E00}
    \definecolor{bufferC}{HTML}{009E73}
    \captionsetup[subfigure]{font=small,labelfont=bf,justification=centering,singlelinecheck=false,skip=2pt}
    \tikzset{
        replace site/.style={draw=black!40,line width=0.5pt,rounded corners=1.5pt,minimum width=0.66cm,minimum height=0.62cm,inner sep=0pt},
        replace A/.style={replace site,draw=sectorA!80,fill=sectorA!10,text=sectorA!80!black},
        replace B/.style={replace site,draw=sectorB!80,fill=sectorB!13,text=sectorB!80!black},
        replace C/.style={replace site,draw=bufferC!80!black,fill=bufferC!12,text=bufferC!70!black},
        replace bond/.style={draw=black!35,line width=0.65pt},
        replace cost/.style={font=\sffamily\fontsize{9}{10.5}\selectfont,inner sep=1pt},
        replace arrow/.style={draw=black!70,line width=0.8pt,-{Stealth[length=2mm]}},
        replace leader/.style={draw=black!50,line width=0.45pt}
    }
    
    \begin{subfigure}{\linewidth}
        \centering
        \resizebox{\linewidth}{!}{%
        \begin{tikzpicture}[x=1cm,y=1cm,font=\sffamily\fontsize{10}{12}\selectfont,line cap=round]
            \path[use as bounding box] (-0.25,-0.86) rectangle (17.55,1.02);
            \draw[replace bond] (0.40,0) -- (5.10,0);
            \draw[replace bond] (7.40,0) -- (12.10,0);
            \foreach \x in {0.75,1.75,2.75,3.75,4.75} {\node[replace A] at (\x,0) {$A$};}
            \foreach \x in {7.75,11.75} {\node[replace A] at (\x,0) {$A$};}
            \foreach \x in {8.75,9.75,10.75} {\node[replace B] at (\x,0) {$B$};}
            \foreach \x in {-0.05,5.50,6.95,12.50} {\node at (\x,0) {$\cdots$};}
            \draw[replace arrow] (5.95,0) -- (6.55,0);
            \foreach \x in {8.25,11.25} {
                \draw[sectorB,line width=1.2pt] (\x-0.16,0) -- (\x+0.16,0);
                \draw[replace leader] (\x,0.13) -- (\x,0.47);
                \node[replace cost] at (\x,0.68) {$+2$};
            }
            \draw[replace leader] (8.42,-0.43) -- (8.42,-0.53) -- (11.08,-0.53) -- (11.08,-0.43);
            \node[replace cost] at (9.75,-0.73) {internal energy $\geq 0$};
            \node[anchor=west] at (13.25,0) {$\Delta E\geq 2+2=4$};
        \end{tikzpicture}%
        }
        \caption{Nonhalting: excite into $B$ interval}
        \label{fig:BulkLocalReplacementNonhalting}
    \end{subfigure}
    
    \vspace{3mm}
    \begin{subfigure}{\linewidth}
        \centering
        \resizebox{\linewidth}{!}{%
        \begin{tikzpicture}[x=1cm,y=1cm,font=\sffamily\fontsize{10}{12}\selectfont,line cap=round]
            \path[use as bounding box] (-0.25,-0.86) rectangle (17.55,1.02);
            \draw[replace bond] (0.40,0) -- (5.10,0);
            \draw[replace bond] (7.40,0) -- (12.10,0);
            \foreach \x in {0.75,4.75,7.75,11.75} {\node[replace B] at (\x,0) {$B$};}
            \foreach \x in {1.75,2.75,3.75} {\node[replace A] at (\x,0) {$A$};}
            \foreach \x in {8.75,9.75,10.75} {\node[replace C] at (\x,0) {$\chi$};}
            \foreach \x in {-0.05,5.50,6.95,12.50} {\node at (\x,0) {$\cdots$};}
            \draw[replace arrow] (5.95,0) -- (6.55,0);
            \foreach \x in {1.25,4.25} {
                \draw[sectorB,line width=1.2pt] (\x-0.16,0) -- (\x+0.16,0);
                \draw[replace leader] (\x,0.13) -- (\x,0.47);
                \node[replace cost] at (\x,0.68) {$2$};
            }
            \foreach \x in {8.25,11.25} {
                \draw[bufferC!80!black,line width=1.2pt] (\x-0.16,0) -- (\x+0.16,0);
                \draw[replace leader] (\x,0.13) -- (\x,0.47);
                \node[replace cost] at (\x,0.68) {$\leq\frac12$};
            }
            \draw[replace leader] (8.42,-0.43) -- (8.42,-0.53) -- (11.08,-0.53) -- (11.08,-0.43);
            \node[replace cost] at (9.75,-0.73) {internal energy $=0$};
            \node[anchor=west] at (13.25,0) {$\Delta E\leq-4+1=-3$};
        \end{tikzpicture}%
        }
        \caption{Halting: replacing a short $A$ interval}
        \label{fig:BulkLocalReplacementShort}
    \end{subfigure}
    
    \vspace{3mm}
    \begin{subfigure}{\linewidth}
        \centering
        \resizebox{\linewidth}{!}{%
        \begin{tikzpicture}[x=1cm,y=1cm,font=\sffamily\fontsize{10}{12}\selectfont,line cap=round]
            \path[use as bounding box] (-0.25,-0.86) rectangle (17.55,1.02);
            \draw[replace bond] (0.05,0) -- (5.45,0);
            \draw[replace bond] (7.05,0) -- (12.45,0);
            \node[replace A,minimum width=4.66cm] at (2.75,0) {$|A\rangle^{\otimes(N+2)}$};
            \foreach \x in {7.75,11.75} {\node[replace C] at (\x,0) {$\chi$};}
            \node[replace B,minimum width=2.66cm] at (9.75,0) {$|\psi_N\rangle\otimes|0\rangle^{\otimes N}$};
            \foreach \x in {-0.05,5.50,6.95,12.50} {\node[fill=white,inner sep=1pt] at (\x,0) {$\cdots$};}
            \draw[replace arrow] (5.95,0) -- (6.55,0);
            \foreach \x in {7.25,12.25} {
                \draw[bufferC!80!black,line width=1.2pt] (\x-0.15,0) -- (\x+0.15,0);
                \draw[replace leader] (\x,0.13) -- (\x,0.47);
                \node[replace cost] at (\x,0.68) {$\leq 2$};
            }
            \draw[replace leader] (8.42,-0.43) -- (8.42,-0.53) -- (11.08,-0.53) -- (11.08,-0.43);
            \node[replace cost] at (9.75,-0.73) {$E_N<-4$};
            \node[anchor=west] at (13.25,0) {$\Delta E\leq E_N+4<0$};
        \end{tikzpicture}%
        }
        \caption{Halting: replacing a long $A$ interval}
        \label{fig:BulkLocalReplacementLong}
    \end{subfigure}\par
    \endgroup
        \caption{\emph{Energy cost of local excitations for bulk undecidability.}
    $A$ and $B$ indicate the local subspaces $\cH_A$ and $\cH_B$, $|\chi\rangle\coloneqq|\perp\rangle\otimes|0\rangle$, and $\Delta E$ is the energy change.
    Numbers above bonds indicate their energy contributions, or upper bounds on their change.
    (a) Replacing $A$ sites by $B$ components in a finite interval creates at least two $A$--$B$ bonds; the computational and XY energies within the $B$ regions are nonnegative.
    In (b,c), choose $N$ so that $K_N(m)$ has a normalized ground-state vector $\ket{\psi_N}$ with energy $E_N<-4$.
    (b) Replacing a short $A$ interval bounded by $B$ sites removes two $A$--$B$ penalties while adding at most $1$ in XY energy.
    (c) The inserted middle block has energy $E_N<-4$, and both bonds joining it to the $\ket{\chi}$ sites have zero energy.
    The two outer bonds contribute at most $4$, since $\ket{\chi}$ switches off their computational terms, regardless of the exterior.
    These local replacements correspond to local excitations in $\cA_{\mathrm{loc}}$, which we substitute into \cref{eq:MainTextKMSGroundState,eq:MainTextNonDegenSpectralGap}.
    }
    \label{fig:BulkLocalReplacements}
\end{figure*}

To introduce the main idea, we start from the chains of~\cite[Theorem~22 and Lemma~23]{bausch2020undecidability}.
Their onsite and nearest-neighbour matrices $h^{(1)}(m)$ and $h^{(2)}(m)$ act on a fixed space $\mathds{C}^r$ with a fixed $r < \infty$.
Define
\begin{align}\label{eq:MainUndecidabilityComputationalChain}
    K_N(m)\coloneqq\sum_{i=1}^{N}h_i^{(1)}(m)+\sum_{i=1}^{N-1}h_{i,i+1}^{(2)}(m).
\end{align}
If $M(m)$ does not halt then $K_N(m)\geq\id$ for every $N$, and if it halts then $E_0(K_N(m))\to-\infty$ as $N \to \infty$.
Note that this is a statement about finite-size systems under open boundary conditions, which suffices for the undecidability result of~\cite{bausch2020undecidability}.
For a bulk statement it is not enough.
We must control every KMS ground state of the infinite chain, including every limit of finite-volume ground states under an arbitrary boundary condition.

To this end, we enlarge the single-site space to
\begin{align}\label{eq:MainUndecidabilityLocalSpace}
    \widetilde{\cH}= \underbrace{\mathds{C}\ket{A}}_{\eqqcolon\,\cHA} \oplus \underbrace{\bigl(\mathds{C}^r\oplus\mathds{C}\ket{\perp}\bigr)\otimes\mathds{C}^{2}}_{\eqqcolon\,\cHB},
\end{align}
and write $\widetilde{H}(m)$ for the chain with the following onsite and nearest-neighbour terms.
\begin{enumerate}[label=(\arabic*)]
    \item The \emph{computational} terms are $h^{(1)}_i(m)$ and $h^{(2)}_{i,i+1}(m)$ of \cref{eq:MainUndecidabilityComputationalChain}, acting on the factors $\mathds{C}^r$, as the identity on the qubits $\mathds{C}^2$, and as zero on $\mathds{C}\ket{\perp}$ and on $\cH_A$.
    \item The \emph{$A$--$B$ penalty} term is $2\,(P^A_{i}\otimes P^B_{i+1}+P^B_{i}\otimes P^A_{i+1})$, where $P^A_i$ and $P^B_i$ project onto $\cHA$ and $\cHB$ at site $i$.
    It charges energy $2$ to every pair of neighbouring sites with one site in $\cHA$ and the other in $\cHB$, which we call an \emph{$A$--$B$ bond}.
    \item The \emph{XY interaction} acts on the qubits $\mathds{C}^2$, as the identity on $\mathds{C}^r$ and $\mathds{C}\ket{\perp}$, and as zero when either site lies in $\cHA$.
    It is the two-site projection $\tfrac12\id_4-\tfrac14(X\otimes X+Y\otimes Y+Z\otimes\id_2+\id_2\otimes Z)$, where $X,Y,Z$ are the Pauli matrices, and generates the dynamics of the XY chain in a transverse field~\cite{araki1985ground}, whose excitations propagate along the chain.
\end{enumerate}
The local dimension is $2r+3$, independent of $m$.
The computational terms have algebraic entries computable from $m$, while the $A$--$B$ penalty and XY terms are fixed rational matrices, so the map $m\mapsto\widetilde{H}(m)$ is computable.
Explicit definitions are given in Supplementary Information \cref{sec:BulkGapUndecidability}.
We write $\ket{\chi}\coloneqq\ket{\perp}\otimes\ket{0}\in\cHB$.
The computational terms vanish on $\ket{\chi}$, and the XY term vanishes on $\ket{\chi}\otimes\ket{\chi}$, but not in general where a $\chi$ site meets another site of $\cHB$.

We state the physical intuition of the proof.
Suppose first that $M(m)$ does not halt, and consider the product state $\omega_A$ with $\ket{A}$ at every site.
Any local excitation orthogonal to it places some sites in $\cHB$ while distant sites remain in $\cHA$, so the configuration acquires at least two $A$--$B$ bonds and the $A$--$B$ penalty term contributes at least $4$ (\cref{fig:BulkLocalReplacementNonhalting}).
Nothing inside the new $B$ intervals can lower the energy.
Indeed, the XY interaction is a projection, hence nonnegative.
The nonhalting condition $K_N(m)\geq\id$ makes the computational contribution nonnegative as well.
Hence $\omega_A$ satisfies \cref{eq:MainTextNonDegenSpectralGap} with $\gamma=4$: it is a KMS ground state with a locally non-degenerate bulk gap of at least $4$.

Next, suppose that $M(m)$ halts.
We claim that no KMS ground state has a component in $\cHA$ at any site.
This is what makes the halting case gapless.
A state with no component in $\cHA$ restricts to a KMS ground state of the qubit algebra under the dynamics generated by the XY interaction.
We prove that this XY dynamics has a unique KMS ground state, whose propagating excitations fill the spectrum $[0,\infty)$ (Supplementary Information \cref{lem:BulkReductionQubitSpectrum}).
It follows that every KMS ground state inherits the bulk spectrum $[0,\infty)$, which is statement~(ii).
This includes every limit of finite-volume ground states, under any boundary condition.

To prove the claim, fix $N$ with $E_N\coloneqq E_0(K_N(m))<-4$.
Wherever a component in $\cHA$ is present, a local excitation replaces it by a vector of lower energy, which \cref{eq:MainTextKMSGroundState} forbids.
If fewer than $N+2$ consecutive sites are in $\ket{A}$ and lie between two sites of $\cHB$, replace each of them by $\ket{\chi}$.
The two $A$--$B$ bonds are removed, lowering the energy by $4$.
The computational terms act trivially on the whole interval, because of the $\ket{\perp}$ factor of $\ket{\chi}$.
The XY interaction vanishes inside the interval since we have consecutive products of $\ket{\chi}$, and contributes at most $1$ at the two end points, so the energy drops by at least $3$ (\cref{fig:BulkLocalReplacementShort}).

Otherwise, at least $N+2$ consecutive sites are in $\ket{A}$, possibly infinitely many.
We then replace $N+2$ of them by $\ket{\chi}\otimes\bigl(\ket{\psi_N}\otimes\ket{0}^{\otimes N}\bigr)\otimes\ket{\chi}$, where $\ket{\psi_N}$ is a ground state of $K_N(m)$.
The inserted state contributes $E_N < - 4$, while each of its two ends changes the energy by at most $2$, whatever lies outside.
Again, $\ket{\chi}$ ensures that the computational terms vanish across the endpoints.
So the energy again decreases (\cref{fig:BulkLocalReplacementLong}).
Formally, these replacements are implemented by local excitation operators, and substituting them into \cref{eq:MainTextKMSGroundState} gives $\omega(P^A_i)=0$ at every site $i$, for every KMS ground state $\omega$.
This proves the claim.

\paragraph{Certified bounds in practice}
We now evaluate the hierarchy at finite levels $(L,d)$ to obtain certified upper bounds $\gamma_{L,d}$.
As a benchmark, we first apply it to the transverse-field Ising chain with spin-flip symmetry imposed, obtaining bounds consistent with the known analytical behaviour; see Supplementary Information \cref{sec:TransverseFieldIsing}.

We then apply the method to the spin-$\frac{1}{2}$ kagome Heisenberg antiferromagnet, a canonical frustrated model for quantum spin liquids.
Whether it is gapped in the thermodynamic limit remains debated, with numerical studies supporting both gapped~\mbox{\cite{yan2011spin,depenbrock2012nature}} and gapless~\mbox{\cite{he2017signatures,liao2017gapless,jiang2019competing}} ground states.

\begin{figure}[!htbp]
    \centering
    \includegraphics[width=0.86\linewidth]{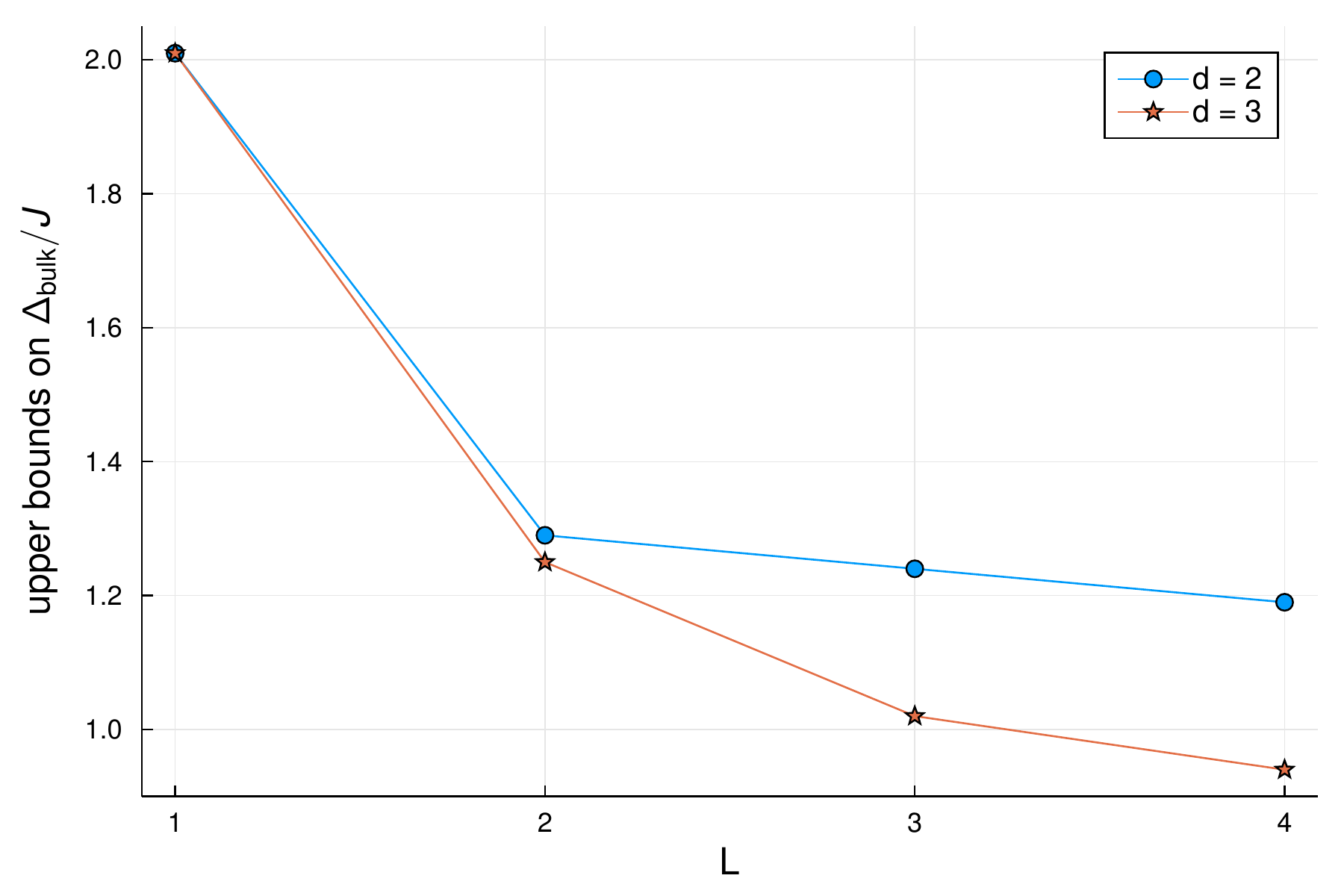}
    \caption{
    \emph{Certified upper bounds for the spin-$\frac{1}{2}$ kagome lattice Heisenberg antiferromagnet.}
    Upper bounds on $\Delta_{\mathrm{bulk}}/J$ versus patch radius $L$ at monomial degrees $d=2,3$.
    The patches contain $5$, $13$, $27$, and $45$ sites for $L=1,2,3,4$, respectively.
    See Supplementary Information \cref{tab:KagomeNumericalBounds,rem:KagomeNecessarySpinSymmetry} for exact values and details.
    }
    \label{fig:MainTextKagomeResult}
\end{figure}

Our hierarchy certifies $\Delta_{\mathrm{bulk}}\leq 0.94J$ at $(L,d)=(4,3)$ (\cref{fig:MainTextKagomeResult}).
This improves the elementary analytical bound $4J/3$ by approximately $30\%$, and the bound $1.214J$ obtained from the reported energy lower bound of~\cite{misguich2005specific} by approximately $23\%$.
Both comparison bounds follow from the single-site variational estimate $\Delta_{\mathrm{bulk}}\leq-8e_0/3$, where $e_0$ is the ground-state energy per site.
We also demonstrate our algorithm by bounding scalar chirality for gapped KMS ground states.
See Supplementary Information \cref{sec:KagomeHeisenberg} for the derivations and numerical details.

\paragraph{Discussion}
Taken together, our two results settle the algorithmic status of the spectral gap in the bulk of an infinite system: the bulk-gaplessness promise problem is $\mathsf{RE}$-complete, with further consequences for what can be computed about KMS ground states and their expectation values.

The same construction has further consequences; see Supplementary Information \cref{cor:KMSGroundStateComputability}.
Deciding whether an explicitly given product state is a KMS ground state is $\mathsf{coRE}$-complete: a failure shows up at some finite level, but no finite computation can confirm success.
In particular, the state can satisfy the bulk-gap inequality for every excitation in a region of any size computed in advance and still fail to be a ground state of the infinite chain.
Convergent certified bounds on ground-state expectation values, including those of~\cite{araujo2023first,fawzi2024certified}, admit no computable convergence rate, already for the largest expectation of a one-site projection.

The two known finite-size reductions run in opposite directions: a halting input gives a gapped family in~\cite{cubitt2022undecidability}, and a gapless one in~\cite{bausch2020undecidability}.
The first rules out general convergent upper bounds on the open-boundary gap, the second rules out general convergent lower bounds.
With open boundary conditions, the finite-size gap is therefore uncomputable from both sides in general.
Our hierarchy nevertheless bounds the bulk gap from above, because it is a different quantity: it concerns all KMS ground states, which include the limits of finite-volume ground states under every boundary condition.
Supplementary Information \cref{sec:DiscussionSemiDecidability} makes this relation precise and applies our semidecidability procedure to the construction of~\cite{cubitt2022undecidability}.
There, \cref{prop:CPW15GappedInDiffBoundaryCondition} shows that a family gapless under open boundary conditions can still have $\Delta_{\mathrm{bulk}} > 0$.
Uncomputability from above disappears when we pass from the open-boundary gap to the bulk gap, whereas uncomputability from below persists.
This shows that the latter is intrinsic to the bulk interactions rather than potential boundary effects.

Our construction is, however, sensitive to perturbations.
Arbitrarily small onsite perturbations change the undecidable behaviour of both the finite-size construction~\cite{castillacastellano2026instability} and our bulk construction (Supplementary Information \cref{rem:BulkReductionOnsiteInstability}).
Whether undecidability can be realized by a family that is stable under such perturbations is open.

Our current bounds on the kagome model are proof-of-principle rather than an optimized kagome-specific computation.
Stronger structural reductions, better exploitation of symmetries such as those in~\cite{wang2026scalable}, and increased computational resources may substantially sharpen them.
This would turn the method into a competitive tool for difficult models where traditional approaches remain inconclusive.
Our bounds are also computed with an SDP solver in floating-point arithmetic, and we do not account for numerical error.
Extending the rigorous post-processing of~\cite{naceur2025certifiedboundsoptimizationproblems} to our optimization problem would give exact rational bounds at a reasonable computational cost.

Improved upper bounds also carry qualitative information: once a certified bound lies below the energy of a proposed excitation, that excitation is rigorously excluded from being the lowest bulk excitation.
Thus, even when a certified upper bound is numerically larger than a variational estimate, it says something rigorous about the thermodynamic limit.
More broadly, semidefinite programs have proved powerful for certified results in quantum information theory~\cite{navascues2008convergent,pironio2010convergent} and, more recently, in many-body physics~\cite{araujo2023first,wang2024certifying,fawzi2024certified,mortimer2025certifying,araujo2025differential,rai2026hierarchyofspectral,rao2026convergenthierarchyspectralgap}.
Our extension of state polynomial optimization also suggests a route to SDP certificates for other thermodynamic-limit properties defined by nonlinear state constraints.

%---------------------------------%
\subsection*{Code availability}
The code used to generate the numerical results and figures in this work is available at \url{https://github.com/wangjie212/SpectralGap}.

%---------------------------------%
% Acknowledgments
\subsection*{Acknowledgments}
We thank Antonio Ac\'in, Ignacio Cirac, Toby Cubitt, Hamza Fawzi, Tommaso Guaita, Lauritz van Luijk, Gr\'egoire Misguich, Miguel Navascu\'es, James Purcell, Robert Salzmann, Lucas Tendick, Thomas Vidick, and James Watson for the helpful discussions.
This work was performed within the project COMPUTE, which is funded within the QuantERA II Programme that has received funding from the EU's H2020 research and innovation programme under the GA No.~101017733 {\normalsize\euflag}. XX and MOR acknowledge funding by the ANR for the JCJC grants LINKS (No. ANR-23-CE47-0003) and the T-ERC QNET (No. ANR-24-ERCS-0008).
JW was also supported by the National Key R\&D Program of China under grant No.~2023YFA1009401, the National Natural Science Foundation of China (NSFC) under grant No.~12571333.
VM was also supported by the HORIZON–MSCA-2023-DN-JD of the European Commission under the Grant Agreement No.~101120296 (TENORS), the AI Interdisciplinary Institute ANITI funding through the French ``France 2030'' program under the Grant agreement No.~ANR-23-IACL-0002, as well as the National Research Foundation, Prime Minister’s Office, Singapore under its Campus for Research Excellence and Technological Enterprise (CREATE) programme. IK also acknowledges support of the Slovenian Research Agency program P1-0222 and grants J1-50002, N1-0217, J1-60011, J1-50001, J1-3004 and J1-60025. 
Partially supported by the Fondation de l’\'Ecole polytechnique
as part of the Gaspard Monge Visiting Professor Program, IK and MS thank \'Ecole polytechnique and Inria Paris Saclay for hospitality during the preparation of this manuscript. OF acknowledges funding by the European Research Council (ERC Grant AlgoQIP, Agreement No. 851716).

%---------------------------------%
% Bibliography
\begin{sloppypar}
\printbibliography
\end{sloppypar}

%---------------------------------%
% Supplementary information premeable
\clearpage
\newgeometry{margin=1in}
\onecolumn
\appendix

% \leftlinenumbers*   % giving line numbers; stop switching; put appendix numbers on the left

% SI heading
\begin{center}
    {\LARGE\bfseries Supplementary Information\par}
    \vspace{0.8em}
\end{center}

\setcounter{section}{0}
\setcounter{subsection}{0}
\setcounter{equation}{0}
\setcounter{figure}{0}
\setcounter{table}{0}

\renewcommand{\thesection}{\arabic{section}}
\renewcommand{\thesubsection}{\thesection.\arabic{subsection}}
\renewcommand{\theequation}{S\arabic{equation}}
\renewcommand{\thefigure}{S\arabic{figure}}
\renewcommand{\thetable}{S\arabic{table}}
%---------------------------------%
% Supplemntary Information start here
We give an outline of this supplementary material, emphasizing where the main definitions and technical results enter.
In \cref{sec:AlgebraicFormalismMainSec} we introduce the operator-algebraic framework for many-body systems in the thermodynamic limit, including dynamical systems, KMS ground states, and bulk spectral gaps.
In \cref{sec:CompleteSDPHierarchyMainSec} we construct our algorithm via an SDP hierarchy and prove its completeness, yielding convergent upper bounds on the bulk spectral gap and a semidecision procedure for excluding proposed gap thresholds.
In \cref{sec:BulkGapUndecidability} we prove that deciding whether the bulk gap is positive is undecidable and that no general algorithm produces convergent lower bounds, even for one-dimensional translationally invariant nearest-neighbour systems with fixed local dimension and uniformly bounded interactions, promised either to admit a KMS ground state with locally non-degenerate bulk gap at least $4$ or to be strongly bulk-gapless.
A direct consequence is that no general algorithm producing convergent upper bounds, including ours, admits a computable convergence rate.
Together with our hierarchy, these results also establish $\mathsf{RE}$-completeness of the bulk-gaplessness promise problem (\cref{cor:BulkGapNoConvergenceRate}).
The same construction further shows that deciding whether an explicitly given product state is a KMS ground state is $\mathsf{coRE}$-complete, with no computable bound on how large a finite region must be checked, and that certified bounds on ground-state expectation values admit no computable convergence rate (\cref{cor:KMSGroundStateComputability}).
In \cref{sec:NumericalShowcase} we apply the hierarchy to the transverse-field Ising model and the kagome lattice Heisenberg antiferromagnet as a proof of principle.
Finally, in \cref{sec:DiscussionSemiDecidability} we compare bulk and finite-size gap notions and discuss the role of boundary conditions in the undecidability construction of~\cite{cubitt2022undecidability}.

\tableofcontents

\section{Operator algebraic formalism for spectral gaps in the thermodynamic limit}\label{sec:AlgebraicFormalismMainSec}
The theory of operator algebras provides a natural framework for studying bulk behaviour directly in the thermodynamic limit~\cite{BratteliRobinsonVol2}; see also~\cite{young2023quantum} for a high-level introduction.
In this section we recall the minimal formalism needed for our construction.
After defining the quasi-local algebra and its dynamics (\cref{sec:AlgebraicDynamicalSystem}), we introduce KMS ground states and their GNS bulk Hamiltonians (\cref{sec:AlgebraicGroundState}).
We then define the state-dependent locally non-degenerate bulk gap, give its equivalent dynamical criterion, identify its relationship with purity/extremity, and define the associated system-level bulk spectral gap (\cref{sec:AlgebraicSpectralGap}).

\subsection{Thermodynamic limit and dynamical systems with operator algebras}\label{sec:AlgebraicDynamicalSystem}
We begin by introducing an algebraic formulation for thermodynamic limits, following~\cite[Chapter~6.2]{BratteliRobinsonVol2}.

Let us denote by $M_q(\mathds{C})$ the set of $q \times q$ matrices with complex entries. 
Let $\latticeL \coloneqq \{-L,\dots, L\}^D$ be the set of vertices of a $D$-dimensional lattice of side length $2L+1 \in \mathds{N}$. 
To each site $i \in \latticeL$ we associate a Hilbert space $\cH^{(i)} \simeq \mathds{C}^q$ and let $\cA_i = B(\cH^{(i)})\simeq M_q(\mathds{C})$ be the $C^*$-algebra of (bounded) observables acting on $\cH^{(i)}$. 
Thanks to finite dimensionality, we may assume that $\cA_i$ is generated by a finite set of self-adjoint operators $\{x_j^{(i)}\} \subset \cA_i$ and is bounded such that $M^{(i)}_j - (x^{(i)}_j)^2 \succeq 0$ for some fixed constant $M^{(i)}_j > 0$.
Similarly, for any finite subset $S \subset \latticeL$, let $\cA_S = \bigotimes_{i \in S} \cA_i$ with the generating set $\{x^{(i)}_j\}_{i \in S}$.
Then, the thermodynamic limit $\bigcup_{L}^{\infty} \latticeL$ can be modeled by an approximately finite-dimensional $C^*$-algebra of quasi-local observables
\begin{align} \label{eq:QuasiLocalCstarAlgebra}
    \cA \coloneqq \overline{\cA_{\mathrm{loc}}}, \quad \text{where }\cA_{\mathrm{loc}} \coloneqq \bigcup_{L \in \mathds{N}} \cA_{\latticeL},
\end{align}
and the closure is taken with respect to the operator norm.

With the observables of the system in the thermodynamic limit defined, we now introduce the dynamics.
For any finite subset $S \subset\Lambda$, the interaction $\Phi$ among the sites of $S$ is given by a self-adjoint matrix $\Phi(S) = \Phi(S)^* \in \cA_S$.
Fixing an interaction $\Phi$, the local Hamiltonian on the sublattice $\latticeL$ is given by
\begin{align}\label{eq:HamtiltonianGeneralDef}
    H^{\latticeL} = \sum_{S \subset \latticeL} \Phi(S) \in \cA_{\latticeL},
\end{align}
which is well-defined as a finite sum.
In our manuscript, we consider only \emph{finite-range interactions}: there exists an $l \geq 0$ such that $\Phi(S) \neq 0$ only if $\mathrm{diam}(S) \leq l$, where $\mathrm{diam}(S) \coloneqq \sup_{x,y\in S} \mathrm{dist}(x,y)$ denotes the diameter of $S$ in the lattice metric.

\begin{example}
    For example, with a translationally invariant nearest neighbour interaction, the local Hamiltonian of the transverse-field Ising model on the sublattice $\{-L, \dots, L\} \subset \mathds{Z}$ at some constant $g \geq 0$~\cite{mbeng2024quantum} is given by
    \begin{align}
        H_{\mathrm{TFIM}, g}^{\Lambda(L)} = - \sum_{i = -L}^{L-1} Z_{i} Z_{i+1} + g \sum_{i = -L}^{L} X_i.
    \end{align}
    where $Z_i, X_i \in \cA_i$ are the Pauli matrices acting on the site $i$.
\end{example}

While the thermodynamic limit of the observables of $\Lambda$ is straightforward, for the family $\familyH = \{H^{\latticeL}\}_L$, taking the norm limit $L\to\infty$ inside $\cA$ is, however, \emph{not} meaningful due to unboundedness.
Instead, we consider the \emph{Heisenberg dynamics} generated by the family $\familyH = \{H^{\latticeL}\}_L$.
For every local observable $a \in \cA_{\mathrm{loc}}$, the limit
\begin{align}\label{eq:DynamicFromHamiltonian}
\tau_{\familyH}^t(a) \coloneqq \lim_{L\to\infty} \mathrm{e}^{\mathrm{i} t H^{\latticeL}} a \mathrm{e}^{-\mathrm{i} t H^{\latticeL}},\quad t\in\mathbb R,
\end{align}
exists and defines a strongly-continuous one-parameter automorphism group on $\cA$ (thanks to the finite range interaction).
Furthermore, let $\delta_{\familyH}$ be the densely defined \emph{derivation} (or \emph{generator}) of the dynamics $\tau_{\familyH}$, which is a closed derivation whose dense domain contains $\cA_{\mathrm{loc}}$ as a core.
For any local observable $a \in \cA_{\latticeL}$ with sufficiently large $L$, it satisfies
\begin{align}\label{eq:DerivationFromHamiltonian}
\delta_{\familyH}(a) = \mathrm{i} [H^{\latticeL}, a].
\end{align}

\begin{definition}\label{def:DynamicalSystem}
    Given the family of the local Hamiltonians $\familyH = \{H^{\latticeL}\}_L$ on the lattice $\Lambda = \bigcup_{L=1}^{\infty} \latticeL$, let $\tau_{\familyH}$ be the associated Heisenberg dynamics and $\delta_{\familyH}$ be its derivation.
    We call the tuple $(\cA, \tau_{\familyH}, \delta_{\familyH})$ a \emph{dynamical system induced by $\familyH$}.
\end{definition}

Crucially, with the operator algebraic formalism of the thermodynamic limit of many-body systems, we directly access the bulk properties of the systems without the need to fix any boundary conditions on the finite-size approximants.

\subsection{KMS ground states}\label{sec:AlgebraicGroundState}
We are ready to introduce a dynamical formulation of ground states, the Kubo–Martin–Schwinger (KMS) state~\cite[Chapter~5.3]{BratteliRobinsonVol2} at temperature $0$.

\begin{definition}[KMS ground state]\label{def:KMSGroundState}
Let $(\cA, \tau_{\familyH}, \delta_{\familyH})$ be a dynamical system induced by the local Hamiltonian family $\familyH = \{H^{\latticeL}\}_L$ and let $\omega: \cA \to \mathds{C}$ be a state.
Then, $\omega$ is a \emph{KMS ground state} for the family $\familyH$ if and only if
\begin{align}\label{eq:KMSGroundCondition}
- \mathrm{i} \omega(a^* \delta_{\familyH}(a)) \geq 0
\end{align}
for all $a$ in the domain of $\delta_{\familyH}$.
In particular, for $a \in \cA_{\mathrm{loc}}$, \cref{eq:KMSGroundCondition} simplifies to $\omega(a^*[H^{\latticeL},a])\geq 0$ for all sufficiently large $L$.
It also follows that $\omega$ is $\tau_{\familyH}$-invariant, i.e., $\omega(\tau_{\familyH}^t(a)) = \omega(a)$ for all $a \in \cA$ and $t \in \mathds{R}$.
\end{definition}

Note that \cref{eq:KMSGroundCondition} implicitly requires $-\mathrm{i}\omega(a^*\delta_{\familyH}(a))\in\mathds{R}$ for every $a$, so $\omega(a^*\delta_{\familyH}(a))$ is purely imaginary. 
Hence
\begin{equation}
    \begin{aligned}
        0
        &=
        \omega(a^*\delta_{\familyH}(a))
        +
        \overline{\omega(a^*\delta_{\familyH}(a))} \\
        &=
        \omega(a^*\delta_{\familyH}(a))
        +
        \omega(\delta_{\familyH}(a^*)a)
        =
        \omega(\delta_{\familyH}(a^*a)).
    \end{aligned}
\end{equation}
Since every element of a unital $*$-algebra is a linear combination of squares, it follows that every KMS ground state satisfies the stationary condition
\begin{align}\label{eq:StationaryCondition}
    \omega(\delta_{\familyH}(a))=0
\end{align}
for all $a$ in the domain of $\delta_{\familyH}$.
Equivalently, $\omega$ is stationary under the dynamics, $\omega\circ\tau_{\familyH}^t=\omega$, in agreement with~\cite[Lemma~5.3.16]{BratteliRobinsonVol2}.
Moreover,
\begin{align}\label{eq:KMSGroundLHSHermitian}
    -\mathrm{i}\omega(a^*\delta_{\familyH}(a))
    =
    -\frac{\mathrm{i}}{2}
    \omega\!\bigl(
        a^*\delta_{\familyH}(a)-\delta_{\familyH}(a^*)a
    \bigr).
\end{align}
This rewriting, where the right-hand side is hermitian, will be useful later in \cref{sec:CompleteSDPHierarchyMainSec}.

\cref{eq:KMSGroundCondition} captures the intuition that any perturbation of $\omega$ by an observable must increase the energy of the system.
We note that the KMS ground state is generally not unique.
In fact, the set of all KMS ground states is weak-${*}$ compact and convex, and its extremal points are also pure states on $\cA$~\cite[Chapter~5.3.3]{BratteliRobinsonVol2}.
The following proposition/definition of \emph{bulk Hamiltonians} connects the dynamical formulation with the usual spectral definition of ground states. 
It relies on the Gelfand-Naimark-Segal (GNS) representation~\cite[Theorem~9.14]{takesaki2002theory}. 

\begin{proposition}[Bulk Hamiltonian]\label{prop:BulkHamiltonian}
Let $(\cA, \tau_{\familyH}, \delta_{\familyH})$ be a dynamical system induced by the local Hamiltonian family $\familyH = \{H^{\latticeL}\}_L$.
Let $\omega: \cA \to \mathds{C}$ be a KMS ground state for $\familyH$ and let $(\cH_{\omega},\pi_{\omega},\ket{\Omega_{\omega}})$ be its GNS representation.

Then there exists a unique positive self-adjoint (possibly unbounded) operator $H_{\omega}$ with a dense domain in $\cH_{\omega}$, called the \emph{bulk Hamiltonian} associated with $\omega$, such that
\begin{equation}\label{eq:BulkHamiltonianProperty}
    \begin{aligned}
        \pi_{\omega}(\tau_{\familyH}^t(a)) &= \mathrm{e}^{\mathrm{i} t H_{\omega}} \pi_{\omega}(a) \mathrm{e}^{-\mathrm{i} t H_{\omega}}, \\
        \mathrm{e}^{\mathrm{i} t H_{\omega}} \pi_{\omega}(a) \ket{\Omega_{\omega}} &= \pi_{\omega}(\tau_{\familyH}^t(a)) \ket{\Omega_{\omega}}, \\
        H_{\omega} \ket{\Omega_{\omega}} &= 0, \\
        \ H_{\omega} \pi_{\omega}(a) \ket{\Omega_{\omega}} &= - \mathrm{i} \pi_{\omega}(\delta_{\familyH}(a)) \ket{\Omega_{\omega}}.
    \end{aligned}
\end{equation}
The first two identities hold for every $t\in\mathds{R}$ and $a\in\cA$, while the last identity holds for every $a$ in the domain of $\delta_{\familyH}$.
Consequently, $\Spectrum{H_{\omega}} \subset [0,\infty)$, and $\ket{\Omega_{\omega}}$ is a ground state for $H_{\omega}$ in the usual spectral sense.
\end{proposition}
\begin{proof}
    See~\cite[Proposition~5.3.19]{BratteliRobinsonVol2}.
\end{proof}

In other words, for any KMS ground state $\omega$ of $(\cA, \tau_{\familyH}, \delta_{\familyH})$, there exists a unique bulk Hamiltonian $H_{\omega}$ in the GNS representation that implements the dynamics induced by the family $\familyH = \{H^{\latticeL}\}_L$.
Moreover, the GNS cyclic state $\ket{\Omega_{\omega}}$ is the zero-energy ground state of this bulk Hamiltonian $H_{\omega}$ in the usual sense.

\begin{remark}[Connection to finite-dimensional system]\label{rem:FiniteDimEquivalence}
For a finite-dimensional quantum system $M_q(\mathds{C})$, the KMS ground state in \cref{def:KMSGroundState} for a Hamiltonian $H \in M_q(\mathds{C})$ coincides with the usual spectral ground state for $H$.
Indeed, every state $\omega$ on $M_q(\mathds{C})$ satisfies $\omega(\cdot) = \Tr{\rho\; \cdot}$ for some density operator $\rho$.
It is straightforward to check that every ground state for $H$ (as the eigenvector corresponding to the lowest eigenvalue) satisfies \cref{eq:KMSGroundCondition}.
Conversely, diagonalizing $H$ with its eigenvectors $\{\ket{E_i}\}$ and choosing $a_{ij} = \ketbra{E_i}{E_j}$ forces a KMS ground state to be supported only on the minimum eigenspace of $H$.

It follows that the KMS ground state can be regarded as a natural generalization of the usual definition of ground states to infinite systems.
\end{remark}

\begin{remark}[Existence of KMS ground states]\label{rem:ExistenceKMSGroundState}
    Consider the local Hamiltonians $\familyH=\{H^{\latticeL}\}_L$ induced by a finite-range interaction, and let $(\cA,\tau_{\familyH},\delta_{\familyH})$ be the associated dynamical system.
    For each $L$, choose a normalized eigenvector $\ket{\psi_L}$ of $H^{\latticeL}$ with minimal eigenvalue, that is, a ground state of the finite-size system.
    Let $\omega_L(a)\coloneqq\sandwich{\psi_L}{a}{\psi_L}$ for $a\in\cA_{\latticeL}$ be the induced vector state.

    By compactness of the state space on each local matrix algebra and a diagonal subsequence argument, there exists a sequence $L_k\to\infty$ such that $\omega_{L_k}(a)$ converges for every $a\in\cA_{\mathrm{loc}}$.
    The limits define consistent normalized positive functionals on the local algebras, and hence a state $\omega$ on $\cA$ by continuity.

    For every local observable $a$, once $\latticeL$ contains its support and all interaction terms meeting that support, we have $\delta_{\familyH}(a)=\mathrm{i}[H^{\latticeL},a]$.
    By \cref{rem:FiniteDimEquivalence}, $\omega_L$ satisfies the finite-size KMS ground-state condition, so $-\mathrm{i}\omega_L(a^*\delta_{\familyH}(a))\geq0$.
    Since $a^*\delta_{\familyH}(a)$ is local, passing to the limit along $L_k$ gives $-\mathrm{i}\omega(a^*\delta_{\familyH}(a))\geq0$.
    Extending from the local core shows that $\omega$ is a KMS ground state of $(\cA,\tau_{\familyH},\delta_{\familyH})$.
    Thus the dynamical system admits at least one KMS ground state.
    In the terminology of \cref{sec:FiniteVolumeSpectralGap}, this construction gives a limit state with open boundary conditions.
\end{remark}

\subsection{Spectral gaps in the bulk}\label{sec:AlgebraicSpectralGap}
We now discuss the algebraic formalism of spectral gaps in the thermodynamic limit, which characterizes bulk excitations in many-body systems~\cite[Definition~3.2]{ogata2023classification}.

\begin{definition}[Locally non-degenerate bulk gap on states]\label{def:GappedAlgebraic}
    Let $(\cA, \tau_{\familyH}, \delta_{\familyH})$ be a dynamical system induced by the local Hamiltonian family $\familyH = \{H^{\latticeL}\}_L$.
    Let $\omega: \cA \to \mathds{C}$ be a KMS ground state for $\familyH$ and let $(\cH_{\omega},\pi_{\omega},\ket{\Omega_{\omega}})$ be its GNS representation with the bulk Hamiltonian $H_{\omega}$.
    
    We say that the system $(\cA, \tau_{\familyH}, \delta_{\familyH})$ with the state $\omega$ is \emph{locally non-degenerate bulk-gapped} if there exists a constant $\gamma > 0$ such that the bulk Hamiltonian $H_{\omega}$:
    \begin{enumerate}
        \item  is locally non-degenerate in the sense that
        \begin{align}
            \ker(H_{\omega}) = \Span{\ket{\Omega_{\omega}}};
        \end{align}
        \item satisfies the bulk spectral gap condition
        \begin{align}
            \Spectrum{H_{\omega}} \cap (0, \gamma) = \emptyset.
        \end{align}
    \end{enumerate}
    In this case, we say the system $(\cA, \tau_{\familyH}, \delta_{\familyH})$ with the state $\omega$ has a \emph{locally non-degenerate bulk spectral gap of at least $\gamma$}.
\end{definition}

There is a dynamical characterization of locally non-degenerate bulk-gapped systems (by substituting $A=a-\omega(a)\id$ in \cite[Eq.~(2.22)]{nachtergaele2024stability}), which plays a central role in \cref{sec:CompleteSDPHierarchyMainSec}.
\begin{proposition}[Dynamical criterion for locally non-degenerate bulk gap on ground states]\label{prop:GappedSystemProperty}
    Let $(\cA, \tau_{\familyH}, \delta_{\familyH})$ be a dynamical system induced by the local Hamiltonian family $\familyH = \{H^{\latticeL}\}_L$.
    Let $\omega: \cA \to \mathds{C}$ be a KMS ground state for $\familyH$ and let $(\cH_{\omega},\pi_{\omega},\ket{\Omega_{\omega}})$ be its GNS representation with the bulk Hamiltonian $H_{\omega}$.
    
    Then, the system $(\cA, \tau_{\familyH}, \delta_{\familyH})$ with the state $\omega$ has a locally non-degenerate bulk spectral gap of at least $\gamma$ if and only if
    \begin{align}\label{eq:GappedSystemProperty}
        -\mathrm{i} \omega(a^* \delta_{\familyH}(a))  \geq \gamma ( \omega(a^* a) - \abs{\omega(a)}^2 ).
    \end{align}
    for all $a$ in the domain of $\delta_{\familyH}$.
\end{proposition}

Note that $\omega(a^*a)-|\omega(a)|^2\geq 0$ for any state $\omega$ and any $a$ by the Cauchy--Schwarz inequality for states~\cite[II.6.2.6]{blackadar2006operator} so \cref{eq:GappedSystemProperty} automatically implies \cref{eq:KMSGroundCondition}.
Thus, by the discussion following \cref{def:KMSGroundState}, \cref{eq:GappedSystemProperty} is equivalent to the stationary condition $\omega(\delta_{\familyH}(a))=0$ and
\begin{align}\label{eq:GappedSystemPropertyHermitian}
     -\frac{\mathrm{i}}{2}\omega\bigl( a^*\delta_{\familyH}(a) - \delta_{\familyH}(a^*)a\bigr) \geq \gamma \bigl( \omega(a^* a) - \abs{\omega(a)}^2 \bigr),
\end{align}
for all $a$ in the domain of $\delta_{\familyH}$.
This symmetrized expression of \cref{eq:GappedSystemPropertyHermitian} is used later in \cref{sec:CompleteSDPHierarchyMainSec}.

\cref{eq:GappedSystemProperty}, equivalently $\omega(\delta_{\familyH}(a)) = 0$ and \cref{eq:GappedSystemPropertyHermitian}, corresponds to the following physical intuition:
\begin{enumerate}[label=(\alph*)]
    \item The left-hand side $-\mathrm{i} \omega(a^* \delta_{\familyH}(a))$ measures the energy increase due to the excitation $a$.
    It is ``energy increase'' because $\omega$ satisfies \cref{eq:KMSGroundCondition}.
    \item The term $\omega(a^* a) - \abs{\omega(a)}^2$ measures the variance of the operator $a$.
    Geometrically, it represents the squared norm of the component of the vector $\pi_{\omega}(a)\ket{\Omega_{\omega}}$ that lies in the orthogonal complement $\{\ket{\Omega_{\omega}}\}^\perp$.
    \item The parameter $\gamma$ represents a lower bound on the ``conversion rate'' between distance (the right-hand-side) and energy (the left-hand-side).
    The inequality implies that the energy increase is proportional to how far $\pi_{\omega}(a)$ pushes the ground state $\ket{\Omega_{\omega}}$ away from $0$-eigenspace.
    Specifically, every unit that $a$ drives $\ket{\Omega_{\omega}}$ out of the $0$-eigenspace costs at least $\gamma$ units of energy.
\end{enumerate}

\begin{remark}\label{rem:NonDegeneracyDifferent}
    Note that ``non-degeneracy'' of a given KMS ground state $\omega$ does not refer to the same thing in the conventional physics literature.
    Local non-degeneracy does not exclude the global degeneracy coming from several distinct pure KMS ground states with inequivalent GNS representations and only excludes zero-energy excitations that can be reached by local operators from the chosen KMS ground state.
    
    A simple example is the ordered phase of the transverse-field Ising chain; see \cref{sec:TransverseFieldIsing}.
    In the usual physics language this phase is two-fold degenerate, corresponding to the two pure KMS ground states induced by $\ket{\uparrow\cdots}$ and $\ket{\downarrow\cdots}$.
    Each of these two states admits a locally non-degenerate bulk spectral gap, yet is separated by a global spin flip and is not connected by local observables.
    Thus, this global degeneracy should not be confused with the local non-degeneracy tested by \cref{def:GappedAlgebraic,prop:GappedSystemProperty}.
\end{remark}

Recall that the set of all KMS ground states is convex.
We make the following observation about local non-degeneracy and pure states (i.e., the extremal points of the KMS ground state set).
\begin{proposition}\label{prop:NonPureKMSGroundStateDegenerate}
    Let $(\cA, \tau_{\familyH}, \delta_{\familyH})$ be a dynamical system induced by the local Hamiltonian family $\familyH = \{H^{\latticeL}\}_L$.
    Let $\omega: \cA \to \mathds{C}$ be a KMS ground state for $\familyH$.

    If $\omega$ is not pure, then its bulk Hamiltonian $H_{\omega}$ is locally degenerate.
    Subsequently, \cref{eq:GappedSystemProperty} holds for all $a$ only when $\gamma = 0$.
\end{proposition}
\begin{proof}
    Let $(\cH_\omega,\pi_\omega,\ket{\Omega_\omega})$ be the GNS representation of $\omega$, and let $H_\omega$ be the corresponding bulk Hamiltonian.
    Since $\omega$ is not pure, there exist two distinct KMS ground state $\omega_1$ and $\omega_2$, $\omega_1 \neq \omega_2$, and $s \in (0,1)$ such that $\omega = s \omega_1 + (1-s)\omega_2$.
    
    Since $s\omega_1\leq \omega$ and $(1-s)\omega_2\leq \omega$, the Arveson's Radon--Nikodym theorem~\cite[II.6.4.6]{blackadar2006operator} gives unique operators $T_1,T_2\in \pi_\omega(\cA)'$, $0 \leq T_1, T_2 \leq \id$, such that, for all $a\in\cA$,
    \begin{equation}
        \begin{aligned}
            \sandwich{\Omega_\omega}{T_1\pi_\omega(a)}{\Omega_\omega}
            &= s\omega_1(a),\\
            \sandwich{\Omega_\omega}{T_2\pi_\omega(a)}{\Omega_\omega}
            &= (1-s)\omega_2(a).
        \end{aligned}
    \end{equation}
    Moreover, uniqueness gives $T_1+T_2=\id$, and neither $T_1$ nor $T_2$ is a scalar multiple of the identity since the convex decomposition is assumed to be nontrivial.
    Both $T_1\ket{\Omega_\omega}, T_2\ket{\Omega_\omega} \neq 0$ since $\omega_1, \omega_2 \neq 0$.

    We now show that $T_1\ket{\Omega_\omega}$ lies in $\ker(H_\omega)$.
    Recall from \cref{prop:BulkHamiltonian} that $\pi_{\omega}(\tau_{\familyH}^t(a)) = \mathrm{e}^{\mathrm{i} t H_{\omega}} \pi_{\omega}(a) \mathrm{e}^{-\mathrm{i} t H_{\omega}}$ and $\mathrm{e}^{\mathrm{i} t H_{\omega}} \ket{\Omega_\omega} = \ket{\Omega_\omega}$.
    For any $a, b \in \cA$, since $\omega_1$ is a $\tau_{\familyH}$-invariant KMS ground state, we compute
    \begin{equation}
        \begin{aligned}
            \sandwich{\Omega_\omega}{\pi_\omega(a)T_1\pi_\omega(b)}{\Omega_\omega} &=
            s\omega_1(ab) \\
            &= s\omega_1(\tau_{\familyH}^t(ab))\\
            &= s\omega_1(\tau_{\familyH}^t(a) \tau_{\familyH}^t(b)) \\
            &= \sandwich{\Omega_\omega}{T_1\pi_\omega(\tau_{\familyH}^t(a)) \pi_\omega(\tau_{\familyH}^t(b))}{\Omega_\omega} \\
            &= \sandwich{\Omega_\omega}{\pi_\omega(\tau_{\familyH}^t(a)) T_1 \pi_\omega(\tau_{\familyH}^t(b))}{\Omega_\omega} \\
            &= \sandwich{\Omega_\omega}{ (\mathrm{e}^{-\mathrm{i} t H_{\omega}} \pi_\omega(a) \mathrm{e}^{\mathrm{i} t H_{\omega}}) T_1 (\mathrm{e}^{-\mathrm{i} t H_{\omega}} \pi_\omega(b) \mathrm{e}^{\mathrm{i} t H_{\omega}})}{\Omega_\omega} \\
            &= \sandwich{\Omega_\omega}{ \pi_\omega(a) (\mathrm{e}^{\mathrm{i} t H_{\omega}} T_1 \mathrm{e}^{-\mathrm{i} t H_{\omega}}) \pi_\omega(b) }{\Omega_\omega},
        \end{aligned}
    \end{equation}
    where we also use the fact that $\tau_{\familyH}$ is a $*$-automorphism on $\cA$ and $T_1 \in \pi_\omega(\cA)'$.
    It follows from the cyclicity of $\ket{\Omega_\omega}$ that $T_1 = \mathrm{e}^{\mathrm{i} t H_{\omega}}T_1 \mathrm{e}^{-\mathrm{i} t H_{\omega}}$, i.e., $T_1$ commutes with $\mathrm{e}^{\mathrm{i} t H_{\omega}}$.
    Thus,
    \begin{align}
        H_{\omega} T_1 \ket{\Omega_\omega} = -\mathrm{i}\lim_{t \to 0} \frac{\mathrm{e}^{\mathrm{i} t H_{\omega}}T_1 \ket{\Omega_\omega} - T_1 \ket{\Omega_\omega} }{t} = 0.
    \end{align}
    The same argument shows that $T_2\ket{\Omega_\omega} \in \ker(H_\omega)$ as well.
    
    The vector $T_1\ket{\Omega_\omega}$ is not colinear with $\ket{\Omega_\omega}$.
    Indeed, if $T_1\ket{\Omega_\omega}=\lambda\ket{\Omega_\omega}$, then $s\omega_1 = \lambda\omega$.
    Evaluating $a=1$ gives $\lambda=s$, and hence $\omega_1=\omega$, contradicting the nontriviality of the convex decomposition.
    It follows that $\ker(H_\omega)$ is of at least dimension $2$ and $H_\omega$ is locally degenerate.

    For clarity, we directly prove the ``subsequent'' statement without invoking \cref{prop:GappedSystemProperty}.
    As $\sandwich{\Omega_\omega}{T_1 }{\Omega_\omega} = s\omega_1(1) = s$, set
    \begin{align}
        \ket{\eta} \coloneqq \frac{T_1 \ket{\Omega_\omega} - s\ket{\Omega_\omega}}{\sqrt{ \sandwich{\Omega_\omega}{T_1^2 }{\Omega_\omega} - s^2}}.
    \end{align}
    Then $\ket{\eta} \perp \ket{\Omega_\omega}$, $\braket{\eta}{\eta} = 1$, and $\ket{\eta} \in \ker(H_\omega)$ since both $T_1\ket{\Omega_\omega}$ and $\ket{\Omega_\omega}$ are.

    Therefore, since $\pi_{\omega}(\cA_{\mathrm{loc}})\ket{\Omega_{\omega}}$ is a core for $H_{\omega}$, there exists a sequence $\{a_k\}_k \subset \cA_{\mathrm{loc}}$ such that $\pi_{\omega}(a_k) \ket{\Omega_\omega} \to \ket{\eta}$ and $H_\omega \pi_{\omega}(a_k) \ket{\Omega_\omega} \to H_\omega\ket{\eta} = 0$ as $k \to \infty$.
    It follows from \cref{prop:BulkHamiltonian} that 
    \begin{align}
        -\mathrm{i} \omega(a_k^* \delta_{\familyH}(a_k)) = \sandwich{\Omega_{\omega}}{\pi_{\omega}(a_k)^* H_\omega \pi_{\omega}(a_k)}{\Omega_{\omega}} \to \sandwich{\eta}{H_\omega}{\eta} = 0.
    \end{align}
    On the other hand,
    \begin{align}
        \omega(a_k^* a_k) - \abs{\omega(a_k)}^2 \to \braket{\eta}{\eta} - \abs{\braket{\Omega_\omega}{\eta}}^2 = 1.
    \end{align}
    Thus, \cref{eq:GappedSystemProperty} can hold only when $\gamma = 0$.
\end{proof}

For $\gamma \geq 0$, let $\cG_\gamma$ denote the set of KMS ground states of $(\cA,\tau_{\familyH},\delta_{\familyH})$ satisfying \cref{eq:GappedSystemProperty} with parameter $\gamma$.
For $\gamma>0$, these are precisely the states with locally non-degenerate bulk gap at least $\gamma$.
For $\gamma=0$, $\cG_0$ is the full set of KMS ground states.
By \cref{prop:NonPureKMSGroundStateDegenerate}, $\cG_\gamma$ with $\gamma > 0$ need not be convex, even though the set $\cG_0$ of KMS ground states is convex.
On the other hand, $\cG_\gamma$ is weak-$^*$-closed.
Indeed, for every fixed $a$, both sides of \cref{eq:GappedSystemProperty} depend weak-$^*$-continuously on $\omega$, and thus define a weak-$^*$-closed set; intersecting over all $a$ preserves closedness.
In fact, $\cG_\gamma$ is even weak$^*$-compact as a closed subset of the weak$^*$-compact space of all states.

So far, the discussion has been state-dependent, it concerns whether a given KMS ground state is locally non-degenerate bulk-gapped.
We now define bulk spectral gaps at the level of dynamical systems. 
\begin{definition}[Bulk spectral gap]\label{def:BulkSpectralGap}
    Let $(\cA,\tau_{\familyH},\delta_{\familyH})$ be a dynamical system induced by the local Hamiltonian family $\familyH=\{H^{\Lambda(L)}\}_L$.
    The \emph{bulk spectral gap} $\Delta_{\mathrm{bulk}}$ of $(\cA,\tau_{\familyH},\delta_{\familyH})$ is the supremum of all $\gamma\geq0$ for which there exists a KMS ground state $\omega$ satisfying \cref{eq:GappedSystemProperty}.
    Namely,
    \begin{align}\label{eq:BulkGapDefinition}
        \Delta_{\mathrm{bulk}}\coloneqq\sup\{ \gamma\geq0 \mid \cG_\gamma\neq\emptyset \}.
    \end{align}
\end{definition}

To finish up the discussion of the algebraic spectral gap, we introduce the notion of bulk gaplessness and strong bulk gaplessness.
\begin{definition}[Bulk-gaplessness]\label{def:GaplessAlgebraic}
    Let $(\cA,\tau_{\familyH},\delta_{\familyH})$ be a dynamical system, and let $\omega$ be a KMS ground state with bulk Hamiltonian $H_\omega$.
    We say that $\omega$ is \emph{bulk-gapless} if
    \begin{align}\label{eq:BulkGaplessness}
        \Spectrum{H_\omega}\cap(0,\varepsilon)\neq\emptyset
        \qquad\text{for every }\varepsilon>0.
    \end{align}
    We say that $\omega$ is \emph{strongly bulk-gapless} if there exists $c>0$ such that
    \begin{align}\label{eq:StrongBulkGaplessness}
        [0,c]\subset\Spectrum{H_\omega}.
    \end{align}
    We call the dynamical system $(\cA,\tau_{\familyH},\delta_{\familyH})$ strongly bulk-gapless if it admits at least one KMS ground state and there exists $c>0$ such that every KMS ground state $\omega$ satisfies \cref{eq:StrongBulkGaplessness} with this same $c$.
\end{definition}

\begin{remark}\label{rem:StrongBulkGaplessness}
    Strong bulk gaplessness requires an entire interval of excitation energies, rather than only a sequence of positive energies converging to zero.
    It is the infinite-volume spectral analogue of strong gaplessness under a fixed boundary condition (\cref{def:StrongGaplessCPW15}), as defined in~\cite[Definition~2]{cubitt2022undecidability}.
    In the halting case of \cref{thm:BulkGapUndecidability}, the dynamical system is strongly bulk-gapless.
    Indeed, we prove the stronger statement $\Spectrum{H_\omega}=[0,\infty)$ for every KMS ground state $\omega$, so every $c>0$ satisfies \cref{eq:StrongBulkGaplessness}.
\end{remark}

\section{Complete SDP hierarchy for certifying bulk spectral gaps}\label{sec:CompleteSDPHierarchyMainSec}
We now construct the SDP hierarchy used to certify upper bounds on the bulk spectral gap.
The hierarchy builds on the state polynomial optimization~\cite{klep2024state} and is designed to test, for a fixed threshold $\gamma$, whether there exists a KMS ground state satisfying the dynamical bulk-gap criterion of \cref{prop:GappedSystemProperty}.
Its central property is completeness: feasibility at every level is equivalent to the existence of such a KMS ground state, while infeasibility at some finite level rules out that threshold and hence gives a certified upper bound on the bulk spectral gap of \cref{def:BulkSpectralGap}.

We first present the high-level intuition behind the hierarchy of relaxations in \cref{sec:ConvexRelaxationHierarchy}.
In \cref{sec:StatePolyFormulation}, we formulate it as a state polynomial optimization problem whose levels are SDPs, and in \cref{sec:CompletenessHierarchy} we prove its completeness.
Finally, \cref{sec:SemidecisionExtension} derives the semidecision procedure and discusses extensions to more general lattices and to symmetry constraints.

\subsection{Intuition for the hierarchy of relaxations}\label{sec:ConvexRelaxationHierarchy}
Let the interaction $\Phi$ be finite-range with the interaction length $l$.
Our goal is to turn the dynamical characterization in \cref{prop:GappedSystemProperty} into a tractable algorithm.
Rather than working with the full quasi-local $C^*$-algebra $\cA$, we restrict, for each $n$, to the local algebra $\cA_{\Lambda(n)}$ of observables supported on $\Lambda(n)$, equipped with the corresponding finite-volume Hamiltonian $H^{\Lambda(n)}$.
Note that we use $n$ (rather than $L$ as in \cref{sec:AlgebraicFormalismMainSec}) to label the finite regions $\Lambda(n)$ in this section, since $n$ will also denote the corresponding level of the relaxation hierarchy.

Within $\cA_{\Lambda(n)}$ we introduce the notion of operator complexity.
Any $a \in \cA_{\Lambda(n)}$ can be written as a noncommutative polynomial in the local generators $\{{x^{(i)}_j}\}_{i \in \Lambda(n)}$.
We define the degree of $a$, denoted by $\deg(a)$, to be the minimum total degree among all such polynomial representations.
Finally, let $W_{\Lambda(n)}^d \subset \cA_{\Lambda(n)}$ denote the set of all monomials in the generators of total degree $\leq d$ and $\cA_{\Lambda(n)}^d = \SpanNoBracket{(W_{\Lambda(n)}^d)}$ the polynomials in $\cA_{\Lambda(n)}$ of total degree $\leq d$.

In the following we assume that the degree of the local Hamiltonians is uniformly bounded, i.e., there exists a (minimal) integer $\deg(H)$ independent of the region $\Lambda(n)$ such that $H^{\Lambda(n)} \in \cA^{\deg(H)}_{\Lambda(n)}$ for all $n$.
This assumption is fulfilled for many physical models; for example, for all translation-invariant Hamiltonians with finite-range interactions.

\begin{definition}[The hierarchy of relaxations]\label{def:HierarchyOfRelaxations}
    Let $(\cA, \tau_{\familyH}, \delta_{\familyH})$ be a dynamical system induced by a family of local Hamiltonians $\familyH = \{H^{\Lambda(n)}\}_n$ with the interaction length $l$.
    Fix a gap parameter $\gamma \geq 0$.

    We first define the feasibility hierarchy for certifying the existence of a KMS ground state with a locally non-degenerate bulk spectral gap at least $\gamma$.
    For $n,d\in\mathds N$, with $n,d$ sufficiently large, the relaxation on $\cA_{\Lambda(n)}^{2d}$ asks whether there exists a linear functional $\omega_n^d:\cA_{\Lambda(n)}^{2d}\to \mathds{C}$ satisfying
    \begin{equation}\label{eq:GappedFeasibilityHierarchy}
        \begin{aligned}
          &\omega_n^d(1)=1, &&\text{(normalization)}\\[0.4em]
          &\omega_n^d(a^*a)\ge0,
            \quad\forall\,a\in\cA_{\Lambda(n)}^{d}, &&\text{(state positivity)}\\[0.4em]
          &\omega_n^d\bigl([H^{\Lambda(n)},a]\bigr)=0,
            \quad\forall\,a\in\cA_{\Lambda(n-l)}^{2d-\deg(H)},
            &&\text{(stationarity)}\\[0.4em]
                &\frac{1}{2}\,
            \omega_n^d\!\left(
                a^*[H^{\Lambda(n)},a]-[H^{\Lambda(n)},a^*]a
            \right)
            \;\ge\;
            \gamma\bigl(\omega_n^d(a^*a)-|\omega_n^d(a)|^2\bigr), \quad\forall\,a\in\cA_{\Lambda({n-l})}^{d- \lceil \deg(H) / 2 \rceil}
            &&\text{(bulk spectral gap)}
        \end{aligned}
    \end{equation}

    More generally, let $O=O^*\in\cA$ be a local observable.
    We are interested in the extremal expectation values of $O$ over $\cG_\gamma$, the weak-$*$-closed set of KMS ground states of $(\cA,\tau_{\familyH},\delta_{\familyH})$ with locally non-degenerate bulk gap at least $\gamma$:
    \begin{equation}\label{eq:ExtremalExpectationValues}
        \begin{aligned}
            \langle O\rangle_{\min}
            &\coloneqq
            \min_{\omega\in\cG_\gamma}\omega(O),\\
            \langle O\rangle_{\max}
            &\coloneqq
            \max_{\omega\in\cG_\gamma}\omega(O).
        \end{aligned}
    \end{equation}
    
    Whenever $n,d$ are sufficiently large such that $O,H^{\Lambda(n)}\in\cA_{\Lambda(n)}^{2d}$, we define the corresponding upper and lower relaxations by optimizing $O$ over the same feasible set:
    \begin{equation}\label{eq:GappedRelaxHierarchy}
        \begin{aligned}
        \alpha^{\max/\min}_{n,d} &\;=\operatorname*{maximize/minimize}_{\substack{
            \omega_n^d: \cA_{\Lambda(n)}^{2d} \to \mathds{C} \text{ linear}
        }} 
          \;\omega_n^d(O),\\[0.4em]
        \mathrm{s.t.}\quad
        &\omega_n^d(1)=1, &&\text{(normalization)}\\[0.4em]
          &\omega_n^d(a^*a)\ge0,
            \quad\forall\,a\in\cA_{\Lambda(n)}^{d}, &&\text{(state positivity)}\\[0.4em]
          &\omega_n^d\bigl([H^{\Lambda(n)},a]\bigr)=0,
            \quad\forall\,a\in\cA_{\Lambda(n-l)}^{2d-\deg(H)},
            &&\text{(stationarity)}\\[0.4em]
                &\frac{1}{2}\,
            \omega_n^d\!\left(
                a^*[H^{\Lambda(n)},a]-[H^{\Lambda(n)},a^*]a
            \right)
            \;\ge\;
            \gamma\bigl(\omega_n^d(a^*a)-|\omega_n^d(a)|^2\bigr), \quad\forall\,a\in\cA_{\Lambda({n-l})}^{d- \lceil \deg(H) / 2 \rceil}.
            &&\text{(bulk spectral gap)}
        \end{aligned}
    \end{equation}
    By setting $O=1$, the trivial optimization reduces to the feasibility hierarchy of \cref{eq:GappedFeasibilityHierarchy}.
\end{definition}

The normalization condition and state positivity constraint correspond to the definition of states on $\cA$, and the stationarity constraint and bulk spectral gap condition are meant to recover \cref{eq:StationaryCondition,eq:GappedSystemPropertyHermitian}, which are equivalent to \cref{eq:GappedSystemProperty}.
Note that while the term $|\omega_n^d(a)|^2$ in the constraints is nonlinear, this optimization problem can be cast as an SDP using the framework of state polynomial optimization~\cite{klep2024state}; see \cref{def:GappedRelaxHierarchyStatePoly} for the full formulation.

\begin{remark}\label{rem:SmallerSupportGapConstraint}
    The subscript $\Lambda(n-l)$ in $\cA_{\Lambda(n-l)}^{2d-\deg(H)}$ in stationarity constraint and $\cA_{\Lambda({n-l})}^{d- \lceil \deg(H) / 2 \rceil}$ in the bulk spectral gap constraint is essential.
    Since the interaction has a range $l$, any such excitation $a$ is separated from the boundary of $\Lambda(n)$ by enough room for the commutator $[H^{\Lambda(n)},a]$ to coincide with the infinite-volume dynamic derivation $\delta_{\familyH}(a)$.
 Thus, the finite relaxation tests the bulk gap condition in a local window without imposing any boundary condition outside $\Lambda(n)$.
    Without this buffer, the constraint would involve boundary terms of the chosen finite-volume Hamiltonian, reducing the test to a finite-size gap calculation with that boundary condition.
\end{remark}

\begin{remark}\label{rem:ComputabilityQuantifierElimination}
For fixed $(n,d)$, the relaxation \cref{eq:GappedFeasibilityHierarchy} is a finite-dimensional feasibility problem (over finitely many parameters specifying the linear functional $\omega_n^d$ on $\cA_{\Lambda(n)}^{2d}$).
Consequently, it can in principle be solved using quantifier elimination over the reals (the Tarski--Seidenberg theorem),
for instance via cylindrical algebraic decomposition (CAD).
However, quantifier elimination/CAD has doubly exponential worst-case complexity, and is therefore impractical except for very small instances~\cite{basupollackroy2006algorithms}.
In contrast, the state polynomial optimization yields an SDP hierarchy, which is capable of handling optimization problems of \cref{eq:GappedRelaxHierarchy} beyond feasibility, and is far more effective in practice; see \cref{sec:NumericalShowcase}.
\end{remark}

This hierarchy of relaxations admits a clear physical interpretation.
If the system $(\cA, \tau_{\familyH}, \delta_{\familyH})$ admits a KMS ground state $\omega$ with a locally non-degenerate bulk gap of at least $\gamma$, then $\omega$ restricted to the subspace $\cA_{\Lambda(n)}^{2d}$, denoted by $\omega_n^{2d}$, must be a valid solution to \cref{eq:GappedFeasibilityHierarchy,eq:GappedRelaxHierarchy} for every $(n, d)$.
Furthermore, it is straightforward to see that the optimal values of the hierarchy form a monotonic sequence as $n$ (or $d$) increases.

Consequently, if the problem is found to be infeasible at any level $(n, d)$, it provides a rigorous certificate that the system cannot admit a bulk spectral gap of at least $\gamma$.
On the other hand, feasibility at a single level $(n,d)$ does not certify a bulk spectral gap; it only confirms the absence of gap-violating excitations due to operators in $\cA_{\Lambda(n)}^{2d}$ rather than the full $\cA$.

We will show in \cref{thm:CompletenessSDPHierarchy} that the state polynomial optimization formulation of \cref{def:HierarchyOfRelaxations} is complete.
Furthermore, the hierarchy estimates the expectation values in the sense that $\alpha^{\min}_{n, d} \nearrow \langle O \rangle_{\min}$ and $\alpha^{\max}_{n, d} \searrow \langle O \rangle_{\max}$.

\subsection{State polynomial formulation}\label{sec:StatePolyFormulation}\label{sec:ConvexRelaxationHierarchyCompleteness}
We now formalize the state polynomial optimization problem corresponding to \cref{def:HierarchyOfRelaxations}; each of its levels is an SDP.

Recall that $\cA_{\mathrm{loc}} = \bigcup_{L \in \mathds{N}} \cA_{\latticeL}$ and $\cA \coloneqq \overline{\cA_{\mathrm{loc}}}$, where each matrix algebra $\cA_i \simeq M_q(\mathds{C})$, $i \in \latticeL$, is generated by $\{x^{(i)}_j\}_j$.
As $\cA_i$ is just matrix algebra, all elements of $\cA_i$ are bounded; there exists a constant $M^{(i)}_j$ such that $M^{(i)}_j - (x^{(i)}_j)^2 \succeq 0$, and in particular, the square root  $s^i_j \coloneqq \bigl(M^{(i)}_j - (x^{(i)}_j)^2 \bigr)^{1/2} \in \cA_i$ exists.
(For example, with spin-$\frac{1}{2}$ chain and Pauli matrices, we have $M^{(i)}_j=1$ and $1 - (X^{(i)})^2 = 1 - (Y^{(i)})^2 = 1- (Z^{(i)})^2 = 0$.)

To begin with, we write the explicit state polynomial optimization formulation of \cref{eq:GappedRelaxHierarchy}, following the notation of~\cite{klep2024state}.
For every (noncommutative) monomial $w$ in the local generators $\{x^{(i)}_j\}$, we introduce a formal \emph{state symbol} $\varsigma(w)$.
Heuristically, $\varsigma(w)$ represents the scalar expectation of $w$ in the state $\varsigma$.

We extend $\varsigma(\cdot)$ by linearity and involution, and impose the compatibility rules for all monomials $w,v$ and complex numbers $c$
\begin{equation}\label{eq:StateSymbolRules}
    \varsigma(1)=1,\quad
    \varsigma(w^*)=\varsigma(w)^*,\quad
    \varsigma(w+cv)=\varsigma(w)+c\,\varsigma(v),
\end{equation}
together with commutativity of all state symbols and the convention that $\varsigma(\cdot)$ treats state symbols as scalars:
\begin{equation}
    [\varsigma(w),\varsigma(v)]=0,\quad
    \varsigma\bigl(w\,\varsigma(v)\bigr)=\varsigma(w)\varsigma(v).
\end{equation}

These rules give rise to the ring of state polynomials
\begin{equation}
    \thinS \coloneqq \mathds{C}[ \varsigma(w) : w \neq 1 \text{ is a monomial in the generators}]
\end{equation}
with an antilinear involution $*$.
We then define the ring of \emph{noncommutative state polynomials} by
\begin{equation}
    \fatS \coloneqq \thinS \otimes \cA_{\mathrm{loc}},
\end{equation}
whose elements are finite sums $\sum_l p_l a_l$ with $p_l \in \thinS$ and $a_l \in \cA_{\mathrm{loc}}$.

We set $\deg(\varsigma(w))\coloneqq\deg(w)$ and extend $\deg(\cdot)$ multiplicatively, so that $\deg(\varsigma(w)\varsigma(v))=\deg(w)+\deg(v)$ and $\deg(\varsigma(w)u)=\deg(w)+\deg(u)$ for $u\in\fatS$.
For each $(n,d)$ we write $\thinS_{\Lambda(n)}^d$ and $\fatS_{\Lambda(n)}^d$ for the truncations to (noncommutative) state polynomials supported on $\Lambda(n)$ and of total degree $\leq d$.
By a \emph{monomial in $\thinS$ (or $\fatS$)} we mean a product of factors $\varsigma(w)$ (and one factor $w'$) with $w$ (and $w'$) monomials in $\cA_{\mathrm{loc}}$.

\begin{definition}[State polynomial optimization for bulk spectral gap]\label{def:GappedRelaxHierarchyStatePoly}
    Fix a gap parameter $\gamma \geq 0$ and assume there is a finite maximal interaction length $l$.
    Let $n, d \in \mathds{N}$ be large enough such that $O, H^{\Lambda(n)} \in \cA_{\Lambda(n)}^{2d}$.
    
    Let $\mathcal{L}^d_n: \thinS_{\Lambda(n)}^{2d} \to \mathds{C}$ be a linear functional.
    The associated moment matrix $\mathbf{M}_d(\mathcal{L}^d_n)$ is indexed by the monomials $s, t \in\fatS_{\Lambda(n)}^{d}$ and defined by
    \begin{align}
    \left[\mathbf{M}_d(\mathcal{L}^d_n)\right]_{s,t} \coloneqq \mathcal{L}^d_n(\varsigma(s^*t)).
    \end{align}
    The moment matrix associated with the bulk spectral gap, $\mathbf{M}_{d-\lceil\deg(H)/2\rceil}^{\mathrm{gap},\gamma}\bigl( \mathcal{L}^d_n\bigr)$, is indexed by the monomials $s,t \in \fatS^{d-\lceil\deg(H)/2\rceil}_{\Lambda(n-l)}$ and defined by
    \begin{equation}
        \Bigl[\mathbf{M}_{d-\lceil\deg(H)/2\rceil}^{\mathrm{gap},\gamma}\bigl( \mathcal{L}^d_n\bigr)\Bigr]_{s,t}
        \coloneqq
        \mathcal L_n^d\!\left( \frac{1}{2} \Bigl(\varsigma\bigl( s^* [H^{\Lambda(n)}, t] - [H^{\Lambda(n)}, s^*]t \bigr) \Bigr) 
        -
        \gamma \Bigl(
            \varsigma(s^*t) - \varsigma(s^*) \varsigma(t)
        \Bigr) \right).
    \end{equation}

    The state polynomial relaxation at level $(n,d)$ for the local observable $O \in \cA_{\Lambda(n)}^{2d}$ is defined as
    \begin{equation}\label{eq:GappedRelaxHierarchyStatePoly}
        \begin{alignedat}{2}
        \alpha^{\max/\min}_{n,d} &\;\coloneqq \operatorname*{maximize/minimize}_{\substack{
        \mathcal{L}_d \colon \thinS_{\Lambda(n)}^{2d} \to \mathds{C} \text{ linear functional}
        }}
        \;\mathcal{L}^d_n(\varsigma(O)),\\[0.4em]
        \mathrm{s.t.}\quad
        &\;\mathcal{L}^d_n(1) = 1, &&\text{(normalization)}\\
        &\;\mathbf{M}_d(\mathcal{L}^d_n) \succeq 0, &&\text{(positivity)}\\
        &\;\mathcal{L}_n^d\bigl(\varsigma([H^{\Lambda(n)}, w])\bigr)=0, \,\quad \forall\, w \in \fatS^{2d-\deg(H)}_{\Lambda(n-l)}, \quad &&\text{(stationarity)}
        \\
        &\; \mathbf{M}_{d-\lceil\deg(H)/2\rceil}^{\mathrm{gap},\gamma}\bigl( \mathcal{L}^d_n\bigr) \succeq 0. &&(\text{bulk spectral gap})
        \end{alignedat}
    \end{equation}
    By setting $O=1$, the trivial optimization reduces to the hierarchy for \cref{eq:GappedFeasibilityHierarchy}.
\end{definition}

Let us spell out the meaning of the constraints in \cref{eq:GappedRelaxHierarchyStatePoly}.
The positivity constraint $\mathbf{M}_d(\mathcal L_n^d)\succeq0$ is equivalent to
\begin{align}\label{eq:MomentPositivityExplained}
    \mathcal L_n^d\bigl( \varsigma(q^*q)\bigr)\geq0,
    \qquad
    \forall q\in \fatS_{\Lambda(n)}^d
\end{align}
as in~\cite[Lemma~6.4]{klep2024state}.
In particular, it contains conditions $\mathcal L_n^d(\varsigma(a^*a))\geq0$ for $a\in\cA_{\Lambda(n)}^d$ as special cases, because the moment matrix $\mathbf{M}_d(\mathcal L_n^d)$ is indexed by $\fatS_{\Lambda(n)}^d$ rather than only by $\thinS_{\Lambda(n)}^d$.
The stationarity constraint is equivalent to
\begin{align}
    \label{eq:StationarityConditionEquivalence}
    \mathcal L_n^d\bigl(\varsigma([H^{\Lambda(n)},q])\bigr)
    =
    0,
    \qquad
    \forall q\in \fatS_{\Lambda(n-l)}^{2d-\deg(H)}.
\end{align}
The bulk spectral gap constraint $\mathbf{M}_{d-\lceil\deg(H)/2\rceil}^{\mathrm{gap},\gamma}\bigl( \mathcal{L}^d_n\bigr)\succeq0$ is equivalent to
\begin{align}\label{eq:BulkSpectralGapConditionEquivalence}
    \mathcal L_n^d\!\left(
    \frac{1}{2}\Bigl(
        \varsigma\bigl(q^*[H^{\Lambda(n)},q] - [H^{\Lambda(n)},q^*]q\bigr)
    \Bigr)
    -
    \gamma\bigl(\varsigma(q^*q)-\abs{\varsigma(q)}^2\bigr)
    \right)
    \geq 0,
    \qquad
    \forall q\in\fatS_{\Lambda(n-l)}^{d-\lceil\deg(H)/2\rceil}
\end{align}
which is a state-polynomial formulation of the symmetrized bulk-gap criterion \cref{eq:GappedSystemPropertyHermitian}.

Let $\cQ_n^{d,\mathrm{pos}}$ and $\cQ_n^{d,\mathrm{gap}}$ be the convex cones in $\thinS_{\Lambda(n)}^{2d}$ generated by the hermitian elements appearing in
\cref{eq:MomentPositivityExplained,eq:BulkSpectralGapConditionEquivalence}:
\begin{equation}
    \begin{aligned}
        \cQ_n^{d,\mathrm{pos}}
        &\coloneqq
        \operatorname{cone}
        \Bigl\{
            \varsigma(q^*q)
            \;:\;
            q\in\fatS_{\Lambda(n)}^d
        \Bigr\},\\
        \cQ_n^{d,\mathrm{gap}}
        &\coloneqq
        \operatorname{cone}
        \Bigl\{
            \frac{1}{2}\Bigl(
                \varsigma\bigl(q^*[H^{\Lambda(n)},q] - [H^{\Lambda(n)},q^*]q\bigr)
            \Bigr)
            -
            \gamma\bigl(\varsigma(q^*q)-\abs{\varsigma(q)}^2\bigr)
            \;:\;
            q\in\fatS_{\Lambda(n-l)}^{d-\lceil\deg(H)/2\rceil}
        \Bigr\}.
    \end{aligned}
\end{equation}
Similarly, let $\mathcal{I}_n^d$ be the hermitian part of the linear subspace
appearing in \eqref{eq:StationarityConditionEquivalence}, i.e.,
\begin{equation}
    \mathcal{I}_n^d
    \coloneqq
    \bigl\{
        \varsigma([H^{\Lambda(n)},q-q^*]) \;:\; q \in \fatS_{\Lambda(n-l)}^{2d-\deg(H)}
    \bigr\}
\end{equation}
where we use that $[H^{\Lambda(n)},a]$ is hermitian when $a^* = -a$.
Then the positivity, stationarity, and bulk spectral gap constraints are equivalent to the requirement that $\mathcal L_n^d$ is nonnegative on the convex cone
\begin{align}\label{eq:ConvexConedlevel}
    \cQ_n^d \coloneqq \cQ_n^{d,\mathrm{pos}}+\mathcal{I}_n^d + \cQ_n^{d,\mathrm{gap}}.
\end{align}
By construction, $\cQ_n^d\subset\cQ_{n'}^{d'}$ whenever $n\leq n'$ and $d\leq d'$.

\subsection{Completeness of the hierarchy}\label{sec:CompletenessHierarchy}
We now show that the hierarchy of \cref{def:GappedRelaxHierarchyStatePoly} is complete: feasibility at all relaxation levels is equivalent to the existence of a KMS ground state with a locally non-degenerate bulk spectral gap at least $\gamma$, and the corresponding optimization values converge to the exact extremal expectation values over such states.

\begin{theorem}\label{thm:CompletenessSDPHierarchy}
    Let $(\cA, \tau_{\familyH}, \delta_{\familyH})$ be a dynamical system induced by a family of local Hamiltonians $\familyH = \{H^{\Lambda(n)}\}_n$ with interaction length $l$, and let $O \in \cA$ be a local observable.
    Fix a gap parameter $\gamma \geq 0$.
    
    Then, the state polynomial reformulation (\cref{def:GappedRelaxHierarchyStatePoly}) of the relaxation hierarchy (\cref{def:HierarchyOfRelaxations}) is feasible at every level $(n, d)$ with the parameter $\gamma$ if and only if $\cG_\gamma\neq\emptyset$, that is, there exists a KMS ground state $\omega$ on the dynamical system $(\cA, \tau_{\familyH}, \delta_{\familyH})$ with a locally non-degenerate bulk spectral gap of at least $\gamma$.
    
    Furthermore, the hierarchy values $\alpha^{\min}_{n,d}$ and $\alpha^{\max}_{n,d}$ converge monotonically to the extremal expectation values from \cref{eq:ExtremalExpectationValues}:
    \begin{align}
            \alpha^{\min}_{n, d} \nearrow \langle O \rangle_{\min}, \quad
            \alpha^{\max}_{n, d} \searrow \langle O \rangle_{\max}
    \end{align}
    as $n, d \to \infty$.

    In particular, taking $O=\id$ and writing
    \begin{align}\label{eq:HierarchyThresholdDefinition}
        \gamma_{n,d}\coloneqq\sup\{\gamma\geq0\mid\text{\cref{def:GappedRelaxHierarchyStatePoly} is feasible at level }(n,d)\text{ with parameter }\gamma\},
    \end{align}
    each $\gamma_{n,d}$ is an upper bound on the bulk spectral gap $\Delta_{\mathrm{bulk}}$, and
    \begin{align}\label{eq:HierarchyGapConvergence}
        \gamma_{n,d}\searrow\Delta_{\mathrm{bulk}}\qquad(n,d\to\infty).
\end{align}
\end{theorem}

For the proof of this theorem, we need to construct an archimedean quadratic module $\cQ \subset \thinS$.
As in~\cite{klep2024state}, let $\cQ$ be the convex cone generated by $\varsigma(q^*q)$, the hermitian bulk-gap polynomials $-\frac{\mathrm{i}}{2}\bigl(\varsigma\bigl(q^*\delta_\familyH(q)-\delta_\familyH(q^*)q\bigr)\bigr)-\gamma\bigl(\varsigma(q^*q)-\varsigma(q^*)\varsigma(q)\bigr)$, and the hermitian polynomials encoding the stationarity condition, 
$\varsigma(\delta_\familyH(q-q^*))$, for all $q\in \fatS$.
Then $\cQ$ is the union of the cones $\cQ_n^d$ of \cref{eq:ConvexConedlevel} on which the $\mathcal{L}_n^d$ are required to be positive.
Moreover, $\cQ$ is a quadratic module of $\thinS$, i.e., $p^*\cQ p \subset \cQ$ for all $p\in\thinS$, because the maps $\varsigma$ and $\delta_\familyH$
treat elements of $\thinS$ as scalars.
Next, recall that for each $x_j^{(i)}$ there is an $s_j^{(i)}=\bigl(M_j^{(i)}-(x_j^{(i)})^2\bigr)^{1/2}\in\cA_i$.
Hence, for every $q\in\fatS$,
\begin{align}
    \varsigma\!\left(q^*
    \bigl(M_j^{(i)}-(x_j^{(i)})^2\bigr)q\right)
    =
    \varsigma\!\left((s_j^{(i)}q)^*(s_j^{(i)}q)\right)
    \in \cQ .
\end{align}
Thus, \cite[Lemma~5.4]{klep2024state} applies and therefore the quadratic module $\cQ$ is archimedean.

\begin{proof}
    The necessity of the hierarchy is immediate from the discussion at the end of \cref{sec:ConvexRelaxationHierarchy} and by letting both $\varsigma$ and $\mathcal{L}^d_n$ be induced by the KMS ground state $\omega$.
    For the sufficiency of the hierarchy, assuming the feasibility of \cref{def:GappedRelaxHierarchyStatePoly} for every $(n, d)$, we need to construct a state $\omega: \cA \to \mathds{C}$ that satisfies \cref{eq:GappedSystemProperty} for the dynamical system $(\cA, \tau_{\familyH}, \delta_{\familyH})$ induced by $\familyH = \{H^{\Lambda(n)}\}_n$.
    To this end, we first extract a limit $\mathcal{L}$ on $\thinS$ from the feasible solutions $\mathcal{L}^d_n$ using a standard argument, and then we use the representation theorem for archimedean quadratic modules~\cite[Theorem~5.4.4]{marshall2008positive} (also known as the Kadison-Dubois Theorem) to obtain an extremal solution $\mathcal K$ to construct the desired $\omega$.
    
    First, given the functionals $\mathcal{L}_n^d: \thinS_{\Lambda(n)}^{2d} \to \mathds{C}$, we extract a state $\mathcal{L}$ on the full $\thinS$.
    To this end, note that $\thinS$ is the inductive limit of $\thinS_{\Lambda(n)}^{2d}$ as $n, d \to \infty$.
    As a vector space of countable dimension, we can enumerate a basis $\{p_1, p_2, \dots\}$ for $\thinS$.
    We can further assume that these basis elements are hermitian, i.e., $p_i^* = p_i$.

    For each basis element $p_i \in \thinS$, since $\cQ \subset \thinS$ is archimedean, there exists a constant $m_i > 0$ such that $m_i \pm p_i \in \cQ$.
    As $\cQ$ is the union of the cones $\cQ_n^d$ on which $\mathcal{L}_n^d$ is nonnegative, for all sufficiently large $(n,d)$, we have $m_i \pm p_i \in \cQ_n^d$ and hence $\abs{\mathcal{L}_n^d(p_i)} \leq m_i$.
    That is, for each fixed $p_i$, the sequence $\{\mathcal L_n^d(p_i)\}_{n,d}$ is uniformly bounded and has a convergent subsequence by compactness.
    A standard diagonal argument over the countable basis $\{p_1,p_2,\dots\}$ gives a convergent subsequence of functionals $\mathcal{L}_{n_k}^{d_k}$ such that $\mathcal{L}_{n_k}^{d_k}(p_i)$ is convergent for all $i$ as $k \to \infty$.
    We define $\mathcal L:\thinS \to \mathds{C}$ on the basis $\{p_1, p_2, \dots\}$ by
    \begin{align}
        \mathcal{L}(p_i)\coloneqq \lim_{k\to\infty}\mathcal{L}_{n_k}^{d_k}(p_i),
    \end{align}
    and extend linearly to all of $\thinS$. 
    By construction, $\mathcal L(1)=1$ and $\mathcal L(\cQ)\subset\mathds R_{\ge0}$.

    Since $\cQ \subset \thinS$ is archimedean and $\mathcal{L}(\cQ) \subset \mathds{R}_{\geq0}$, by the representation theorem for archimedean quadratic modules \cite[Theorem~5.4.4]{marshall2008positive}, there exists a probability measure $\mu$, supported on the set of characters $\mathcal{K}$ on $\thinS$ satisfying $\mathcal{K}(1)=1$ and $\mathcal{K}(\cQ)\subset\mathds{R}_{\geq0}$, such that
    \begin{align}\label{eq:KadisonDuboisDecomposition}
        \mathcal{L}(a) = \int \mathcal{K}(a) \, d\mu(\mathcal{K}).
    \end{align}
    Therefore, every $\mathcal{K} \in \mathrm{supp}(\mu)$ gives rise to a feasible solution to \cref{eq:GappedRelaxHierarchyStatePoly} for all $n, d$.
    
    Moreover, since $\mathcal{L}(\varsigma(O))=\int \mathcal{K}(\varsigma(O))\,d\mu(\mathcal{K})$, there exists $\mathcal{K} \in \mathrm{supp}(\mu)$ such that $\mathcal{K}(\varsigma(O))\geq \mathcal{L}(\varsigma(O))$ (and similarly there exists $\mathcal{K}\in \mathrm{supp}(\mu)$ with $\mathcal{K}(\varsigma(O))\leq \mathcal{L}(\varsigma(O))$).
    In case of maximization with $\alpha^{\max}_{n,d}$, choose $\mathcal{K}$ such that $\mathcal{K}(\varsigma(O)) \geq \mathcal{L}(\varsigma(O))$; otherwise, for minimization with $\alpha^{\min}_{n,d}$, fix $\mathcal{K}$ such that $\mathcal{K}(\varsigma(O)) \leq \mathcal{L}(\varsigma(O))$.

    Define the functional $\omega: \cA_{\mathrm{loc}} \to \mathds{C}$ by
    \begin{align}
        \omega(a) \coloneqq \mathcal{K}(\varsigma(a)).
    \end{align}
    It is straightforward to check that $\omega(1)=1$, and $\omega(a^* a) \geq 0$ for all $a \in \cA_{\mathrm{loc}}$ 
    because $\varsigma(a^*a) \in \cQ$ for all $a \in \cA_{\mathrm{loc}}$.
    Hence, $\omega$ is a state on $\cA_{\mathrm{loc}}$.
    Moreover, from $\pm\bigl(-\mathrm{i}\varsigma(\delta_{\familyH}(a-a^*))\bigr)\in\cQ$ for all $a\in\cA_{\mathrm{loc}}$, it follows that $\omega(\delta_{\familyH}(a-a^*))=0$ for all $a\in\cA_{\mathrm{loc}}$, hence $\omega(\delta_{\familyH}(a))=0$ for all $a\in\cA_{\mathrm{loc}}$ because $a=\frac12\bigl((a-a^*)-\mathrm{i}((\mathrm{i}a)-(\mathrm{i}a)^*)\bigr)$.
    This shows that $\omega$ fulfils the stationarity condition \cref{eq:StationaryCondition} on $\cA_{\mathrm{loc}}$.

    Similarly, as $-\frac{\mathrm{i}}{2}\varsigma\bigl(a^*\delta_{\familyH}(a)-\delta_{\familyH}(a^*)a\bigr)-\gamma\bigl(\varsigma(a^*a)-\varsigma(a^*)\varsigma(a)\bigr)\in\cQ$ for all $a\in\cA_{\mathrm{loc}}$ and as characters are $*$-homomorphic,
    \begin{align}\label{eq:ProofCharacterSpectralGapIneq}
        -\frac{\mathrm{i}}{2}\mathcal{K}\bigl(\varsigma(a^*\delta_{\familyH}(a)-\delta_{\familyH}(a^*)a)\bigr)
        \geq \gamma\Bigl(\mathcal{K}\bigl(\varsigma(a^*a)\bigr)-\abs{\mathcal{K}\bigl(\varsigma(a)\bigr)}^2\Bigr)
    \end{align}
    for all $a\in \cA_{\mathrm{loc}}$, so $\omega$ also satisfies \cref{eq:GappedSystemPropertyHermitian} on $\cA_{\mathrm{loc}}$.
    Equivalently, $\omega$ satisfies
    \begin{align}
        \omega\bigl(a^* [H^{\Lambda(n)}, a] \bigr) \geq \gamma \bigl( \omega(a^*a) - \abs{\omega(a)}^2 \bigr),
    \end{align}
    for all $a \in \cA_{\mathrm{loc}}$, that is, it fulfills \cref{eq:GappedSystemProperty} on $\cA_{\mathrm{loc}}$.

    Now, with $\cA_{\mathrm{loc}} = \bigcup_{n \in \mathds{N}} \cA_{\Lambda(n)}$ and $\familyH = \{H^{\Lambda(n)}\}_n$, following the same steps as in \cref{sec:AlgebraicDynamicalSystem}, we recover the desired dynamical system $(\cA=\overline{\cA_{\mathrm{loc}}}, \tau_\familyH, \delta_\familyH)$.
    As $\cA_{\mathrm{loc}}$ inherits the $C^*$-norm from $\cA$, the positivity and normalization of $\omega$ imply that $\omega$ is norm-bounded on $\cA_{\mathrm{loc}}$, and hence $\omega$ extends by continuity to a state on the quasi-local algebra $\cA$, which we again denote by $\omega$.
    Since $\cA_{\mathrm{loc}}$ is a core for $\delta_\familyH$, the extension $\omega$ satisfies \cref{eq:GappedSystemProperty} for all $a$ in the domain of $\delta_{\familyH}$.

    To sum up the proof of the convergence statement, we recover a KMS ground state $\omega$ on the dynamical system $(\cA, \tau_{\familyH}, \delta_{\familyH})$ admitting a locally non-degenerate bulk spectral gap at least $\gamma$, proving the ``if and only if'' statement of the theorem.

    Now, we prove the ``furthermore'' statement.
    Since level-$(n,d)$ feasible sets are nested as $n, d$ increase, $\alpha^{\min}_{n, d}$ is monotonically nondecreasing and $\alpha^{\max}_{n, d}$ is monotonically nonincreasing, hence both converge.
    
    Every $\gamma$-gapped KMS ground state of $(\cA, \tau_{\familyH}, \delta_{\familyH})$ restricts to a feasible solution at every level.
    So, by optimality and taking $n, d \to \infty$, one has
    \begin{align}
        \lim_{n,d \to \infty} \alpha^{\min}_{n, d} \leq \langle O \rangle_{\min} \leq \langle O \rangle_{\max} \leq \lim_{n,d \to \infty} \alpha^{\max}_{n, d}.
    \end{align}
    
    It remains to show $\lim_{n,d \to \infty} \alpha^{\max}_{n, d} \leq \langle O \rangle_{\max}$.
    By the convergence proof, the hierarchy admits a limiting state $\omega$ on $(\cA, \tau_{\familyH}, \delta_{\familyH})$, and therefore
    \begin{align}
        \lim_{n,d \to \infty} \alpha^{\max}_{n, d} = \mathcal{L}(\varsigma(O)) \leq \mathcal{K}(\varsigma(O)) = \omega(O) \leq \langle O \rangle_{\max}.
    \end{align}
    The first equality is by construction, the first inequality follows from the choice of $\mathcal{K}$ (see the paragraph below \cref{eq:KadisonDuboisDecomposition}), the equality $\cK(\varsigma(O))=\omega(O)$ follows from the definition of $\omega$, and the final inequality follows from the definition of $\langle O\rangle_{\max}$.
    The case of minimization is analogous.

    Finally, we prove \cref{eq:HierarchyGapConvergence}.
    Write $F_{n,d}\coloneqq\{\gamma\geq0\mid\text{the level-}(n,d)\text{ relaxation is feasible with parameter }\gamma\}$, so that $\gamma_{n,d}=\sup F_{n,d}$.
    Each $F_{n,d}$ is an interval containing $0$.
    Indeed, set $b\coloneqq q-\varsigma(q)$, so that $\varsigma(b^*b)=\varsigma(q^*q)-\varsigma(q^*)\varsigma(q)$ by $\varsigma(\id)=1$ and the fact that $\varsigma(q)$ is a commuting symbol.
    Suppose now that $\mathcal{L}^d_n$ is feasible at $\gamma$ and let $0\leq\gamma'\leq\gamma$.
    Then for every $q\in\fatS$,
    \begin{equation}\label{eq:FeasibilityMonotoneInThreshold}
        \begin{aligned}
            &\mathcal{L}^d_n\Bigl(-\frac{\mathrm{i}}{2}\varsigma\bigl(q^*\delta_{\familyH}(q)-\delta_{\familyH}(q^*)q\bigr)-\gamma'\bigl(\varsigma(q^*q)-\varsigma(q^*)\varsigma(q)\bigr)\Bigr)\\
            &\qquad=\mathcal{L}^d_n\Bigl(-\frac{\mathrm{i}}{2}\varsigma\bigl(q^*\delta_{\familyH}(q)-\delta_{\familyH}(q^*)q\bigr)-\gamma\bigl(\varsigma(q^*q)-\varsigma(q^*)\varsigma(q)\bigr)\Bigr)+(\gamma-\gamma')\,\mathcal{L}^d_n\bigl(\varsigma(b^*b)\bigr)\geq 0,
        \end{aligned}
    \end{equation}
    because the first term is nonnegative by feasibility at $\gamma$, $\gamma-\gamma'\geq0$, and $\varsigma(b^*b)\in\cQ_n^{d,\mathrm{pos}}$.
    So $\mathcal{L}^d_n$ is feasible at $\gamma'$.

    By the necessity direction, $\cG_\gamma\neq\emptyset$ implies $\gamma\in F_{n,d}$ for every $(n,d)$, so $\gamma_{n,d}\geq\Delta_{\mathrm{bulk}}$.
    By the nesting $\cQ_n^d\subset\cQ_{n'}^{d'}$ for $n\leq n'$ and $d\leq d'$, we have $F_{n',d'}\subset F_{n,d}$, so the $\gamma_{n,d}$ are nonincreasing and converge to some $\gamma_\infty\geq\Delta_{\mathrm{bulk}}$.
    Suppose $\gamma_\infty>\Delta_{\mathrm{bulk}}$ and pick $\gamma$ with $\Delta_{\mathrm{bulk}}<\gamma<\gamma_\infty$.
    Since $F_{n,d}$ is an interval with $\sup F_{n,d}=\gamma_{n,d}\geq\gamma_\infty>\gamma$, we get $\gamma\in F_{n,d}$ for every $(n,d)$, so the sufficiency direction gives $\cG_\gamma\neq\emptyset$, contradicting the definition of $\Delta_{\mathrm{bulk}}$.
    Hence $\gamma_\infty=\Delta_{\mathrm{bulk}}$.
\end{proof}

\subsection{Semidecision procedure and further extensions}\label{sec:SemidecisionExtension}
The completeness theorem gives a semidecision procedure for excluding a proposed bulk spectral gap $\gamma>0$.
If a KMS ground state with locally non-degenerate bulk gap at least $\gamma$ exists, every level is feasible.
\cref{eq:HierarchyGapConvergence} ensures that $\gamma_{n,d}$ converges to the bulk spectral gap $\Delta_{\mathrm{bulk}}$.
Conversely, if no such state exists, \cref{thm:CompletenessSDPHierarchy} guarantees infeasibility at a finite level.

\begin{algorithm}[H]
    \caption{Semidecision of exclusion of a bulk-gap threshold}
    \label{alg:SpectralGapSemi}
    \begin{algorithmic}[1]
        \REQUIRE Dynamical system $(\cA,\tau_{\familyH},\delta_{\familyH})$ and a threshold $\gamma>0$ with interaction range $l$ and uniformly bounded degree $\deg(H)$
        \ENSURE Returns \textbf{excluded} if no KMS ground state has locally non-degenerate bulk gap at least $\gamma$; otherwise runs indefinitely
        \STATE Set $n=l+1$ and $d=\max\{1,\lceil\deg(H)/2\rceil\}$
        \LOOP
            \STATE Test feasibility of the level-$(n,d)$ SDP in \cref{def:GappedRelaxHierarchyStatePoly} with threshold $\gamma$ and $O=1$
            \IF{the SDP is infeasible}
                \RETURN \textbf{excluded}
            \ENDIF
            \STATE Increase both $n$ and $d$ by $1$
        \ENDLOOP
    \end{algorithmic}
\end{algorithm}

In particular, \cref{alg:SpectralGapSemi} terminates for every $\gamma>\Delta_{\mathrm{bulk}}$.
Feasibility at any finite level does not certify the existence of a KMS ground state with gap at least $\gamma$, so this procedure does not decide both outcomes in finite time.
In \cref{sec:BulkGapUndecidability}, we prove that deciding whether $\Delta_{\mathrm{bulk}}>0$ is undecidable and that no general algorithm produces convergent lower bounds.

\begin{remark}[Scope of Construction]\label{rem:ScopeOfConstruction}
The proof does not rely on the specific lattice realization of the quasi-local algebra beyond the existence of
(i) a distinguished dense $*$-subalgebra $\cA_{\mathrm{loc}}\subset \cA$ (typically an increasing union of finite-dimensional $C^*$-subalgebras),
(ii) a strongly continuous dynamics $\tau$ on $\cA$ with a generator $\delta$ whose dense domain contains $\cA_{\mathrm{loc}}$ as a core, and
(iii) an archimedean constraint on the local generators.
Whenever these structural assumptions hold, the corresponding variant of \cref{def:GappedRelaxHierarchyStatePoly} leads to an analogous \cref{thm:CompletenessSDPHierarchy}.
In particular, the locally non-degenerate bulk spectral gap inequality can be replaced by any family of constraints that is expressible in the framework of state polynomial optimization, in order to tackle a specific class of problems or achieve better convergence behaviour.

As an example beyond square lattices, we adapt our SDP hierarchy to the kagome lattice in \cref{sec:KagomeHeisenberg}.
\end{remark}

\begin{remark}[Hierarchy with local symmetry constraints]\label{rem:SymmetryRestrictedHierarchy}
    Let $\alpha$ preserve $\cA_{\mathrm{loc}}$ and be either a complex-linear $*$-automorphism commuting with $\tau_{\familyH}^t$, or a conjugate-linear $*$-automorphism satisfying $\alpha\circ\tau_{\familyH}^t=\tau_{\familyH}^{-t}\circ\alpha$.
    In the first case, a state $\omega$ is $\alpha$-symmetric if $\omega(\alpha(a))=\omega(a)$; in the second case, the condition is $\omega(\alpha(a))=\overline{\omega(a)}$.

    To impose these conditions in the state-polynomial hierarchy, for each local observable $a$ define
    \begin{align}\label{eq:SymmetryStatePolynomial}
        s_a\coloneqq
        \begin{cases}
            \varsigma(\alpha(a))-\varsigma(a), & \text{if $\alpha$ is complex-linear},\\
            \varsigma(\alpha(a))-\varsigma(a)^*, & \text{if $\alpha$ is conjugate-linear}.
        \end{cases}
    \end{align}
    At level $(n,d)$, add the linear constraints
    \begin{align}\label{eq:SymmetryMomentConstraint}
        \mathcal L_n^d(s_a^*s_a)=0
    \end{align}
    whenever $s_a^*s_a\in\thinS_{\Lambda(n)}^{2d}$.
    These constraints are necessary for every $\alpha$-symmetric state.
    In the completeness proof, they give $\int|\mathcal K(s_a)|^2\,d\mu(\mathcal K)=0$, so every character in the support of $\mu$ satisfies $\mathcal K(s_a)=0$.
    Thus the extracted state is $\alpha$-symmetric, and the completeness statement holds with $\cG_\gamma$ restricted to its $\alpha$-symmetric states.

    Such constraints restrict the hierarchy to a specified symmetry class and can reduce the computational cost.
\end{remark}

\section{Undecidability of the bulk spectral gap}\label{sec:BulkGapUndecidability}
The SDP hierarchy of \cref{sec:CompleteSDPHierarchyMainSec} gives convergent upper bounds on the bulk spectral gap: for every $\gamma>\Delta_{\mathrm{bulk}}$, infeasibility at a finite level excludes that threshold.
We now prove that deciding whether $\Delta_{\mathrm{bulk}}>0$ is undecidable and that no general algorithm produces convergent lower bounds.
The result holds even under the promise that the system either admits a KMS ground state with locally non-degenerate bulk gap at least $4$ or is strongly bulk-gapless in the sense of \cref{def:GaplessAlgebraic}.
Together with our semidecision procedure, this implies that the promise problem $\Delta_{\mathrm{bulk}}=0$ (\textsc{yes}) versus $\Delta_{\mathrm{bulk}}\geq1$ (\textsc{no}) is $\mathsf{RE}$-complete, and its complement is $\mathsf{coRE}$-complete.
It also follows that no general algorithm producing convergent upper bounds, including ours, admits a computable convergence rate.
Finally, the same construction gives uncomputability results about the KMS ground states themselves.
Deciding whether an explicitly given product state is a KMS ground state is $\mathsf{coRE}$-complete.
No algorithm computes, from the local interaction matrices, a finite region on which the bulk-gap inequality already certifies that property.
Convergent certified bounds on ground-state expectation values also admit no computable convergence rate (\cref{cor:KMSGroundStateComputability}).

\begin{theorem}[Undecidability of the bulk spectral gap]\label{thm:BulkGapUndecidability}
    For a fixed universal Turing machine $M$, there exists a computable map $m \mapsto \familyH(m) \coloneqq \{\widetilde{H}^I(m)\}_{I \subset \mathds{Z}}$ from inputs $m\in\mathds{N}$ to families of translationally invariant nearest-neighbour Hamiltonians on the 1D chain $\mathds{Z}$, with fixed local dimension, uniformly bounded local terms, and algebraic matrix entries.
    The dynamics induced by $\familyH(m)$ admits at least one KMS ground state, and the following statements hold:
    \begin{enumerate}[label=(\roman*)]
        \item If $M(m)$ does not halt, there exists a KMS ground state with locally non-degenerate bulk gap at least $4$.
        \item If $M(m)$ halts, the dynamical system is strongly bulk-gapless in the sense of \cref{def:GaplessAlgebraic}; more precisely, $\Spectrum{H_\omega}=[0,\infty)$ for every KMS ground state $\omega$.
    \end{enumerate}
    
    Consequently,
    \begin{align}\label{eq:HaltingToGapinBulk}
        \begin{cases}
            M(m)\text{ halts} &\Longrightarrow\Delta_{\mathrm{bulk}}(\familyH(m))=0,\\
            M(m)\text{ does not halt} &\Longrightarrow\Delta_{\mathrm{bulk}}(\familyH(m))\geq1.
        \end{cases}
    \end{align}
    Deciding whether $\Delta_{\mathrm{bulk}}>0$ is therefore undecidable, even under the promise that either $\Delta_{\mathrm{bulk}}\geq1$ or the system is strongly bulk-gapless with $\Spectrum{H_\omega}=[0,\infty)$ for every KMS ground state $\omega$. 
    In particular, there is no general algorithm that, given the local interaction matrices, produces rational lower bounds $\ell_s$, $s\in\mathds{N}$, satisfying
    \begin{align}\label{eq:BulkGapLowerBoundAlgorithm}
        \ell_s\leq\Delta_{\mathrm{bulk}}\quad\text{for every }s,\qquad
        \lim_{s\to\infty}\ell_s=\Delta_{\mathrm{bulk}}.
    \end{align}
\end{theorem}

The aim of this section is constructing the proof of \cref{thm:BulkGapUndecidability}.
We first record its consequences on bulk spectral gaps.

\begin{corollary}[Computability and $\mathsf{RE}$-completeness of the bulk spectral gap]\label{cor:BulkGapNoConvergenceRate}
    Consider translationally invariant nearest-neighbour spin chains on $\mathds{Z}$, specified by local interaction matrices with exactly encoded algebraic entries.
    The promise problem $\Delta_{\mathrm{bulk}}=0$ (\textnormal{\textsc{yes}}) versus $\Delta_{\mathrm{bulk}}\geq1$ (\textnormal{\textsc{no}}) is $\mathsf{RE}$-complete.
    Equivalently, exchanging \textnormal{\textsc{yes}} and \textnormal{\textsc{no}} gives a $\mathsf{coRE}$-complete problem.
    These conclusions hold already for fixed local dimension and uniformly bounded local terms, even if the gapless instances are required to satisfy $\Spectrum{H_\omega}=[0,\infty)$ for every KMS ground state $\omega$.
    No general algorithm produces convergent lower bounds on $\Delta_{\mathrm{bulk}}$ or computes convergence rates for convergent upper bounds, including those from our hierarchy.
\end{corollary}
\begin{proof}
    Run \cref{alg:SpectralGapSemi} at threshold $\gamma=1/2$, using exact feasibility tests as in \cref{rem:ComputabilityQuantifierElimination}.
    If $\Delta_{\mathrm{bulk}}=0$, \cref{thm:CompletenessSDPHierarchy} guarantees infeasibility at some finite level, so the algorithm terminates.
    If $\Delta_{\mathrm{bulk}}\geq1$, every level is feasible at $\gamma=1/2$, so it runs indefinitely.
    Thus the \textsc{yes}-instances are semidecidable under the promise, proving membership in $\mathsf{RE}$.

    By \cref{eq:HaltingToGapinBulk}, the computable map $m\mapsto\familyH(m)$ sends halting inputs to \textsc{yes}-instances and nonhalting inputs to \textsc{no}-instances.
    This is a computable many-one reduction from the halting problem, proving $\mathsf{RE}$-hardness.
    The construction has fixed local dimension and uniformly bounded local terms, and its halting instances satisfy $\Spectrum{H_\omega}=[0,\infty)$ for every KMS ground state $\omega$.
    Hence completeness holds with the stated restrictions.
    Exchanging \textsc{yes} and \textsc{no} gives $\mathsf{coRE}$-completeness.

    Lastly, suppose there exists an algorithm producing rational upper bounds $u_j\to\Delta_{\mathrm{bulk}}$ with a computable convergence rate, and apply it to the systems constructed in \cref{thm:BulkGapUndecidability}.
    Given $k\in\mathds{N}$, this convergence rate would compute an index $N_k$ at which the upper bound $u_{N_k}$ satisfies $\Delta_{\mathrm{bulk}} \leq u_{N_k} \leq \Delta_{\mathrm{bulk}} + 2^{-k}$.
    Then $\ell_k\coloneqq u_{N_k}-2^{-k}$, satisfying $\Delta_{\mathrm{bulk}}-2^{-k}\leq\ell_k\leq\Delta_{\mathrm{bulk}}$, are computable convergent rational lower bounds and contradicting \cref{thm:BulkGapUndecidability}.
\end{proof}

Strong bulk gaplessness requires $[0,c]\subset\Spectrum{H_\omega}$ for one $c>0$ common to all KMS ground states.
It is the infinite-volume spectral analogue of strong gaplessness under open boundary condition by \cite[Definition~2]{cubitt2022undecidability} (see \cref{def:StrongGaplessCPW15}), where the finite-volume spectra become dense in an interval of fixed length above their ground-state energies under a prescribed boundary condition.
Our halting conclusion proves more than the $[0,c]$-condition defining strong bulk gaplessness: every KMS ground state has in fact bulk spectrum $[0,\infty)$.

This conclusion also clarifies the role of boundary conditions.
As we shall recall in \cref{sec:FiniteVolumeSpectralGap,rem:RelationLimitGroundAndKMS}, a common way to obtain KMS ground states is to choose ground states of finite-volume Hamiltonians with prescribed boundary conditions and take a limit as the system size tends to infinity.
In the halting case, every such limit state has bulk spectrum $[0,\infty)$.
The same holds for generalized boundary conditions whenever the resulting state is a KMS ground state of the same dynamics.
We discuss the comparison with gaplessness under a fixed boundary condition further in \cref{sec:DiscussionSemiDecidability}; our construction has the opposite halting/gap direction to \cref{eq:HaltingToGapCPW15}, the main result of~\cite{cubitt2022undecidability}.

\subsection{Construction of the interaction}\label{sec:BulkUndecidabilityConstruction}
The proof modifies the construction of~\cite{bausch2020undecidability}.
We enlarge their system by tensoring its single-site algebra with a qubit algebra $M_2(\mathds{C})$ and adding an XY interaction on the qubits~\cite{araki1985ground}.
We show that this XY interaction has a unique KMS ground state $\nu$ and that its bulk spectrum satisfies $\Spectrum{H_\nu}=[0,\infty)$.
We also add an extra vector $\ket{\perp}$ to the original single-site Hilbert space, on which the interaction vanishes while the XY interaction still acts on the qubit.
This allows us to prove that, in the halting case, every KMS ground state of our construction inherits the spectrum $[0,\infty)$ from $H_\nu$.

Consider the Hamiltonians on a 1D chain with single-site algebra $M_r(\mathds{C})$ obtained in~\cite[Lemma~23, using Theorem~22]{bausch2020undecidability}.
They are defined by hermitian matrices $h^{(1)}(m)\in M_r(\mathds{C})$ and $h^{(2)}(m)\in M_{r^2}(\mathds{C})$, with algebraic entries computably determined by $m$.
The Hamiltonians on $\{1,\dots,N\}$ with open boundary conditions,
\begin{align}\label{eq:BCLPShiftedHamiltonian}
    K_N(m)\coloneqq\sum_{i=1}^{N}h_i^{(1)}(m)+\sum_{i=1}^{N-1}h_{i,i+1}^{(2)}(m),\qquad N\geq1,
\end{align}
satisfy
\begin{align}\label{eq:BCLPShiftedEnergyDichotomy}
    \begin{cases}
        M(m)\text{ does not halt} &\Longrightarrow K_N(m)\geq\id\quad\text{for every }N\geq1,\\
        M(m)\text{ halts} &\Longrightarrow \lambda_0(K_N(m))\longrightarrow-\infty\quad\text{as }N\to\infty.
    \end{cases}
\end{align}
The dimension $r$ is independent of $m$, and there is a constant $C>0$, independent of $m$, such that $\norm{h^{(1)}(m)}\leq C$ and $\norm{h^{(2)}(m)}\leq C$ for every $m$.

Let $\ket{\perp}$ be a unit vector.
At each site of the infinite chain $\mathds{Z}$, take the Hilbert space
\begin{align}\label{eq:BulkReductionLocalSpace}
    \widetilde{\cH}\coloneqq\cHA\oplus\cHB,\qquad
    \cHA\coloneqq\mathds{C}\ket{A},\qquad
    \cHB\coloneqq(\mathds{C}^r\oplus\mathds{C}\ket{\perp})\otimes\mathds{C}^2,
\end{align}
where $\ket{A}$ is a unit vector.
Let $\{\ket{0},\ket{1}\}$ be an orthonormal basis of the qubit space $\mathds{C}^2$.
Let $P_A,P_B$ be the projections onto $\cHA,\cHB$, respectively.
We define $*$-homomorphisms $\iota_r:M_r(\mathds{C})\to B(\widetilde{\cH})$ and $\iota_2:M_2(\mathds{C})\to B(\widetilde{\cH})$ embedding the original single-site algebra and the new qubit algebra, respectively, by
\begin{align}\label{eq:BulkReductionEmbeddings}
    \iota_r(x)\coloneqq0\oplus\bigl((x\oplus0)\otimes\id_2\bigr),\qquad
    \iota_2(y)\coloneqq0\oplus(\id_{r+1}\otimes y).
\end{align}
Here $x\oplus0$ acts as $x$ on $\mathds{C}^r$ and as zero on $\mathds{C}\ket{\perp}$.
In particular, $\iota_2(\id_2)=P_B$, whereas $\iota_r(\id_r)$ is the projection onto $\mathds{C}^r\otimes\mathds{C}^2\subset\cHB$.

In the two-qubit algebra $M_2(\mathds{C})\otimes M_2(\mathds{C})$, define the XY interaction by the projection
\begin{equation}\label{eq:BulkReductionQubitInteraction}
    \begin{aligned}
        q&\coloneqq\ketbra{11}{11}+\frac12(\ket{01}-\ket{10})(\bra{01}-\bra{10})\\
        &=\frac12\id_4-\frac14\bigl(X\otimes X+Y\otimes Y+Z\otimes\id_2+\id_2\otimes Z\bigr),
    \end{aligned}
\end{equation}
where $X,Y,Z$ are the Pauli matrices in the basis $\{\ket{0},\ket{1}\}$, with $Z=\ketbra{0}{0}-\ketbra{1}{1}$.
The second expression displays the $XX+YY$ exchange and the transverse field from which $q$ takes its name, with the constant $\tfrac12\id_4$ placing both $\ket{00}$ and $(\ket{01}+\ket{10})/\sqrt{2}$ at zero energy.
It is an instance of the class studied in~\cite{araki1974equivalence}.

Now we are ready to define our family of local Hamiltonians $\widetilde{H}(m)$ for undecidability.
We embed the interaction of~\cite{bausch2020undecidability} and the two-qubit projection $q$ to define the onsite and nearest-neighbour interactions on the chain $\mathds{Z}$, respectively, by
\begin{equation}\label{eq:BulkReductionLocalTerms}
    \begin{aligned}
        \widetilde h^{(1)}(m)&\coloneqq\iota_r(h^{(1)}(m)),\\
        \widetilde h^{(2)}(m)&\coloneqq(\iota_r\otimes\iota_r)(h^{(2)}(m))+(\iota_2\otimes\iota_2)(q)+2(P_A\otimes P_B+P_B\otimes P_A).
    \end{aligned}
\end{equation}
We call $\widetilde h^{(1)}(m)$ and $(\iota_r\otimes\iota_r)(h^{(2)}(m))$ the \emph{computational} terms, $(\iota_2\otimes\iota_2)(q)$ the \emph{XY} term, and $2(P_A\otimes P_B+P_B\otimes P_A)$ the \emph{$A$--$B$ penalty} term.
The last term equals $2\id$ on $(\cHA\otimes\cHB)\oplus(\cHB\otimes\cHA)$ and zero on $(\cHA\otimes\cHA)\oplus(\cHB\otimes\cHB)$.
Thus each pair of neighbouring sites with one in $\cHA$ and the other in $\cHB$ contributes energy $2$, and we call such a pair an \emph{$A$--$B$ bond}.
For every $x\in M_r(\mathds{C})$ and $\ket{u}\in\mathds{C}^2$, we have $\iota_r(x)(\ket{\perp}\otimes\ket{u})=0$.
Choosing $\ket{\perp}$ at a site therefore switches off the computational interaction: $\iota_r(h^{(1)}(m))$ vanishes at that site, and $(\iota_r\otimes\iota_r)(h^{(2)}(m))$ vanishes between that site and either neighbour.
In contrast, $\iota_2(y)(\ket{\perp}\otimes\ket{u})=\ket{\perp}\otimes y\ket{u}$ for every $y\in M_2(\mathds{C})$.
Thus $\ket{\perp}$ leaves the XY interaction unchanged: the restriction of $(\iota_2\otimes\iota_2)(q)$ to $(\mathds{C}\ket{\perp}\otimes\mathds{C}^2)^{\otimes2}$ still acts as $q$ on the qubit space $\mathds{C}^2\otimes\mathds{C}^2$.
Finally, $\iota_r(x)\ket{A}=\iota_2(y)\ket{A}=0$, so the embedded onsite terms vanish on $\cHA$, and both embedded two-site interactions vanish on $\cHA\otimes\widetilde{\cH}$ and $\widetilde{\cH}\otimes\cHA$.

For a finite interval $I\subset\mathds{Z}$, set $\cA_I\coloneqq\bigotimes_{i\in I}B(\widetilde{\cH})$ and
\begin{align}\label{eq:BulkReductionFiniteHamiltonian}
    \widetilde H^I(m)\coloneqq\sum_{i\in I}\widetilde h_i^{(1)}(m)+\sum_{\{i,i+1\}\subset I}\widetilde h_{i,i+1}^{(2)}(m).
\end{align}
As in \cref{sec:AlgebraicDynamicalSystem}, let $\cA$ be the norm closure of $\cA_{\mathrm{loc}}\coloneqq\bigcup_L\cA_{\latticeL}$, with $\latticeL=\{-L,\dots,L\}$, and let $\familyH(m)\coloneqq\{\widetilde H^{\latticeL}(m)\}_L$ induce the dynamical system $(\cA,\tau_{\familyH(m)},\delta_{\familyH(m)})$.
For $a\in\cA_{\mathrm{loc}}$, its derivation is
\begin{align}\label{eq:BulkReductionDerivation}
    \delta_{\familyH(m)}(a)=\mathrm{i}\sum_{i\in\mathds{Z}}[\widetilde h_i^{(1)}(m),a]+\mathrm{i}\sum_{i\in\mathds{Z}}[\widetilde h_{i,i+1}^{(2)}(m),a],
\end{align}
where only finitely many commutators are nonzero.
We write $P_i^A,P_i^B,\iota_r(x)_i,\iota_2(y)_i$ for the corresponding single-site observables at $i$.
By \cref{eq:BulkReductionEmbeddings}, $P_A\iota_r(x)=\iota_r(x)P_A=0$ and $P_A\iota_2(y)=\iota_2(y)P_A=0$.
Hence each local interaction preserves the subspaces $\cHA$ and $\cHB$ at every site, so
\begin{align}\label{eq:BulkReductionSectorCommutation}
    [P_i^A,\widetilde h_j^{(1)}(m)]=[P_i^A,\widetilde h_{j,j+1}^{(2)}(m)]=0\qquad(i,j\in\mathds{Z}).
\end{align}

Let
\begin{align}\label{eq:BulkReductionProductState}
    \ket{A_I}\coloneqq\bigotimes_{i\in I}\ket{A},\qquad
    P_{A,I}\coloneqq\ketbra{A_I}{A_I},\qquad
    \omega_A(a)\coloneqq\sandwich{A_I}{a}{A_I}\quad(a\in\cA_I).
\end{align}
These consistent local restrictions define a product state $\omega_A$ on $\cA_{\mathrm{loc}}$, which extends by continuity to $\cA$.
We will also use the unit vector $\ket{\chi}\coloneqq\ket{\perp}\otimes\ket{0}\in\cHB$.
As above, $\iota_r(h^{(1)}(m))\ket{\chi}=0$, and $(\iota_r\otimes\iota_r)(h^{(2)}(m))$ annihilates $\ket{\chi}\otimes\ket{v}$ and $\ket{v}\otimes\ket{\chi}$ for every $\ket{v}\in\widetilde{\cH}$.
Also, $q\ket{00}=0$.

\subsection{Bulk spectrum of the XY interaction}\label{sec:XYSpectrum}
We first study the spectral properties of the XY interaction on the qubits.
Let $\cQ_I\coloneqq\bigotimes_{i\in I}M_2(\mathds{C})$ and let $\cQ$ be the norm closure of $\bigcup_L\cQ_{\latticeL}$.
The local Hamiltonians $D^I\coloneqq\sum_{\{i,i+1\}\subset I}q_{i,i+1}$ induce a dynamical system $(\cQ,\tau_q,\delta_q)$, with $\delta_q(b)=\mathrm{i}\sum_i[q_{i,i+1},b]$ for local $b$.
The state $\nu$ is the product state determined by
\begin{align}\label{eq:BulkReductionQubitVacuum}
    \ket{0_I}\coloneqq\bigotimes_{i\in I}\ket{0},\qquad
    \nu(b)\coloneqq\sandwich{0_I}{b}{0_I}\quad(b\in\cQ_I).
\end{align}
These consistent local restrictions define a product state $\nu$ on $\bigcup_L\cQ_{\latticeL}$, which extends by continuity to $\cQ$.

\begin{lemma}[Spectrum of the XY interaction]\label{lem:BulkReductionQubitSpectrum}
    The dynamical system $(\cQ,\tau_q,\delta_q)$ has the unique KMS ground state $\nu$.
    Its bulk Hamiltonian $H_\nu$ satisfies $\Spectrum{H_\nu}=[0,\infty)$.
\end{lemma}
\begin{proof}
    Let $X,Y,Z$ be the Pauli matrices in the basis $\{\ket{0},\ket{1}\}$, with $Z\ket{0}=\ket{0}$.
    Recall from \cref{eq:BulkReductionQubitInteraction},
    \begin{align}\label{eq:BulkReductionQubitPauliForm}
        q=\frac12\id_4-\frac14(Z\otimes\id_2+\id_2\otimes Z+X\otimes X+Y\otimes Y).
    \end{align}
    Hence, for every local qubit observable $b$,
    \begin{align}\label{eq:BulkReductionXYDerivation}
        \delta_q(b)=-\frac{\mathrm{i}}4\sum_{j\in\mathds{Z}}[X_j X_{j+1}+Y_j Y_{j+1}+2Z_j,b].
    \end{align}
    This is the XY dynamics of~\cite{araki1985ground} with coupling $J=1/4$, anisotropy $\gamma=0$, and field $\lambda=1$ in their convention.
    By~\cite[Theorem~1]{araki1985ground}, the KMS ground state is unique for $|\lambda|\geq1$, so $(\cQ,\tau_q,\delta_q)$ has a unique KMS ground state.
    We show that this unique state is $\nu$.
    Take a local $b$ and a finite interval $I$ containing its support and the neighbouring sites.
    Since $D^I\geq0$ and $D^I\ket{0_I}=0$, we obtain
    \begin{align}\label{eq:BulkReductionQubitGroundCondition}
        -\mathrm{i}\nu(b^*\delta_q(b))=\sandwich{0_I}{b^*[D^I,b]}{0_I}=\sandwich{0_I}{b^*D^I b}{0_I}\geq0.
    \end{align}
    The local algebra is a core for $\delta_q$, so $\nu$ is a KMS ground state and hence the unique one.

    Let $(\cH_\nu,\pi_\nu,\ket{\Omega_\nu})$ be its GNS representation.
    We first show that the spectrum of the bulk Hamiltonian $H_\nu$ contains the interval $[0,2]$.
    For $j\in\mathds{Z}$, define the one-qubit excitation vector
    \begin{align}\label{eq:BulkReductionOneExcitation}
        \ket{e_j}\coloneqq\pi_\nu(\ketbra{1}{0}_j)\ket{\Omega_\nu}.
    \end{align}
    By \cref{eq:BulkReductionQubitVacuum}, these vectors are orthonormal.
    They belong to the domain of $H_\nu$ by \cref{eq:BulkHamiltonianProperty}.
    Only $q_{j-1,j}$ and $q_{j,j+1}$ contribute to $\delta_q(\ketbra{1}{0}_j)$, since all other terms commute with $\ketbra{1}{0}_j$.
    Using $q\ket{10}=\tfrac12(\ket{10}-\ket{01})$, $q\ket{01}=\tfrac12(\ket{01}-\ket{10})$, and $q\ket{00}=0$ in \cref{eq:BulkHamiltonianProperty} gives
    \begin{align}\label{eq:BulkReductionOneExcitationAction}
        H_\nu \ket{e_j}=\ket{e_j}-\frac12\ket{e_{j-1}}-\frac12\ket{e_{j+1}}.
    \end{align}
    Thus the XY interaction couples the excitation at $j$ to each neighbouring site with coefficient $-1/2$, allowing it to propagate along the chain under time evolution.
    For any integer $R\geq1$ and $k\in[0,\pi]$, define the unit vector
    \begin{align}\label{eq:BulkReductionOneExcitationPacket}
        \ket{\psi_{R,k}}\coloneqq\frac1{\sqrt{2R+1}}\sum_{j=-R}^{R}\mathrm{e}^{\mathrm{i}kj}\ket{e_j},
    \end{align}
    which is analogous to a Fourier expansion.
    Applying \cref{eq:BulkReductionOneExcitationAction} and using $2\cos k=\mathrm{e}^{\mathrm{i}k}+\mathrm{e}^{-\mathrm{i}k}$ gives
    \begin{align}\label{eq:BulkReductionOneExcitationBoundaryResidual}
        (H_\nu-(1-\cos k))\ket{\psi_{R,k}}=\frac{\mathrm{e}^{-\mathrm{i}k(R+1)}\ket{e_{-R}}+\mathrm{e}^{\mathrm{i}k(R+1)}\ket{e_R}-\mathrm{e}^{-\mathrm{i}kR}\ket{e_{-R-1}}-\mathrm{e}^{\mathrm{i}kR}\ket{e_{R+1}}}{2\sqrt{2R+1}}.
    \end{align}
    The four vectors on the right are orthonormal, so
    \begin{align}\label{eq:BulkReductionOneExcitationResidual}
        \norm{(H_\nu-(1-\cos k))\ket{\psi_{R,k}}}=\frac1{\sqrt{2R+1}}\longrightarrow0\qquad(R\to\infty).
    \end{align}
    Since $\norm{\ket{\psi_{R,k}}}=1$, this implies $1-\cos k\in\Spectrum{H_\nu}$: otherwise $H_\nu-(1-\cos k)$ would have a bounded inverse, contradicting \cref{eq:BulkReductionOneExcitationResidual}.
    Hence $\{ 1-\cos k\mid k\in[0,\pi] \}=[0,2]\subset\Spectrum{H_\nu}$.

    We next show that sums of the values $1-\cos k$ also belong to $\Spectrum{H_\nu}$ by applying the excitation operators $\ketbra{1}{0}$ at sites in separated intervals.
    Let $n\geq1$.
    For integers $j_1<\cdots<j_n$, define
    \begin{align}\label{eq:BulkReductionSeveralExcitations}
        \ket{e_{j_1,\dots,j_n}}\coloneqq\pi_\nu\!\left(\prod_{p=1}^n\ketbra{1}{0}_{j_p}\right)\ket{\Omega_\nu}.
    \end{align}
    This vector has an excitation at each of the sites $j_1,\dots,j_n$.
    These vectors are orthonormal.
    If $j_{p+1}-j_p\geq3$ for $1\leq p<n$, no nearest-neighbour term acts on two excited sites.
    Consequently, as in \cref{eq:BulkReductionOneExcitationAction},
    \begin{align}\label{eq:BulkReductionSeveralExcitationsAction}
        H_\nu\ket{e_{j_1,\dots,j_n}}=n\ket{e_{j_1,\dots,j_n}}-\frac12\sum_{p=1}^n\left(\ket{e_{j_1,\dots,(j_p-1),\dots,j_n}}+\ket{e_{j_1,\dots,(j_p+1),\dots,j_n}}\right).
    \end{align}
    We use this identity to construct approximate eigenvectors whose energies are sums of the single-excitation energies.
    
    Fix $k_1,\dots,k_n\in[0,\pi]$.
    For each integer $R\geq1$, choose ordered intervals $I_{1,R},\dots,I_{n,R}$ of $2R+1$ sites each, with $\min I_{p+1,R}-\max I_{p,R}\geq3$.
    Define
    \begin{align}\label{eq:BulkReductionSeveralExcitationsPacket}
        \ket{\Psi_R^{(n)}}\coloneqq\frac1{(2R+1)^{n/2}}\sum_{\substack{j_1\in I_{1,R},\dots,j_n\in I_{n,R}}}\mathrm{e}^{\mathrm{i}\sum_{p=1}^n k_p j_p}\ket{e_{j_1,\dots,j_n}}.
    \end{align}
    This vector is obtained by applying a local observable to $\ket{\Omega_\nu}$, so it belongs to the domain of $H_\nu$ by \cref{eq:BulkHamiltonianProperty}.
    Orthogonality gives $\norm{\ket{\Psi_R^{(n)}}}=1$.
    
    For $1\leq p\leq n$, define
    \begin{equation}\label{eq:BulkReductionSeveralExcitationsBoundaryContribution}
        \begin{aligned}
            \ket{\eta_{p,R}}
            &\coloneqq\frac1{(2R+1)^{n/2}}\sum_{\substack{j_1\in I_{1,R},\dots,j_n\in I_{n,R}}}\mathrm{e}^{\mathrm{i}\sum_{s=1}^n k_sj_s} \left(\cos k_p\,\ket{e_{j_1,\dots,j_n}} -\frac12\ket{e_{j_1,\dots,j_p-1,\dots,j_n}}-\frac12\ket{e_{j_1,\dots,j_p+1,\dots,j_n}}\right).
        \end{aligned}
    \end{equation}
    It follows from \cref{eq:BulkReductionSeveralExcitationsAction} that
    \begin{align}\label{eq:BulkReductionSeveralExcitationsResidualDecomposition}
        \left(H_\nu-\sum_{p=1}^n(1-\cos k_p)\right)\ket{\Psi_R^{(n)}}=\sum_{p=1}^n\ket{\eta_{p,R}}.
    \end{align}
    By the same argument as in \cref{eq:BulkReductionOneExcitationBoundaryResidual,eq:BulkReductionOneExcitationResidual}, using orthogonality of the excitation vectors, we obtain
    \begin{align}\label{eq:BulkReductionSeveralExcitationsBoundaryNorm}
        \norm{\ket{\eta_{p,R}}}^2=\frac{4(2R+1)^{n-1}}{4(2R+1)^n}=\frac1{2R+1}.
    \end{align}
    
    The triangle inequality now gives
    \begin{align}\label{eq:BulkReductionSeveralExcitationsResidual}
        \norm{\left(H_\nu-\sum_{p=1}^n(1-\cos k_p)\right)\ket{\Psi_R^{(n)}}}\leq\sum_{p=1}^n\norm{\ket{\eta_{p,R}}}=\frac{n}{\sqrt{2R+1}}\longrightarrow0\qquad(R\to\infty).
    \end{align}
    
    It follows that $\sum_{p=1}^n(1-\cos k_p)\in\Spectrum{H_\nu}$, and hence $[0,2n]\subset\Spectrum{H_\nu}$ for every $n\geq1$.
    Taking the union over $n$ and using $H_\nu\geq0$ proves $\Spectrum{H_\nu}=[0,\infty)$.
\end{proof}

\subsection{Properties of the nonhalting and halting KMS ground states}\label{sec:PropertyKMSLemma}
We make two observations about the KMS ground states of the dynamics induced by the family of local Hamiltonians $\familyH(m)=\{\widetilde{H}^I\}_{I \subset \mathds{Z}}$ in the nonhalting and halting cases.

\begin{lemma}[Nonhalting case]\label{lem:BulkReductionNonhaltingGap}
    If $M(m)$ does not halt, then $\omega_A$ is a KMS ground state of the dynamics induced by $\familyH(m)=\{\widetilde{H}^I\}_{I \subset \mathds{Z}}$, with locally non-degenerate bulk gap at least $4$.
\end{lemma}
\begin{proof}
    Fix $m$; throughout this proof, we suppress $(m)$ for simplicity in all Hamiltonians, interaction matrices, and the family $\familyH$.
    Suppose that $M(m)$ does not halt.
    We show that the product state $\omega_A$ induced by the vectors $\ket{A_I}$ is a KMS ground state with locally non-degenerate bulk gap at least $4$.
    The idea is that a local excitation orthogonal to $\ket{A_I}$ must involve $\cHB$ at some site, while sufficiently distant sites remain in the state $\ket{A}$.
    Therefore, there must be two $A$--$B$ bonds, each contributing $2$ through the $A$--$B$ penalty term $2(P_i^A P_{i+1}^B+P_i^B P_{i+1}^A)$.
    The nonhalting bound $K_N\geq\id$ in \cref{eq:BCLPShiftedEnergyDichotomy} and positivity of $q$ ensure that the remaining terms have nonnegative total contribution.
    See \cref{fig:BulkLocalReplacementNonhalting} in the main text for a pictorial illustration.

    Formally, fix $a\in\cA_{\mathrm{loc}}$ and choose $I=\{u,\dots,v\}$ such that $a$ is supported on $J = \{u+1,\dots,v-1\}$.
    Then $\delta_{\familyH}(a)=\mathrm{i}[\widetilde H^I,a]$.
    We separate the component of $a\ket{A_I}$ along $\ket{A_I}$ from its orthogonal component:
    \begin{align}\label{eq:BulkReductionOrthogonalExcitation}
        a\ket{A_I}=\sandwich{A_I}{a}{A_I}\ket{A_I}+\ket{\psi_a},\qquad
        \ket{\psi_a}\coloneqq(\id-P_{A,I})a\ket{A_I}.
    \end{align}
    Since $a$ is supported on $J$, we have
    \begin{align}\label{eq:BulkReductionExcitationEndpoints}
        \ket{\psi_a}=\ket{A}_u\otimes\bigl((\id-P_{A,J})a\ket{A_J}\bigr)\otimes\ket{A}_v,\qquad P_{A,I}\ket{\psi_a}=0.
    \end{align}

    We first check the XY contribution.
    Since $q$ and its embedding are projections, every vector $\ket{\xi}\in\widetilde{\cH}^{\otimes|I|}$ satisfies
    \begin{align}\label{eq:BulkReductionXYEnergyExpectation}
        \sandwich{\xi}{(\iota_2\otimes\iota_2)(q)_{i,i+1}}{\xi}=\norm{(\iota_2\otimes\iota_2)(q)_{i,i+1}\ket{\xi}}^2\geq0\qquad(u\leq i<v).
    \end{align}

    We next check the nonnegativity of the sum of the computational terms $\iota_r(h^{(1)})_i$ and $(\iota_r\otimes\iota_r)(h^{(2)})_{i,i+1}$.
    Suppose first that $\ket{\xi}$ lies in $\mathds{C}^r\otimes\mathds{C}^2$ at each site of one interval $J'\subset I$ of length $N$, and in $\cHA\oplus(\mathds{C}\ket{\perp}\otimes\mathds{C}^2)$ at each site outside $J'$.
    By \cref{eq:BulkReductionEmbeddings}, only $\iota_r(h^{(1)})_i$ with $i\in J'$ and $(\iota_r\otimes\iota_r)(h^{(2)})_{i,i+1}$ with $i,i+1\in J'$ act nontrivially on $\ket{\xi}$.
    Their sum is $K_N$ on the $N$ copies of $\mathds{C}^r$ in $J'$, tensored with the identity on the qubit spaces $(\mathds{C}^2)^{\otimes N}$ and on all sites outside $J'$.
    Its expectation is therefore $\sandwich{\xi}{K_N\otimes\id}{\xi}\geq\norm{\ket{\xi}}^2\geq0$ by \cref{eq:BCLPShiftedEnergyDichotomy}.
    If instead the sites in $\mathds{C}^r\otimes\mathds{C}^2$ form several maximal consecutive intervals, the expectation is the sum of these nonnegative contributions, one from each interval; if there are no such sites, it is zero.

    By construction, $\iota_r(h^{(1)})_i$ and $(\iota_r\otimes\iota_r)(h^{(2)})_{i,i+1}$ preserve each direct-sum component obtained by choosing $\mathds{C}^r\otimes\mathds{C}^2$ or $\cHA\oplus(\mathds{C}\ket{\perp}\otimes\mathds{C}^2)$ at each site.
    On each such component, the sum over every maximal consecutive interval in $\mathds{C}^r\otimes\mathds{C}^2$ has nonnegative expectation by the preceding argument.
    The subspace $\cHA\oplus(\mathds{C}\ket{\perp}\otimes\mathds{C}^2)$ is annihilated by $\iota_r(h^{(1)})$, and $(\iota_r\otimes\iota_r)(h^{(2)})$ vanishes whenever either site lies in this subspace, so these terms contribute zero energy.
    For an arbitrary $\ket{\xi}$, the expectation is the sum of these nonnegative contributions, with no cross terms because the operators preserve the direct-sum components.
    Together with the XY contribution, \cref{eq:BulkReductionLocalTerms} gives
    \begin{align}\label{eq:BulkReductionBoundaryLowerBound}
        \widetilde H^I\geq2\sum_{i=u}^{v-1}\bigl(P_i^A P_{i+1}^B+P_i^B P_{i+1}^A\bigr).
    \end{align}

    Using $P_i^A+P_i^B=\id$ and commutation of projections at different sites, we obtain the following identities and inequalities on $\widetilde{\cH}^{\otimes|I|}$ using telescoping sums:
    \begin{equation}\label{eq:BulkReductionBoundaryGap}
        \begin{aligned}
            \id-P_{A,I}&=P_u^B+\sum_{j=u+1}^{v}\left(\prod_{i=u}^{j-1}P_i^A\right)P_j^B
            \leq P_u^B+\sum_{i=u}^{v-1}P_i^A P_{i+1}^B,\\
            \id-P_{A,I}&=P_v^B+\sum_{j=u}^{v-1}P_j^B\left(\prod_{i=j+1}^{v}P_i^A\right)
            \leq P_v^B+\sum_{i=u}^{v-1}P_i^B P_{i+1}^A.
        \end{aligned}
    \end{equation}
    Adding the two inequalities in \cref{eq:BulkReductionBoundaryGap}, multiplying by $2$, and using \cref{eq:BulkReductionBoundaryLowerBound} gives
    \begin{align}\label{eq:BulkReductionEndpointOperatorGap}
        \widetilde H^I \geq 4(\id-P_{A,I}) - 2(P_u^B+P_v^B).
    \end{align}
    By \cref{eq:BulkReductionExcitationEndpoints}, the expectations satisfy $\sandwich{\psi_a}{P_u^B}{\psi_a} = \sandwich{\psi_a}{P_v^B}{\psi_a} = 0$.
    Since also $P_{A,I}\ket{\psi_a}=0$ and $\widetilde H^I\ket{A_I}=0$, we obtain
    \begin{equation}\label{eq:BulkReductionProductStateGap}
        \begin{aligned}
            -\mathrm{i}\omega_A(a^*\delta_{\familyH}(a))
            &=\sandwich{A_I}{a^*\widetilde H^I a}{A_I}=\sandwich{\psi_a}{\widetilde H^I}{\psi_a}\\
            &\geq4\sandwich{\psi_a}{\id-P_{A,I}}{\psi_a}-2\sandwich{\psi_a}{P_u^B+P_v^B}{\psi_a}\\
            &=4\norm{\ket{\psi_a}}^2=4\bigl(\omega_A(a^*a)-\abs{\omega_A(a)}^2\bigr),
        \end{aligned}
    \end{equation}
    where the last equality follows from \cref{eq:BulkReductionOrthogonalExcitation}.
    Since $\cA_{\mathrm{loc}}$ is a core for $\delta_{\familyH}$, this inequality extends by continuity to its full domain.
    By \cref{eq:KMSGroundCondition,prop:GappedSystemProperty}, $\omega_A$ is a KMS ground state with locally non-degenerate bulk gap at least $4$.
\end{proof}

We recap the underlying idea for proving \cref{lem:BulkReductionNonhaltingGap}: for a local excitation orthogonal to the product state induced by $\ket{A}$, the two boundaries between $\cHA$ and $\cHB$ correspond to two $A$--$B$ bonds, costing at least $4$ in energy. 
Meanwhile the nonhalting bound $K_N\geq\id$ and positivity of the XY interaction ensure that the remaining terms have nonnegative total contribution.
We also note that $\omega_A$ is not the only KMS ground state: the product state with $\ket{\chi}$ at every site is also one.

\begin{lemma}[Halting case]\label{lem:BulkReductionHaltingSupport}
    If $M(m)$ halts, then every KMS ground state $\omega$ of the dynamics induced by $\familyH(m)=\{\widetilde{H}^I\}_{I \subset \mathds{Z}}$ satisfies $\omega(P_i^A)=0$ for every $i\in\mathds{Z}$.
\end{lemma}
\begin{proof}
    Fix $m$; throughout this proof, we suppress $(m)$ for simplicity in all Hamiltonians, interaction matrices, and the family $\familyH$.
    Suppose that $M(m)$ halts, and let $\omega$ be any KMS ground state.

    \proofstep{1}{Physical intuition.}
    We first explain what it would mean to have $\omega(P_i^A)>0$.
    On any finite interval $I$ containing $i$, the restriction of $\omega$ is represented by a density operator $\rho_I$ on $\widetilde{\cH}^{\otimes|I|}$, with
    \begin{align}\label{eq:BulkReductionLocalDensityWeight}
        \omega(x)=\operatorname{Tr}(\rho_I x)\quad(x\in\cA_I),\qquad
        \omega(P_i^A)=\operatorname{Tr}(P_i^A\rho_I P_i^A).
    \end{align}
    A positive value of $\omega(P_i^A)$ would mean that $\rho_I$ has a nonzero component in $\cHA$ at site $i$.
    If this is the case, we can construct local excitations that replace $\ket{A_I}$ by a unit vector $\ket{\Phi}\in\widetilde{\cH}^{\otimes|I|}$ to lower the energy.
    Such an energy decrease would contradict the KMS ground-state condition for $\omega$, so $\omega(P_i^A)=0$.
    The excitation operator is $\ketbra{\Phi}{A_I}$, which replaces the component along $\ket{A_I}$ by $\ket{\Phi}$ according to $\ketbra{\Phi}{A_I}\rho_I\ketbra{A_I}{\Phi}=\omega(P_{A,I})\ketbra{\Phi}{\Phi}$.

    We first give the physical intuition for choosing $\ket{\Phi}$ so that $\ketbra{\Phi}{A_I}$ lowers the energy, including the interactions between $I$ and its two neighbouring sites.
    Fix $N\geq1$ such that $E_N\coloneqq\lambda_0(K_N)<-4$, as provided by \cref{eq:BCLPShiftedEnergyDichotomy}.
    For a short interval $I=\{1,\dots,\ell\}$ with $1\leq\ell<N+2$, consider unit vectors $\ket{\psi_1},\ket{\psi_2}\in\cHB$ at the sites $0$ and $\ell+1$ immediately outside $I$.
    Define $\ket{\Phi^\ell}\coloneqq\ket{\chi}^{\otimes\ell}\in\cHB^{\otimes\ell}$, where $\ket{\chi}=\ket{\perp}\otimes\ket{0}$.
    The excitation $\ketbra{\Phi^\ell}{A_I}$ on $I$ makes the change
    \begin{align}\label{eq:BulkReductionShortIntervalPicture}
        \ket{\psi_1}\otimes\ket{A_I}\otimes\ket{\psi_2}
        \ \longmapsto\ \ket{\psi_1}\otimes\ket{\Phi^\ell}\otimes\ket{\psi_2}.
    \end{align}
    Before the excitation, the $A$--$B$ bonds at the two ends of $I$ each contribute energy $2$ through the $A$--$B$ penalty term $2(P_j^A P_{j+1}^B+P_j^B P_{j+1}^A)$.
    After the excitation, both sites at each of these bonds lie in $\cHB$, so these two contributions disappear and the boundary energy decreases by $4$.
    The terms $\iota_r(h^{(1)})$ and $(\iota_r\otimes\iota_r)(h^{(2)})$ involving $I$ contribute no energy, since they vanish on both $\ket{A}$ and $\ket{\chi}$.
    The XY terms inside $I$ vanish on $\ket{\Phi^\ell}$, while the XY expectations on $\ket{\psi_1}\otimes\ket{\chi}$ and $\ket{\chi}\otimes\ket{\psi_2}$ at two ends of $I$ are at most $1/2$ each.
    Thus $\ketbra{\Phi^\ell}{A_I}$ lowers the energy by at least $4-1/2-1/2=3$ when these two exterior sites lie in $\cHB$; see also \cref{fig:BulkLocalReplacementShort} in the main text for illustration.

    Now consider a product state of $\ket{A}$ on a long interval.
    Regardless of its length, we act only on a fixed subinterval $I=\{0,\dots,N+1\}$ of length $N+2$, translating the site labels if necessary; the same construction therefore applies inside any longer or infinite interval.
    By the halting assumption in \cref{eq:BCLPShiftedEnergyDichotomy}, the chosen $K_N$ has a unit ground-state vector $\ket{\psi_N}$ with energy $E_N<-4$.
    We use it to define
    \begin{align}\label{eq:BulkReductionNegativeExcitationVector}
        \ket{\Phi_N}\coloneqq\ket{\chi}\otimes\bigl(\ket{\psi_N}\otimes\ket{0}^{\otimes N}\bigr)\otimes\ket{\chi}\in\cHB^{\otimes(N+2)},
    \end{align}
    where $\ket{\psi_N}$ occupies the $\mathds{C}^r$ spaces at sites $1,\dots,N$, with every qubit in $\ket{0}$.
    The excitation $\ketbra{\Phi_N}{A_I}$ replaces $\ket{A_I}$ by $\ket{\Phi_N}$ and acts as the identity outside $I$.
    Since $K_N\ket{\psi_N}=E_N\ket{\psi_N}$, the sum of $\iota_r(h^{(1)})_i$ and $(\iota_r\otimes\iota_r)(h^{(2)})_{i,i+1}$ on sites $1,\dots,N$ contributes energy $E_N$ on $\ket{\Phi_N}$, compared with zero on $\ket{A_I}$.
    The XY terms inside $I$ vanish because all its qubits are in $\ket{0}^{\otimes(N+2)}$.
    The $A$--$B$ penalty terms $2(P_j^A P_{j+1}^B+P_j^B P_{j+1}^A)$ inside $I$ also vanish because $\ket{\Phi_N}\in\cHB^{\otimes(N+2)}$.
    The operators $(\iota_r\otimes\iota_r)(h^{(2)})_{-1,0}$, $(\iota_r\otimes\iota_r)(h^{(2)})_{0,1}$, $(\iota_r\otimes\iota_r)(h^{(2)})_{N,N+1}$, and $(\iota_r\otimes\iota_r)(h^{(2)})_{N+1,N+2}$ all vanish when sites $0$ and $N+1$ are in $\ket{\chi}$.
    If $I$ lies strictly inside the product state of $\ket{A}$, the new pairs of $A$--$B$ bonds, $\ket{A}\otimes\ket{\chi}$ and $\ket{\chi}\otimes\ket{A}$, each contribute energy $2$, giving the total change $E_N+4<0$.
    More generally, a local density operator on the long interval need not be a product state, but we show below that the same energy bound holds on its component along $\ket{A_I}$.
    We refer to \cref{fig:BulkLocalReplacementLong} in the main text for illustration.

    \proofstep{2}{Short interval case.}
    We now formalize the above intuition for any KMS ground state $\omega$.
    For $1\leq\ell<N+2$ and $I=\{1,\dots,\ell\}$, recall $\ket{\Phi^\ell}=\ket{\chi}^{\otimes\ell}$ from the discussion above and define
    \begin{align}\label{eq:BulkReductionFiniteAIntervalReplacement}
        a_\ell\coloneqq P_0^B\ketbra{\Phi^\ell}{A_I}P_{\ell+1}^B,\qquad
        a_\ell^*a_\ell=P_0^B\left(\prod_{i=1}^{\ell}P_i^A\right)P_{\ell+1}^B.
    \end{align}
    For arbitrary $\ket{\psi_1},\ket{\psi_2}\in\widetilde{\cH}$, its action is
    \begin{align}\label{eq:BulkReductionConditionalExcitationAction}
        a_\ell\bigl(\ket{\psi_1}\otimes\ket{A}^{\otimes\ell}\otimes\ket{\psi_2}\bigr)
        =(P_B\ket{\psi_1})\otimes\ket{\Phi^\ell}\otimes(P_B\ket{\psi_2}),
    \end{align}
    which corresponds exactly to \cref{eq:BulkReductionShortIntervalPicture} when $\ket{\psi_1},\ket{\psi_2}\in\cHB$.

    Let $J=\{-1,\dots,\ell+2\}$.
    We now calculate $-\mathrm{i}\omega\bigl(a_\ell^*\delta_{\familyH}(a_\ell)\bigr)=\omega\bigl(a_\ell^*[\widetilde H^J,a_\ell]\bigr)$ term by term.
    We first compute $[\iota_r(h^{(1)})_j,a_\ell]$ for $j\in\mathds{Z}$.
    For $j\in I$, we have
    \begin{equation}\label{eq:BulkReductionShortEmbeddedOnsite}
        \begin{aligned}
            \iota_r(h^{(1)})_j a_\ell
            &=P_0^B\iota_r(h^{(1)})_j\ketbra{\Phi^\ell}{A_I}P_{\ell+1}^B=0,\\
            a_\ell\iota_r(h^{(1)})_j
            &=P_0^B\ketbra{\Phi^\ell}{A_I}\iota_r(h^{(1)})_j^*P_{\ell+1}^B=0.
        \end{aligned}
    \end{equation}
    Here $\iota_r(h^{(1)})_j$ is self-adjoint, commutes with the exterior projections, and annihilates both $\ket{A}_j$ and $\ket{\chi}_j$ by \cref{eq:BulkReductionEmbeddings}.
    For $j\notin I$, it commutes with $a_\ell$ because it acts outside $I$ and commutes with $P_0^B$ and $P_{\ell+1}^B$.
    The same argument applies to $(\iota_r\otimes\iota_r)(h^{(2)})_{j,j+1}$.
    Thus
    \begin{align}\label{eq:BulkReductionShortEmbeddedCommutators}
        [\iota_r(h^{(1)})_j,a_\ell]=[(\iota_r\otimes\iota_r)(h^{(2)})_{j,j+1},a_\ell]=0\qquad(j\in\mathds{Z}).
    \end{align}

    Next, for $1\leq j<\ell$, we have
    \begin{equation}\label{eq:BulkReductionShortInternalPenalties}
        \begin{aligned}
            2(P_j^AP_{j+1}^B+P_j^BP_{j+1}^A)\ket{A_I}&=0,\\
            2(P_j^AP_{j+1}^B+P_j^BP_{j+1}^A)\ket{\Phi^\ell}&=0.
        \end{aligned}
    \end{equation}
    Hence $[2(P_j^AP_{j+1}^B+P_j^BP_{j+1}^A),a_\ell]=0$ for these $j$.
    This commutator also vanishes for $j\leq-1$ or $j\geq\ell+1$, since the projections act outside $I$ and commute with $P_0^B$ and $P_{\ell+1}^B$.
    For $j=0$ and $j=\ell$, we compute
    \begin{equation}\label{eq:BulkReductionShortBoundaryPenalties}
        \begin{aligned}
            2(P_0^AP_1^B+P_0^BP_1^A)a_\ell&=0,\qquad a_\ell\,2(P_0^AP_1^B+P_0^BP_1^A)=2a_\ell,\\
            2(P_\ell^AP_{\ell+1}^B+P_\ell^BP_{\ell+1}^A)a_\ell&=0,\qquad a_\ell\,2(P_\ell^AP_{\ell+1}^B+P_\ell^BP_{\ell+1}^A)=2a_\ell,\\
            a_\ell^*\bigl[2(P_0^AP_1^B+P_0^BP_1^A)&+2(P_\ell^AP_{\ell+1}^B+P_\ell^BP_{\ell+1}^A),a_\ell\bigr]=-4a_\ell^*a_\ell.
        \end{aligned}
    \end{equation}

    Finally, consider the XY terms.
    Since $\ket{\chi}=\ket{\perp}\otimes\ket{0}$ and $q\ket{00}=0$, the terms $(\iota_2\otimes\iota_2)(q)_{j,j+1}$ with $\{j,j+1\}\subset I$ annihilate both $\ket{\Phi^\ell}$ and $\ket{A_I}$, so they commute with $a_\ell$.
    Terms supported outside $I$ also commute with $a_\ell$ by \cref{eq:BulkReductionSectorCommutation}.
    For $(j,j+1)=(0,1)$ or $(j,j+1)=(\ell,\ell+1)$, rewriting \cref{eq:BulkReductionQubitInteraction} gives
    \begin{equation}\label{eq:BulkReductionEmbeddedXYExpansion}
        \begin{aligned}
            (\iota_2\otimes\iota_2)(q)_{j,j+1}
            &=\frac12\Bigl(\iota_2(\ketbra{1}{1})_jP_{j+1}^B+P_j^B\iota_2(\ketbra{1}{1})_{j+1}\\
            &\qquad-\iota_2(\ketbra{1}{0})_j\iota_2(\ketbra{0}{1})_{j+1}-\iota_2(\ketbra{0}{1})_j\iota_2(\ketbra{1}{0})_{j+1}\Bigr),
            \qquad j\in\{0,\ell\}.
        \end{aligned}
    \end{equation}
    By \cref{eq:BulkReductionEmbeddings}, $\iota_2(x)\ket{\chi}=\ket{\perp}\otimes x\ket{0}$ and $\iota_2(x)\ket{A}=0$ for $x\in M_2(\mathds{C})$.
    Direct calculation therefore gives
    \begin{equation}\label{eq:BulkReductionShortXYBoundaryCommutators}
        \begin{aligned}
            a_\ell^*[(\iota_2\otimes\iota_2)(q)_{0,1},a_\ell]
            &=\frac12a_\ell^*a_\ell\iota_2(\ketbra{1}{1})_0,\\
            a_\ell^*[(\iota_2\otimes\iota_2)(q)_{\ell,\ell+1},a_\ell]
            &=\frac12a_\ell^*a_\ell\iota_2(\ketbra{1}{1})_{\ell+1}.
        \end{aligned}
    \end{equation}

    Combining \cref{eq:BulkReductionShortEmbeddedCommutators,eq:BulkReductionShortBoundaryPenalties,eq:BulkReductionShortXYBoundaryCommutators} with the KMS ground-state condition for $\omega$ in \cref{eq:KMSGroundCondition} and $\iota_2(\ketbra{1}{1})_j\leq P_j^B$ gives
    \begin{equation}\label{eq:BulkReductionShortIntervalZeroWeight}
        \begin{aligned}
            0\leq-\mathrm{i}\omega\bigl(a_\ell^*\delta_{\familyH}(a_\ell)\bigr)
            &=\omega\bigl(a_\ell^*[\widetilde H^J,a_\ell]\bigr)\\
            &=\omega\left(a_\ell^*a_\ell\left(-4+\frac12\iota_2(\ketbra{1}{1})_0+\frac12\iota_2(\ketbra{1}{1})_{\ell+1}\right)\right)\\
            &\leq\omega\left(a_\ell^*a_\ell\left(-4+\frac12P_0^B+\frac12P_{\ell+1}^B\right)\right)
            =-3\omega(a_\ell^*a_\ell)\leq0.
        \end{aligned}
    \end{equation}
    Therefore $\omega(P_0^B(\prod_{i=1}^{\ell}P_i^A)P_{\ell+1}^B)=0$.

    \proofstep{3}{Long interval case.}
    For $I=\{0,\dots,N+1\}$, use $\ket{\Phi_N}$ defined in \cref{eq:BulkReductionNegativeExcitationVector} and set
    \begin{align}\label{eq:BulkReductionNegativeExcitation}
        a_{N+2}\coloneqq\ketbra{\Phi_N}{A_I},\qquad a_{N+2}^*a_{N+2}=P_{A,I}.
    \end{align}
    Recall that $E_N+4<0$.
    By construction, $\widetilde H^I\ket{\Phi_N}=E_N\ket{\Phi_N}$ and $\widetilde H^I\ket{A_I}=0$, so
    \begin{align}\label{eq:BulkReductionLongInternalCommutator}
        a_{N+2}^*[\widetilde H^I,a_{N+2}]=E_NP_{A,I}.
    \end{align}
    At either site $j\in\{-1,N+2\}$ outside $I$, the interaction with its neighbour in $I$ contributes $2P_j^B$ before the replacement and $2P_j^A+\frac12\iota_2(\ketbra{1}{1})_j$ after taking the expectation in $\ket{\Phi_N}$ on $I$.
    Explicitly, using \cref{eq:BulkReductionEmbeddedXYExpansion},
    \begin{equation}\label{eq:BulkReductionLongBoundaryContributions}
        \begin{aligned}
            a_{N+2}^*a_{N+2}\left(\widetilde h^{(2)}_{-1,0}+\widetilde h^{(2)}_{N+1,N+2}\right)
            &=2P_{A,I}\sum_{j\in\{-1,N+2\}}P_j^B,\\
            a_{N+2}^*\left(\widetilde h^{(2)}_{-1,0}+\widetilde h^{(2)}_{N+1,N+2}\right)a_{N+2}
            &=P_{A,I}\sum_{j\in\{-1,N+2\}}\left(2P_j^A+\frac12\iota_2(\ketbra{1}{1})_j\right).
        \end{aligned}
    \end{equation}
    Let $J=\{-1,\dots,N+2\}$.
    Combining these identities and using $\iota_2(\ketbra{1}{1})_j\leq P_j^B$ and $P_j^A+P_j^B=\id$ gives
    \begin{equation}\label{eq:BulkReductionLongAInterval}
        \begin{aligned}
            -\mathrm{i}a_{N+2}^*\delta_{\familyH}(a_{N+2})
            &=a_{N+2}^*[\widetilde H^J,a_{N+2}]\\
            &=P_{A,I}\left(E_N+\sum_{j\in\{-1,N+2\}}\left(2(P_j^A-P_j^B)+\frac12\iota_2(\ketbra{1}{1})_j\right)\right)\\
            &\leq(E_N+4)P_{A,I},
        \end{aligned}
    \end{equation}
    so taking the expectation in $\omega$ and applying the KMS ground-state condition again gives
    \begin{align}\label{eq:BulkReductionLongIntervalZeroWeight}
        0\leq-\mathrm{i}\omega\bigl(a_{N+2}^*\delta_{\familyH}(a_{N+2})\bigr)\leq(E_N+4)\omega(P_{A,I})\leq0
        \quad\Longrightarrow\quad\omega(P_{A,I})=0.
    \end{align}
    We can reindex any finite interval in $\mathds{Z}$ and repeat the same calculations, so both conclusions hold at every position.
    These arguments use no assumption of translation invariance of $\omega$.

    We now prove $\omega(P_{A,I})=0$ for every interval $I$ of length $\ell$ by induction from $\ell=N+2$ to $\ell=1$.
    The case $\ell=N+2$ is \cref{eq:BulkReductionLongIntervalZeroWeight}.
    Suppose the conclusion holds for all intervals of length $\ell+1$, where $1\leq\ell\leq N+1$, and let $I=\{u,\dots,v\}$ have length $\ell$.
    Using $P_j^A+P_j^B=\id$ at $j=u-1,v+1$, we have
    \begin{align}\label{eq:BulkReductionIntervalInduction}
        P_{A,I}=P_{A,\{u-1,\dots,v\}}+P_{u-1}^BP_{A,\{u,\dots,v+1\}}+P_{u-1}^BP_{A,I}P_{v+1}^B.
    \end{align}
    Taking the expectation in $\omega$, we obtain
    \begin{equation}\label{eq:BulkReductionIntervalInductionExpectation}
        \begin{aligned}
            0\leq\omega(P_{A,I})
            &=\omega\!\left(P_{A,\{u-1,\dots,v\}}\right)+\omega\!\left(P_{u-1}^BP_{A,\{u,\dots,v+1\}}\right)+\omega\!\left(P_{u-1}^BP_{A,I}P_{v+1}^B\right)\\
            &\leq\omega\!\left(P_{A,\{u-1,\dots,v\}}\right)+\omega\!\left(P_{A,\{u,\dots,v+1\}}\right)+\omega\!\left(P_{u-1}^BP_{A,I}P_{v+1}^B\right)=0,
        \end{aligned}
    \end{equation}
    where the inequality uses $P_{u-1}^B\leq\id$, and the last equality follows from the induction hypothesis and \cref{eq:BulkReductionShortIntervalZeroWeight}.
    This completes the induction, and taking $\ell=1$ gives $\omega(P_i^A)=0$ for every $i\in\mathds{Z}$.
\end{proof}

While the proof of \cref{lem:BulkReductionHaltingSupport} is involved, its physical intuition is simple.
Replacing a short interval in $\cHA$ surrounded by $\cHB$ saves $4$ from two $A$--$B$ bonds at a cost of at most $1$ due to XY interaction.
On a long interval the halting assumption provides $\ket{\Phi_N}$ with energy $E_N<-4$, outweighing the cost of two $A$--$B$ bonds which is at most $4$.
Thus any KMS ground state cannot have contribution from $\cHA$, formally $\omega(P^A_i) = 0$.

\subsection{Proof of the main theorem}\label{sec:BulkUndecidabilityProof}
To prove \cref{thm:BulkGapUndecidability}, recall that the single-site space is $\widetilde{\cH}=\mathds{C}\ket{A}\oplus((\mathds{C}^r\oplus\mathds{C}\ket{\perp})\otimes\mathds{C}^2)$, with $P_A=\ketbra{A}{A}$ and $P_B=\id-P_A$, and the interaction is
\begin{equation}\label{eq:BulkReductionInteractionRecall}
    \begin{aligned}
        \widetilde h^{(1)}(m)&=\iota_r(h^{(1)}(m)),\\
        \widetilde h^{(2)}(m)&=(\iota_r\otimes\iota_r)(h^{(2)}(m))+(\iota_2\otimes\iota_2)(q)+2(P_A\otimes P_B+P_B\otimes P_A),
    \end{aligned}
\end{equation}
where the embeddings and $q$ are defined in \cref{eq:BulkReductionEmbeddings,eq:BulkReductionQubitInteraction}, and $h^{(1)}(m)$ and $h^{(2)}(m)$ are constructed in~\cite[Theorem~22 and Lemma~23]{bausch2020undecidability}.
The product state $\omega_A$ is determined by $\omega_A(a)=\sandwich{A_I}{a}{A_I}$ for $a\in\cA_I$, where $\ket{A_I}=\bigotimes_{i\in I}\ket{A}$.

\begin{proof}[Proof of \cref{thm:BulkGapUndecidability}]
    The local dimension is $\dim\widetilde{\cH}=2r+3$.
    The embeddings in \cref{eq:BulkReductionEmbeddings} are isometric, $\norm{q}=1$, and $P_A\otimes P_B$, $P_B\otimes P_A$ are orthogonal projections, so \cref{eq:BulkReductionInteractionRecall} gives $\norm{\widetilde h^{(1)}(m)}\leq C$ and $\norm{\widetilde h^{(2)}(m)}\leq C+3$, independently of $m$.
    Computability, algebraicity, translation invariance, and nearest-neighbour support follow directly from the construction.
    Existence of a KMS ground state follows from \cref{rem:ExistenceKMSGroundState}, applied to $\familyH(m)$ with open boundary conditions.
    In the nonhalting case, \cref{lem:BulkReductionNonhaltingGap} shows that the product state $\omega_A$ has locally non-degenerate bulk gap at least $4$, proving statement~(i).

    For statement~(ii), suppose that $M(m)$ halts, and let $\omega$ be any KMS ground state of $\familyH(m)$.
    For each finite interval $I$, recall that $\cQ_I=\bigotimes_{i\in I}M_2(\mathds{C})$ is the local qubit algebra in \cref{sec:XYSpectrum}, and set $P_{B,I}\coloneqq\prod_{i\in I}P_i^B$ and $\iota_I\coloneqq\bigotimes_{i\in I}\iota_2:\cQ_I\to\cA_I$.
    By \cref{lem:BulkReductionHaltingSupport}, $\omega(P_i^A)=0$ at every site, so $\omega(P_{B,I})=1$ by $\id - P_{B,I} \leq \sum_{i \in I}P^A_i$.
    Cauchy--Schwarz gives $|\omega((\id-P_{B,I})a)|^2\leq\omega(\id-P_{B,I})\omega(a^*a)=0$ for every $a\in\cA$, and taking adjoints also gives $\omega(a(\id-P_{B,I}))=0$.
    Thus we have the compression identity $\omega(a)=\omega(P_{B,I}a)=\omega(P_{B,I}aP_{B,I})$.
    
    Define
    \begin{align}\label{eq:BulkReductionQubitMarginal}
        \eta(b)\coloneqq\omega(\iota_I(b))\qquad(b\in\cQ_I)
    \end{align}
    on each $\cQ_I$.
    This is positive because $\iota_I$ is a $*$-homomorphism, and normalized because $\eta(\id_I)=\omega(\iota_I(\id_I))=\omega(P_{B,I})=1$.
    Since $\iota_2(\id_2)=P_B$, for $I\subset J$ and $b\in\cQ_I$ we have
    \begin{align}\label{eq:BulkReductionEmbeddingCompression}
        \iota_J(b\otimes\id_{J\setminus I})=P_{B,J}\iota_I(b)P_{B,J}.
    \end{align}
    Thus $\omega(\iota_J(b\otimes\id_{J\setminus I}))=\omega(\iota_I(b))$ by the compression identity for $\omega$.
    Hence the local states agree when the interval is enlarged and define a state on $\cQ$ by continuity.
    
    We next check the KMS ground-state condition for $\eta$.
    Fix $I=\{u,\dots,v\}$, $b\in\cQ_I$, and $J=\{u-1,\dots,v+1\}$, identifying $b$ with $b\otimes\id_{J\setminus I}$ in $\cQ_J$.
    By construction, only the terms $(\iota_2\otimes\iota_2)(q)_{j,j+1}$ in \cref{eq:BulkReductionInteractionRecall} can have a nonzero commutator with $\iota_I(b)$.
    Applying \cref{eq:BulkReductionEmbeddingCompression} to $b$ and to $q_{j,j+1}$, and using the $*$-homomorphism property of $\iota_J$, gives
    \begin{equation}\label{eq:BulkReductionQubitCompressedDerivation}
        \begin{aligned}
            P_{B,J}\delta_{\familyH(m)}(\iota_I(b))P_{B,J}
            &=\mathrm{i}\sum_{\{j,j+1\}\subset J}P_{B,J}[(\iota_2\otimes\iota_2)(q)_{j,j+1},\iota_I(b)]P_{B,J}\\
            &=\mathrm{i}\sum_{\{j,j+1\}\subset J}[\iota_J(q_{j,j+1}),\iota_J(b)]=\iota_J(\delta_q(b)).
        \end{aligned}
    \end{equation}
    Here terms with $\{j,j+1\}\not\subset J$ commute with $\iota_I(b)$ because neither site lies in $I$.
    Using this identity and the compression identity for $\omega$, we obtain
    \begin{equation}\label{eq:BulkReductionQubitMarginalGround}
        \begin{aligned}
            -\mathrm{i}\eta(b^*\delta_q(b))
            &=-\mathrm{i}\omega\!\left(\iota_J(b)^*\iota_J(\delta_q(b))\right)=-\mathrm{i}\omega\!\left(P_{B,J}\iota_I(b)^*P_{B,J}\delta_{\familyH(m)}(\iota_I(b))P_{B,J}\right)\\
            &=-\mathrm{i}\omega\!\left(\iota_I(b)^*\delta_{\familyH(m)}(\iota_I(b))\right)\geq0,
        \end{aligned}
    \end{equation}
    by $P_{B,J}\iota_I(b)^*P_{B,J}=P_{B,J}\iota_I(b)^*$, the compression identity, and the KMS ground state condition of $\omega$.
    Hence $\eta$ is a KMS ground state of the XY dynamics in \cref{lem:BulkReductionQubitSpectrum}, by extension from the local core, implying $\eta=\nu$ by uniqueness.

    Since $\eta=\nu$, the map
    \begin{align}\label{eq:BulkReductionGNSIsometry}
        U:\cH_\nu\longrightarrow\cH_\omega,\qquad U\pi_\nu(b)\ket{\Omega_\nu}=\pi_\omega(\iota_I(b))\ket{\Omega_\omega}\quad(b\in\cQ_I)
    \end{align}
    extends to an isometry: the squared norm on either side is $\nu(b^*b)=\omega(\iota_I(b)^*\iota_I(b))$.
    Its definition is independent of $I$ by \cref{eq:BulkReductionEmbeddingCompression} and $\pi_\omega(P_{B,J})\ket{\Omega_\omega}=\ket{\Omega_\omega}$.
    
    Every local term of $\widetilde{H}^I(m)$ and $\iota_I(b)$ commute with $P_{B,J}$ by \cref{eq:BulkReductionSectorCommutation}, so does $\delta_{\familyH(m)}(\iota_I(b))$ commute with $P_{B,J}$.
    Also, $\pi_\omega(P_{B,J})\ket{\Omega_\omega}=\ket{\Omega_\omega}$ implies that $\pi_\omega(\delta_{\familyH(m)}(\iota_I(b)))\ket{\Omega_\omega}=\pi_\omega(P_{B,J}\delta_{\familyH(m)}(\iota_I(b))P_{B,J})\ket{\Omega_\omega}$.
    Thus \cref{eq:BulkHamiltonianProperty,eq:BulkReductionQubitCompressedDerivation} give
    \begin{align}\label{eq:BulkReductionGNSIntertwining}
        H_\omega U\pi_\nu(b)\ket{\Omega_\nu}=-\mathrm{i}\pi_\omega(\iota_J(\delta_q(b)))\ket{\Omega_\omega}=UH_\nu\pi_\nu(b)\ket{\Omega_\nu}.
    \end{align}
    Since $\pi_\nu(\cQ_{\mathrm{loc}})\ket{\Omega_\nu}$ is a core for $H_\nu$ and $H_\omega$ is closed, \cref{eq:BulkReductionGNSIntertwining} extends to $H_\omega U \ket{\xi}=UH_\nu \ket{\xi}$ for all $\ket{\xi}$ in the domain of $H_\nu$, so the isometry $U$ maps approximate eigenvectors of $H_\nu$ to approximate eigenvectors of $H_\omega$ with the same value, and $\Spectrum{H_\nu}\subset\Spectrum{H_\omega}$.
    By \cref{lem:BulkReductionQubitSpectrum} and $H_\omega\geq0$,
    \begin{align}\label{eq:BulkReductionSpectrumTransfer}
        [0,\infty)=\Spectrum{H_\nu}\subset\Spectrum{H_\omega}\subset[0,\infty).
    \end{align}
    This proves statement~(ii), including strong bulk gaplessness of the dynamical system in the halting case.

    Statements~(i) and~(ii), together with \cref{def:BulkSpectralGap}, give \cref{eq:HaltingToGapinBulk}.
    An algorithm deciding whether $\Delta_{\mathrm{bulk}}>0$ would therefore decide whether $M(m)$ halts.

    Finally, suppose an algorithm producing the lower bounds in \cref{eq:BulkGapLowerBoundAlgorithm} existed.
    Given $m$, compute the local matrices defining $\familyH(m)$ and run $M(m)$ and the lower-bound algorithm in parallel, outputting ``halts'' if $M(m)$ halts and ``does not halt'' if some $\ell_s>0$.
    If $M(m)$ halts, all $\ell_s\leq0$; otherwise, convergence to $\Delta_{\mathrm{bulk}}(\familyH(m))\geq1$ ensures $\ell_s>0$ for some finite $s$.
    This procedure always terminates with the correct answer, contradicting undecidability of the halting problem.
\end{proof}

\subsection{Uncomputability of KMS ground states}\label{sec:UncomputabilityKMSGroundState}
Our \cref{thm:BulkGapUndecidability} has further implications on computability of KMS ground states and their expectation values.

\begin{corollary}[Computational limits for KMS ground states]\label{cor:KMSGroundStateComputability}
    Consider one-dimensional translationally invariant nearest-neighbour systems with fixed local dimension, uniformly bounded local terms, and algebraic matrix entries.
    The following statements hold:
    \begin{enumerate}[label=(\roman*)]
        \item Deciding whether the fixed product state $\omega_A$ in \cref{eq:BulkReductionProductState} is a KMS ground state is $\mathsf{coRE}$-complete.
        \item No computable choice of testing region is enough.
        That is, for every computable function $L$, there is an input $m$ such that $\omega_A$ satisfies the bulk-gap inequality with $\gamma = 4$ for every excitation supported in $\Lambda(L(m))$, yet is not a KMS ground state of $\familyH(m)$.
        \item No general algorithm producing convergent certified upper bounds on the largest KMS ground-state expectation $\langle O\rangle_{\max}$ admits a computable convergence rate, already for the one-site projection $O=P_0^A$.
        Here $\langle O\rangle_{\max}$ is defined in \cref{eq:ExtremalExpectationValues} with $\gamma=0$.
    \end{enumerate}
\end{corollary}

\begin{proof}
    For (i), \cref{lem:BulkReductionNonhaltingGap,lem:BulkReductionHaltingSupport} give
    \begin{align}\label{eq:KMSProductStateHaltingReduction}
        \omega_A \text{ is a KMS ground state of } \familyH(m)
        \quad\Longleftrightarrow\quad
        M(m) \text{ does not halt}.
    \end{align}
    This proves $\mathsf{coRE}$-hardness.
    For membership in $\mathsf{coRE}$, use the increasing finite-level tests underlying \cref{alg:SpectralGapSemi} at $\gamma=0$, fixing the local functional in \cref{eq:GappedFeasibilityHierarchy} to the restriction of $\omega_A$.
    Each test is decidable by \cref{rem:ComputabilityQuantifierElimination}.
    If $\omega_A$ is not a KMS ground state, the conditions fail for some local excitation and hence at a finite level, since $\cA_{\mathrm{loc}}$ is a core for $\delta_{\familyH}$.
    If it is a KMS ground state, every level passes.
    Thus failure to be a KMS ground state can be semidecided, proving membership in $\mathsf{coRE}$.

    For (ii), suppose an algorithm could choose a finite testing region for which passing the bulk-gap test with $\gamma=4$ guarantees that $\omega_A$ is a KMS ground state.
    Applied to $\familyH(m)$, enlarge its output to $\Lambda(L(m))$ if necessary.
    Consider the finite-region condition
    \begin{align}\label{eq:KMSFiniteTestingObstruction}
        -\mathrm{i}\omega_A\bigl(a^*\delta_{\familyH(m)}(a)\bigr)
        \geq 4\bigl(\omega_A(a^*a)-\abs{\omega_A(a)}^2\bigr)
        \qquad\text{for every }a\in\cA_{\Lambda(L(m))}.
    \end{align}
    This condition is decidable because it is a finite-dimensional quadratic-form inequality with algebraic coefficients; see \cref{rem:ComputabilityQuantifierElimination}.
    If $M(m)$ does not halt, \cref{eq:KMSFiniteTestingObstruction} holds by \cref{lem:BulkReductionNonhaltingGap}.
    If $M(m)$ halts, $\omega_A$ is not a KMS ground state by \cref{lem:BulkReductionHaltingSupport}, so the assumed guarantee forces the test to fail.
    The finite test would therefore decide whether $M(m)$ halts, a contradiction.

    For (iii), in the nonhalting case both $\omega_A$ and the product state induced by $\ket{\chi}$ are KMS ground states as already observed.
    Their expectations of $P_0^A$ are $1$ and $0$, respectively.
    Convexity of the KMS ground-state set, together with \cref{lem:BulkReductionHaltingSupport}, implies
    \begin{equation}\label{eq:KMSProjectionExpectationDichotomy}
        \begin{aligned}
            C(m)&\coloneqq\{ \omega(P_0^A) \mid \omega \text{ is a KMS ground state of } \familyH(m) \}\\
            &=\begin{cases}
                [0,1], & M(m)\text{ does not halt},\\
                \{0\}, & M(m)\text{ halts}.
            \end{cases}
        \end{aligned}
    \end{equation}
    Suppose an algorithm produces rational upper bounds $u_s(m)$ satisfying
    \begin{align}\label{eq:KMSExpectationUpperBounds}
        u_s(m)\geq\max C(m),\qquad
        \lim_{s\to\infty}u_s(m)=\max C(m).
    \end{align}
    If its convergence rate were computable, then, given $m$ and an accuracy $2^{-k}$, we could compute an index $N(m,k)$ such that
    \begin{align}\label{eq:KMSExpectationConvergenceRate}
        0\leq u_{N(m,k)}(m)-\max C(m)\leq 2^{-k}.
    \end{align}
    Taking $k=2$, \cref{eq:KMSProjectionExpectationDichotomy} yields
    \begin{align}\label{eq:KMSExpectationHaltingDecision}
        M(m)\text{ halts}
        \quad\Longleftrightarrow\quad
        u_{N(m,2)}(m)<\frac12.
    \end{align}
    This would decide the halting problem, a contradiction.
\end{proof}

Statement~(iii) of \cref{cor:KMSGroundStateComputability} applies to the convergent bounds on ground-state expectation values obtained from \cref{thm:CompletenessSDPHierarchy} at $\gamma=0$ and from the hierarchies of~\cite{araujo2023first,fawzi2024certified}.
Even for a one-site projection, no general algorithm can determine how far to run these hierarchies to achieve a prescribed accuracy.

\subsection{Alternative construction and stability under perturbation}\label{sec:AltConstructionStability}
We finish the section with discussions on an alternative construction due to local degeneracy and the sensitivity of our construction to onsite perturbations.

\begin{remark}[Local degeneracy construction without strong bulk gaplessness]\label{rem:BulkUndecidabilityGapNotion}
    In the halting case of \cref{thm:BulkGapUndecidability}, the dynamical system is strongly bulk-gapless in the sense of \cref{def:GaplessAlgebraic}.
    In fact, $\Spectrum{H_\omega}=[0,\infty)$ implies $[0,c]\subset\Spectrum{H_\omega}$ for every $c>0$ and every KMS ground state $\omega$.
    This is stronger than ordinary bulk gaplessness, which only requires $\Spectrum{H_\omega}\cap(0,\varepsilon)\neq\emptyset$ for every $\varepsilon>0$, and stronger than $\Delta_{\mathrm{bulk}}=0$ in \cref{def:BulkSpectralGap}, which may be caused by local degeneracy alone.

    For undecidability of whether $\Delta_{\mathrm{bulk}}>0$ alone, it suffices to leave the qubits free by removing the XY interaction.
    Keeping the same single-site space and embeddings, let $\familyH_{\mathrm{free}}(m)$ be generated by
    \begin{align}\label{eq:BulkReductionFreeQubitInteraction}
        \widetilde h_{\mathrm{free}}^{(1)}(m)=\widetilde h^{(1)}(m),\qquad \widetilde h_{\mathrm{free}}^{(2)}(m)=\widetilde h^{(2)}(m)-(\iota_2\otimes\iota_2)(q).
    \end{align}
    The nonhalting case still has the product KMS ground state $\omega_A$ with locally non-degenerate bulk gap at least $4$; in the halting case, the same replacement estimates give $\omega(P_i^A)=0$ for every KMS ground state and every site.
    Since $\delta_{\familyH_{\mathrm{free}}(m)}(\iota_2(y)_i)=0$ for every $y\in M_2(\mathds{C})$, every local qubit observable has zero energy cost.
    For the Pauli matrices $X,Y$, the qubit-state bound $\omega(\iota_2(X)_i)^2+\omega(\iota_2(Y)_i)^2\leq1$ ensures that at least one $\sigma\in\{X,Y\}$ satisfies $\bar{\sigma}^2\leq1/2$, where $\bar{\sigma}\coloneqq\omega(\iota_2(\sigma)_i)$.
    Subtracting its mean and normalizing gives $a\coloneqq\iota_2(\sigma-\bar{\sigma}\id_2)_i/\sqrt{1-\bar{\sigma}^2}$, with $\delta_{\familyH_{\mathrm{free}}(m)}(a)=0$, $\omega(a)=0$, and $\omega(a^*a)=1$.
    These zero-energy local excitations rule out any positive locally non-degenerate bulk gap, giving $\Delta_{\mathrm{bulk}}(\familyH_{\mathrm{free}}(m))=0$, but do not establish even ordinary bulk gaplessness in the sense of \cref{def:GaplessAlgebraic}.
\end{remark}

\begin{remark}[Instability under onsite perturbations]\label{rem:BulkReductionOnsiteInstability}
    We note that our construction is not stable under onsite perturbations.
    For $\varepsilon>0$, let $\familyH_\varepsilon(m)$ be generated by
    \begin{align}\label{eq:BulkReductionOnsitePerturbation}
        \widetilde h_\varepsilon^{(1)}(m)\coloneqq\widetilde h^{(1)}(m)+\varepsilon P_A,
        \qquad
        \widetilde h_\varepsilon^{(2)}(m)\coloneqq\widetilde h^{(2)}(m).
    \end{align}
    We show that the perturbed system is strongly bulk-gapless for every input $m$, independently of whether $M(m)$ halts.

    Fix $m$ and let $\omega$ be a KMS ground state of the perturbed dynamics.
    For an interval $I=\{u,\dots,v\}$ of length $\ell$, consider
    \begin{align}\label{eq:BulkReductionPerturbedReplacement}
        b_I\coloneqq\ket{\chi}^{\otimes\ell}\bra{A_I},
        \qquad b_I^*b_I=P_{A,I}.
    \end{align}
    The unperturbed interactions inside $I$ annihilate both $\ket{A_I}$ and $\ket{\chi}^{\otimes\ell}$.
    Since $P_j^A b_I=0$ and $b_I P_j^A=b_I$ for $j\in I$, replacing $\ket{A_I}$ removes the onsite energy $\varepsilon\ell$:
    \begin{align}\label{eq:BulkReductionPerturbedOnsiteContribution}
        b_I^*\left[\varepsilon\sum_{j\in I}P_j^A,b_I\right]
        =-\varepsilon\ell P_{A,I}.
    \end{align}
    At each site $j\in\{u-1,v+1\}$, the boundary bond contributes $2P_j^B$ before the replacement and $2P_j^A+\frac12\iota_2(\ketbra{1}{1})_j$ after taking the expectation in $\ket{\chi}$ at the neighbouring site in $I$, as in \cref{eq:BulkReductionLongBoundaryContributions}.
    Consequently,
    \begin{equation}\label{eq:BulkReductionPerturbedInterval}
        \begin{aligned}
            -\mathrm{i}b_I^*\delta_{\familyH_\varepsilon(m)}(b_I)
            &=P_{A,I}\left(-\varepsilon\ell+\sum_{j\in\{u-1,v+1\}}
            \left(2P_j^A-2P_j^B+\frac12\iota_2(\ketbra{1}{1})_j\right)\right)\\
            &\leq(4-\varepsilon\ell)P_{A,I},
        \end{aligned}
    \end{equation}
    where each boundary contribution is at most $2\id$ because $\iota_2(\ketbra{1}{1})_j\leq P_j^B$ and $P_j^A+P_j^B=\id$.
    Taking the expectation in $\omega$ and applying the KMS ground-state condition gives $\omega(P_{A,I})=0$ whenever $\ell>4/\varepsilon$.

    For the short-interval calculation, take $I=\{1,\dots,\ell\}$ and use the conditional excitation $a_\ell=P_0^B\ket{\chi}^{\otimes\ell}\bra{A_I}P_{\ell+1}^B$ from \cref{eq:BulkReductionFiniteAIntervalReplacement}.
    The replacement removes the two $A$--$B$ bonds, contributing $-4$, and the onsite perturbation contributes $-\varepsilon\ell$.
    Including the two boundary XY contributions from \cref{eq:BulkReductionShortXYBoundaryCommutators}, we obtain
    \begin{equation}\label{eq:BulkReductionPerturbedShortInterval}
        \begin{aligned}
            0\leq-\mathrm{i}\omega\bigl(a_\ell^*\delta_{\familyH_\varepsilon(m)}(a_\ell)\bigr)
            &=\omega\left(a_\ell^*a_\ell\left(-4-\varepsilon\ell
            +\frac12\iota_2(\ketbra{1}{1})_0
            +\frac12\iota_2(\ketbra{1}{1})_{\ell+1}\right)\right)\\
            &\leq-(3+\varepsilon\ell)\omega(a_\ell^*a_\ell)\leq0.
        \end{aligned}
    \end{equation}
    Hence $\omega(P_0^B P_{A,I}P_{\ell+1}^B)=0$.
    This calculation holds for every $\ell\geq1$, and both interval calculations hold at every position by translation-invariance.
    Starting from any integer length greater than $4/\varepsilon$, the same induction argument in \cref{eq:BulkReductionIntervalInduction,eq:BulkReductionIntervalInductionExpectation} now shows that $\omega(P_i^A)=0$ for every $i$.

    Since $[P_j^A,\iota_I(b)]=0$ for every $j$ and every $b\in\cQ_I$, the perturbation leaves the derivation unchanged on embedded qubit observables.
    Together with $\omega(P_i^A)=0$ at every site, the same argument leading to \cref{eq:BulkReductionSpectrumTransfer} therefore gives $\Spectrum{H_{\varepsilon,\omega}}=[0,\infty)$, where $H_{\varepsilon,\omega}$ is now the bulk Hamiltonian of the perturbed dynamics.
    KMS ground states exist by \cref{rem:ExistenceKMSGroundState}, so $\familyH_\varepsilon(m)$ is strongly bulk-gapless for every $m$ and every $\varepsilon>0$.
    Thus an arbitrarily small onsite perturbation destroys the distinction between the halting and nonhalting cases.
\end{remark}

To conclude, by \cref{thm:CompletenessSDPHierarchy}, the SDP hierarchy gives convergent upper bounds on $\Delta_{\mathrm{bulk}}$: for every $\gamma>\Delta_{\mathrm{bulk}}$, infeasibility at a finite level excludes that threshold.
\Cref{thm:BulkGapUndecidability} rules out a general algorithm producing convergent lower bounds, even under the promise that the system either admits a KMS ground state with locally non-degenerate bulk gap at least $4$ or is strongly bulk-gapless.
Together, these results make the bulk-gaplessness promise problem $\mathsf{RE}$-complete and the complementary bulk-gap promise problem $\mathsf{coRE}$-complete (\cref{cor:BulkGapNoConvergenceRate}).

\section{Numerical showcases of the upper-bound algorithm}\label{sec:NumericalShowcase}
Even though \cref{sec:BulkGapUndecidability} establishes that no general algorithm exists to produce convergent lower bounds on the bulk spectral gap for even ``physically reasonable'' systems, our SDP remains a general method for the converse direction.
At each finite level $(n,d)$, feasibility of \cref{def:GappedRelaxHierarchyStatePoly} with parameter $\gamma$ is necessary for the existence of a KMS ground state satisfying \cref{eq:GappedSystemProperty}.
The supremum of feasible values of $\gamma$ therefore gives an upper bound on $\Delta_{\mathrm{bulk}}$.
As both $n$ and $d$ increase, these bounds are monotonically nonincreasing and converge to $\Delta_{\mathrm{bulk}}$ by \cref{thm:CompletenessSDPHierarchy}.
Equivalently, every threshold $\gamma>\Delta_{\mathrm{bulk}}$ is eventually excluded by \cref{alg:SpectralGapSemi}.
Here we evaluate the hierarchy at finite levels to obtain numerical upper bounds for the models below.

We test the SDP hierarchy on the transverse-field Ising model (\cref{sec:TransverseFieldIsing}) and the kagome lattice Heisenberg model (\cref{sec:KagomeHeisenberg}).
In practice, the computational cost can be reduced by exploiting structural constraints, positivity of reduced density matrices, and symmetries specific to the model under study, as in~\cite{wang2026scalable}.
For the numerical implementations reported below, we use a selected subset of state monomials for each degree $d$ to index the moment matrix in order to better balance tightness of SDPs and computational costs.
The precise monomial sets used in each computation, together with reproducibility details, are provided at \url{https://github.com/wangjie212/SpectralGap}.
All numerical experiments were performed on a desktop computer with Intel(R) Core(TM) i9-10900 CPU@2.80GHz and 64G RAM except that the computation on the Kagome lattice Heisenberg model with $L=3,4$ was conducted on a workstation with Intel(R) Xeon(R) Gold 6348 CPU@2.60GHz and 773G RAM.

\subsection{Transverse-field Ising model}\label{sec:TransverseFieldIsing}
A standard analytically solvable example is the transverse-field Ising model on the infinite chain $\mathds{Z}$.
Here each site $i$ carries a qubit algebra $\mathcal A_i \simeq M_2(\mathds C)$ with Pauli operators $X_i,Z_i\in\mathcal A_i$.
For $g\geq 0$, the finite-volume Hamiltonian on $\Lambda(L)=\{-L,\dots,L\}\subset\mathds Z$ is
\begin{align}\label{eq:1DIsingModelHamiltonian}
    H_{\mathrm{TFIM},g}^{\Lambda(L)}
    = -\sum_{i=-L}^{L-1} Z_iZ_{i+1} + g\sum_{i=-L}^{L} X_i.
\end{align}
The exact solution indicates three regimes~\cite{mbeng2024quantum}:
\begin{enumerate}
    \item \emph{Disordered phase} $g > 1$: unique ground state with a finite-volume spectral gap that converges to $2(g-1)$ in the thermodynamic limit.
    \item \emph{Critical point} $g=1$: unique ground state; the finite-size gap closes as $1/L$.
    \item \emph{Ordered phase} $0 \leq g < 1$: two degenerate ground states in the thermodynamic limit aligned along the $z$-axis, each with an energy gap $2(1-g)$.
\end{enumerate}
The set of KMS ground states for transverse-field Ising model is fully characterized by~\cite[Theorem~1]{araki1985ground}.
For $0 < g <1$, there exist exactly two pure KMS ground states whose equal classical mixture is the unique KMS ground state that is $\mathds{Z}_2$-symmetric under the mapping $Y_i \mapsto -Y_i$, $Z_i \mapsto -Z_i$.
By \cref{prop:NonPureKMSGroundStateDegenerate}, this $\mathds{Z}_2$-symmetric equal classical mixture has a locally non-degenerate bulk gap of $0$.
For $g \geq 1$, there exists a unique pure KMS ground state and it is also $\mathds{Z}_2$-symmetric.

Recall that our SDP hierarchy tests the necessary condition for a locally non-degenerate bulk gap $\gamma$ over all KMS ground states of the given dynamical system.  
In the numerical tests here, we impose two symmetries of $H_{\mathrm{TFIM},g}$ using the symmetry-restricted hierarchy of \cref{rem:SymmetryRestrictedHierarchy}: the spin flip $\prod_i X_i$, a complex-linear $*$-automorphism acting as
\begin{align}\label{eq:TFIMSpinFlip}
    (X_i,Y_i,Z_i)\mapsto(X_i,-Y_i,-Z_i),
\end{align}
and complex conjugation in the eigenbasis of the $Z_i$, a conjugate-linear $*$-automorphism acting as $(X_i,Y_i,Z_i)\mapsto(X_i,-Y_i,Z_i)$.
Both leave $H_{\mathrm{TFIM},g}$ invariant, since its coefficients are real in this basis and $Z_i$ enters only through $Z_iZ_{i+1}$.
The resulting hierarchy therefore certifies only the absence of a gapped $\mathds{Z}_2$-symmetric KMS ground state.
For $g \geq 1$, this restriction does not change the result, since the unique pure KMS ground state is $\mathds{Z}_2$-symmetric.
For $0 \leq g < 1$, the symmetry constraint selects the unique $\mathds{Z}_2$-symmetric equal classical mixture which has a bulk gap of $0$.
We impose mirror symmetry $i \mapsto -i$ to further improve the numerical tightness.

\begin{table}
\centering
\caption{
Largest feasible gap parameters returned by the SDP hierarchy for the transverse-field Ising chain.
Within each value of $g$, the first block shows results with the spin-flip $\mathds{Z}_2$ symmetry imposed, while the second shows results without this symmetry.
Rows correspond to the system size parameter $L$, and columns to the monomial degree $d$.
Mirror symmetry $i\mapsto-i$ is imposed throughout.
Missing entries indicate values not computed within the available computational budget.
}
\label{tab:IsingNumericalBounds}

\begin{minipage}[t]{0.48\linewidth}
\centering
\begin{tabular}{c|cc}
\toprule
\multicolumn{3}{c}{$g=0.5$} \\
\midrule
\multicolumn{3}{c}{with spin-flip $\mathds{Z}_2$ symmetry} \\
\midrule
$L$ & $d=2$ & $d=3$ \\
\midrule
2 & 0.47 & 0.44 \\
3 & 0.29 & 0.24 \\
4 & 0.25 & 0.15 \\
5 & 0.24 & 0.10 \\
6 & 0.24 & 0.08 \\
7 & 0.23 & --   \\
\midrule
\multicolumn{3}{c}{without spin-flip $\mathds{Z}_2$ symmetry} \\
\midrule
$L$ & $d=2$ & $d=3$ \\
\midrule
2 & 3.80 & 3.22 \\
3 & 3.41 & 2.93 \\
4 & 3.30 & --   \\
5 & 3.24 & --   \\
\bottomrule
\end{tabular}
\end{minipage}
\hfill
\begin{minipage}[t]{0.48\linewidth}
\centering
\begin{tabular}{c|cc}
\toprule
\multicolumn{3}{c}{$g=1.0$} \\
\midrule
\multicolumn{3}{c}{with spin-flip $\mathds{Z}_2$ symmetry} \\
\midrule
$L$ & $d=2$ & $d=3$ \\
\midrule
2 & 1.15 & 1.14 \\
3 & 0.80 & 0.78 \\
4 & 0.65 & 0.59 \\
5 & 0.60 & 0.47 \\
6 & 0.57 & 0.40 \\
7 & 0.57 & --   \\
\midrule
\multicolumn{3}{c}{without spin-flip $\mathds{Z}_2$ symmetry} \\
\midrule
$L$ & $d=2$ & $d=3$ \\
\midrule
2 & 3.62 & 2.74 \\
3 & 3.01 & 2.09 \\
4 & 2.69 & --   \\
5 & 2.55 & --   \\
\bottomrule
\end{tabular}
\end{minipage}

\vspace{1em}

\begin{minipage}[t]{0.48\linewidth}
\centering
\begin{tabular}{c|cc}
\toprule
\multicolumn{3}{c}{$g=1.5$} \\
\midrule
\multicolumn{3}{c}{with spin-flip $\mathds{Z}_2$ symmetry} \\
\midrule
$L$ & $d=2$ & $d=3$ \\
\midrule
2 & 1.98 & 1.98 \\
3 & 1.58 & 1.57 \\
4 & 1.39 & 1.37 \\
5 & 1.30 & 1.26 \\
6 & 1.25 & 1.19 \\
7 & 1.22 & --   \\
\midrule
\multicolumn{3}{c}{without spin-flip $\mathds{Z}_2$ symmetry} \\
\midrule
$L$ & $d=2$ & $d=3$ \\
\midrule
2 & 3.51 & 2.87 \\
3 & 2.82 & 2.05 \\
4 & 2.46 & --   \\
5 & 2.24 & --   \\
\bottomrule
\end{tabular}
\end{minipage}
\hfill
\begin{minipage}[t]{0.48\linewidth}
\centering
\begin{tabular}{c|cc}
\toprule
\multicolumn{3}{c}{$g=2.0$} \\
\midrule
\multicolumn{3}{c}{with spin-flip $\mathds{Z}_2$ symmetry} \\
\midrule
$L$ & $d=2$ & $d=3$ \\
\midrule
2 & 2.88 & 2.88 \\
3 & 2.47 & 2.47 \\
4 & 2.29 & 2.29 \\
5 & 2.20 & 2.19 \\
6 & 2.15 & 2.14 \\
7 & 2.12 & --   \\
\midrule
\multicolumn{3}{c}{without spin-flip $\mathds{Z}_2$ symmetry} \\
\midrule
$L$ & $d=2$ & $d=3$ \\
\midrule
2 & 3.75 & 3.41 \\
3 & 2.95 & 2.68 \\
4 & 2.60 & --   \\
5 & 2.44 & --   \\
\bottomrule
\end{tabular}
\end{minipage}
\end{table}

We compute the largest feasible gap parameter at selected relaxation levels $(L,d)$ for $g\in\{0.5,1.0,1.5,2.0\}$, with and without the spin-flip $\mathds{Z}_2$ symmetry imposed.
The results are reported in \cref{tab:IsingNumericalBounds}.
As expected from the nested relaxations, the maximal feasible values are nonincreasing as either $L$ or $d$ is increased.

For $g=1.5$ and $g=2.0$, the upper bounds approach the known value $2(g-1)$ from above.
For $g=1$, the decrease is slower, likely due to critical finite-size scaling.
For $g=0.5$, the symmetry-restricted hierarchy returns bounds decreasing towards $0$, consistent with testing the $\mathds{Z}_2$-symmetric equal classical mixture, whose locally non-degenerate bulk gap is $0$.

\cref{tab:IsingNumericalBounds} also illustrates the different numerical roles of $L$, $d$, and the imposed symmetries.
For the ordered case $g=0.5$, both increasing the system size $L$ and increasing the monomial degree from $d=2$ to $d=3$ lead to substantial improvements, while increasing further to $d=4$ gives essentially no additional gain within the range considered; the best bound obtained is $0.08$ at $L=6$ and $d=3$.
For the other values of $g$, increasing $L$ gives the more significant improvement, whereas increasing $d$ produces only modest corrections once $d=2$ or $3$ is reached.
Finally, without imposing the spin-flip symmetry, the computations become substantially more expensive and the resulting bounds are much looser: within the same computational budget we could reach at most $L=5$ for $d=2$ and $L=3$ for $d=3$.

\subsection{Kagome lattice Heisenberg model}\label{sec:KagomeHeisenberg}
The kagome lattice is a two-dimensional periodic graph $(V,E)$ built from corner-sharing triangles, where $V$ is the set of vertices and $E$ is the set of nearest-neighbour edges; see \cref{fig:KagomeLattice}.
It is a canonical setting for geometric frustration and has become a central playground for the study of quantum spin liquids.
\begin{figure}[!ht]
    \centering
    \includegraphics[width=0.32\linewidth]{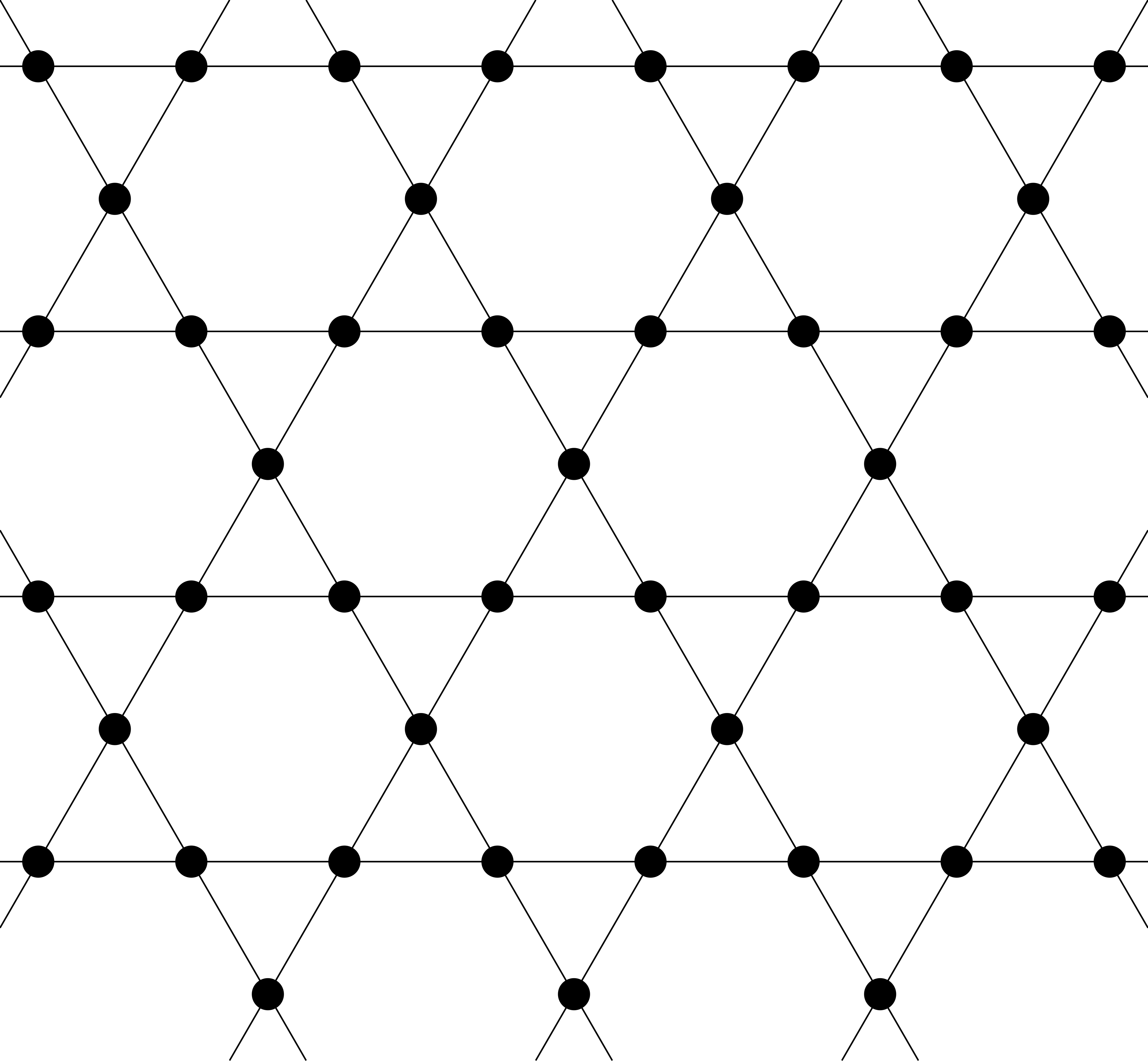}
    \caption{A portion of the kagome lattice. From ``Kagome-lattice-bw.svg'' by WilliamSix (Wikimedia Commons), \href{https://creativecommons.org/licenses/by-sa/4.0/}{CC BY-SA 4.0}.}
    \label{fig:KagomeLattice}
\end{figure}

We express the kagome lattice as the limit of finite sublattices.
To this end, we fix a distinguished vertex $v_0 \in V$ and let $\mathrm{dist}(\cdot,\cdot)$ denote the graph distance on $(V,E)$ (i.e., the length of a shortest edge-path).
For each $L\in\mathds{N}$ we define the finite region
\begin{align}\label{eq:KagomeFinitePatches}
    \Lambda_{\mathrm{kagome}}(L) \coloneqq \{\, v\in V : \mathrm{dist}(v,v_0)\leq L \,\}.
\end{align}
Then $(\Lambda_{\mathrm{kagome}}(L))_{L\in\mathds{N}}$ is an increasing exhaustion of $V$ by finite subsets.

Consider a spin-$\frac{1}{2}$ degree of freedom at each site $i$, with local algebra $\cA_i \simeq M_2(\mathds{C})$ and Pauli matrices $X_i, Y_i, Z_i\in\cA_i$.
Clearly, $\cA_{\Lambda_{\mathrm{kagome}}(L)} \coloneqq \bigotimes_{i \in \Lambda_{\mathrm{kagome}}(L)} \cA_i$ is finite-dimensional and the algebra of quasi-local operators is $\cA \coloneqq \overline{\cup_L \cA_{\Lambda_{\mathrm{kagome}}(L)}}$.
On $\Lambda_{\mathrm{kagome}}(L)$, the spin-$\frac{1}{2}$ kagome lattice Heisenberg model (KLHM) has the local Hamiltonian
\begin{align}\label{eq:KagomeLocalHamiltonian}
    H^{\Lambda_{\mathrm{kagome}}(L)}_{\mathrm{KLHM}} = J \sum_{\langle i,j\rangle\in E: \  i,j \in \Lambda_{\mathrm{kagome}}(L)} \frac{1}{4}(X_i X_j + Y_i Y_j + Z_i Z_j),
\end{align}
for a fixed $J > 0$.

The numerical consensus suggests that the ground state of KLHM is consistent with quantum spin liquids.
However, whether the thermodynamic-limit spectrum is gapped or gapless remains actively debated.
DMRG studies on long kagome cylinders report evidence compatible with a gapped ground state~\cite{yan2011spin,depenbrock2012nature}, with an estimated singlet gap of $0.04$--$0.05J$~\cite{yan2011spin} and an estimated spin gap of $\approx 0.13J$ from cylinder extrapolations~\cite{depenbrock2012nature}.
(Recall that spin and singlet gaps are defined by restricting the allowed excitations to fixed spin sectors; the numerical estimates are not certified upper bounds on the bulk spectral gap, which includes all excitations.
Furthermore, the values cited are estimates and should not be identified directly with the bulk spectral gap.)
Subsequent DMRG work on infinitely long cylinders, however, argued that the spin gap may be significantly smaller and favors a gapless Dirac quantum spin liquid~\cite{he2017signatures}.
Tensor-network approaches based on projected entangled pair/simplex states have also argued for a gapless ground state~\cite{liao2017gapless,jiang2019competing}.
Exact diagonalization~\cite{lauchli2016s} is restricted to comparatively small clusters (e.g., $48$ sites) and has not resolved the issue conclusively.
We refer to the~\cite{zhu2025quantum} for a modern overview and further references.

\begin{remark}[Analytical upper bounds on bulk gap from the ground-state energy]\label{rem:KagomeTrivialUpperBound}
    Before presenting the numerical bounds, we record an analytical upper bound obtained from the bulk-gap inequality \cref{eq:GappedSystemProperty,def:BulkSpectralGap} by choosing single-site spin observables.
    Write $\boldsymbol{S}_i \coloneqq (X_i,Y_i,Z_i)/2$, then the interaction term of \cref{eq:KagomeLocalHamiltonian} becomes $J \boldsymbol{S}_i \cdot \boldsymbol{S}_j$.
    Use $[H,a]$ for the local commutator $-\mathrm{i}\delta_{\mathcal{F}_{\mathrm{KLHM}}}(a)$.
    If a KMS ground state $\omega$ satisfies \cref{eq:GappedSystemProperty} with $\gamma>0$, for every self-adjoint local observable $a$ with positive variance $\omega(a^2)-\omega(a)^2$, we have
    \begin{align}\label{eq:UpperBoundInequalityCompute}
        \gamma \leq \frac{\omega(a[H,a])}{\omega(a^2)-\omega(a)^2} = \frac{\frac{1}{2}\omega([a,[H,a]])}{\omega(a^2)-\omega(a)^2}.
    \end{align}

    As we justify in \cref{rem:KagomeNecessarySpinSymmetry}, every such state is $SU(2)$-invariant, so $\omega(S_i^\alpha)=0$ for $\alpha\in\{x,y,z\}$.
    Moreover, $(S_i^\alpha)^2=I/4$ and $\omega$ is normalized, so $\omega((S_i^\alpha)^2)-\abs{\omega(S_i^\alpha)}^2=\frac{1}{4}$.
    For $a=S_i^z$, stationarity gives $\omega([H,(S_i^z)^2])=0$, which implies that $\omega([H,S_i^z]S_i^z)=-\omega(S_i^z[H,S_i^z])$.
    Thus, we can calculate
    \begin{equation}\label{eq:KagomeSingleSiteSumRule}
        \begin{aligned}
            \frac{\gamma}{4} \leq \omega(S_i^z[H,S_i^z]) = \frac{1}{2}\omega([S_i^z,[H,S_i^z]]) = -\frac{J}{2}\sum_{j\sim i}\omega(S_i^xS_j^x+S_i^yS_j^y).
        \end{aligned}
    \end{equation}
    The same calculation for $a=S_i^x$ and $a=S_i^y$ gives
    \begin{equation}\label{eq:KagomeOtherSpinComponents}
        \begin{aligned}
            \frac{\gamma}{4} \leq -\frac{J}{2}\sum_{j\sim i}\omega(S_i^yS_j^y+S_i^zS_j^z),\qquad \frac{\gamma}{4} \leq -\frac{J}{2}\sum_{j\sim i}\omega(S_i^zS_j^z+S_i^xS_j^x).
        \end{aligned}
    \end{equation}
    Adding these three inequalities, at every site $i$, we obtain
    \begin{align}\label{eq:KagomeSitewiseGapBound}
        \gamma\leq-\frac{4J}{3}\sum_{j\sim i}\omega(\boldsymbol{S}_i\cdot\boldsymbol{S}_j).
    \end{align}

    Let $b_R$ be the number of edges with exactly one endpoint in $\Lambda_{\mathrm{kagome}}(R)$, i.e., the number of bonds crossing the boundary of $\Lambda_{\mathrm{kagome}}(R)$.
    Averaging the \cref{eq:KagomeSitewiseGapBound} over all sites in $\Lambda_{\mathrm{kagome}}(R)$, we have
    \begin{equation}\label{eq:KagomeAveragedGapEnergy}
        \begin{aligned}
            \gamma &\leq - \frac{4J}{3|\Lambda_{\mathrm{kagome}}(R)|}\sum_{i\in\Lambda_{\mathrm{kagome}}(R)} \sum_{i \sim j} \omega(\boldsymbol{S}_i\cdot\boldsymbol{S}_j)\\
            &= -\frac{8}{3}\frac{\omega(H^{\Lambda_{\mathrm{kagome}}(R)}_{\mathrm{KLHM}})}{|\Lambda_{\mathrm{kagome}}(R)|} - \frac{4J}{3|\Lambda_{\mathrm{kagome}}(R)|}\sum_{\substack{ i \sim j\\ i\in\Lambda_{\mathrm{kagome}}(R),\ j\notin\Lambda_{\mathrm{kagome}}(R)}}\omega(\boldsymbol{S}_i\cdot\boldsymbol{S}_j) \\
            & \leq -\frac{8}{3}\frac{\omega(H^{\Lambda_{\mathrm{kagome}}(R)}_{\mathrm{KLHM}})}{|\Lambda_{\mathrm{kagome}}(R)|} + \frac{Jb_R}{|\Lambda_{\mathrm{kagome}}(R)|}
        \end{aligned}
    \end{equation}
    since $\|\boldsymbol{S}_i\cdot\boldsymbol{S}_j\|=3/4$.
    The absolute value of the second term tends to zero as $R\to\infty$.
    Define the thermodynamic ground-state energy per site by
    \begin{align}\label{eq:KagomeGroundEnergyDensity}
        e_0 \coloneqq \lim_{R\to\infty}\frac{\lambda_0(H^{\Lambda_{\mathrm{kagome}}(R)}_{\mathrm{KLHM}})}{|\Lambda_{\mathrm{kagome}}(R)|}.
    \end{align}
    This limit exists for the finite-range, translation-invariant interaction along these regions.
    Using the bound $\omega(H^{\Lambda_{\mathrm{kagome}}(R)}_{\mathrm{KLHM}})\geq\lambda_0(H^{\Lambda_{\mathrm{kagome}}(R)}_{\mathrm{KLHM}})$ in \cref{eq:KagomeAveragedGapEnergy} and taking the supremum over all such $\gamma$ gives us the analytical upper bound
    \begin{align}\label{eq:KagomeGapEnergyBound}
        \Delta_{\mathrm{bulk}}\leq-\frac{8}{3}e_0.
    \end{align}

    An elementary triangle-cluster estimate gives $e_0\geq-J/2$ in our normalization~\cite[note~29]{misguich2005specific}.
    Combining this energy-density lower bound with \cref{eq:KagomeGapEnergyBound} gives
    \begin{align}\label{eq:KagomeAnalyticalUpperBound}
        \Delta_{\mathrm{bulk}}\leq-\frac{8}{3}e_0\leq\frac{4J}{3}.
    \end{align}
    A stronger cluster lower bound is reported in~\cite[Fig.~2 and note~29]{misguich2005specific}: $e_0\geq-0.454976J$ in our normalization, corresponding to $-0.909952$ for their Hamiltonian $2\sum_{\langle i,j\rangle}\boldsymbol{S}_i\cdot\boldsymbol{S}_j$.
    The value is claimed to be rigorous with a $24$-site cluster calculation.
    Combined with \cref{eq:KagomeGapEnergyBound}, the reported bound gives $\Delta_{\mathrm{bulk}}\leq1.214J$ rounded up.
    For comparison, substituting the DMRG estimate $e_0/J=-0.4386(5)$ from~\cite{depenbrock2012nature} into $-8e_0/3$ gives approximately $1.17J$.
    This last value is not a certified upper bound, since the extrapolated energy is not a rigorous lower bound on $e_0$.
\end{remark}

\begin{remark}[Necessary spin-rotation symmetry]\label{rem:KagomeNecessarySpinSymmetry}
    Theorem~3.3 of~\cite{becker2025automorphicequivalencegappedphases} shows that a locally non-degenerate gapped ground state (aka ``locally unique'' in~\cite[Definition~2.8]{becker2025automorphicequivalencegappedphases}) is invariant under any continuous symmetry generated by a sufficiently decaying interaction that commutes with the dynamics.
    % Stated in~\cite[p.~2]{becker2025automorphicequivalencegappedphases}, their results also apply to spin systems.

    We apply \cite[Theorem~3.3]{becker2025automorphicequivalencegappedphases} to the global $SU(2)$ spin-rotation symmetry of the KLHM.
    Group the spins in each up-pointing triangle into one site $i\in\mathds{Z}^2$.
    These triangles are disjoint and cover the kagome lattice, giving a spin system on $\mathds{Z}^2$ with local algebra $M_8(\mathds{C})$.
    Each site $i$ contains three vertices $r\in\{1,2,3\}$ from the original kagome lattice.
    For $r\in\{1,2,3\}$, let $X_{i,r},Y_{i,r},Z_{i,r}$ denote the Pauli operators on the $r$th spin in cell $i$, and define $(S_{i,r}^x,S_{i,r}^y,S_{i,r}^z)\coloneqq\frac{1}{2}(X_{i,r},Y_{i,r},Z_{i,r})$.
    For each $\mu\in\{x,y,z\}$, define the range-zero interaction
    \begin{align}\label{eq:KagomeSpinRotationGenerator}
        \Phi^\mu(\{i\})\coloneqq\sum_{r=1}^{3}S_{i,r}^\mu,\qquad
        \Phi^\mu(F)\coloneqq0\quad\text{if }|F|\ne1,
    \end{align}
    for any finite sets $F\subset\mathds{Z}^2$.
    Since $\|\Phi^\mu(\{i\})\|\leq3/2$ uniformly in $i$, this interaction belongs to $\mathcal{P}_{\infty}$  in the notation of \cite{becker2025automorphicequivalencegappedphases}.
    It generates the automorphisms $\alpha_t^\mu$ that simultaneously rotate every spin around axis $\mu$; the rotations about the three axes generate the global $SU(2)$ action.
    The KLHM interaction is also uniformly bounded and finite range, hence belongs to $\mathcal{P}_{\infty}$.
    The Hamiltonian in \cref{eq:KagomeLocalHamiltonian} is invariant under these rotations, so
    \begin{align}\label{eq:KagomeRotationDynamicsCommutation}
        \delta_{\mathcal{F}_{\mathrm{KLHM}}}(\alpha_t^\mu(a))
        =\alpha_t^\mu(\delta_{\mathcal{F}_{\mathrm{KLHM}}}(a))
    \end{align}
    for every local observable $a$ and every $t\in\mathds{R}$.
    
    This verifies the symmetry assumptions of~\cite[Theorem~3.3]{becker2025automorphicequivalencegappedphases}.
    Therefore $\omega\circ\alpha_t^\mu=\omega$ for every KMS ground state satisfying \cref{eq:GappedSystemProperty} with $\gamma>0$.
    Since these rotations generate $SU(2)$, every such state is $SU(2)$-invariant.
    Consequently, spin-rotation symmetry is a necessary condition for all non-degenerate KMS ground state with $\gamma >0$, and imposing it in our hierarchy still gives the general certified bounds on $\Delta_{\mathrm{bulk}}$.
\end{remark}

By \cref{rem:ScopeOfConstruction}, our algorithm generalizes to KLHM using $\Lambda_{\mathrm{kagome}}(L)$.
In view of \cref{rem:KagomeNecessarySpinSymmetry}, we impose full $SU(2)$ spin-rotation symmetry in our numerical implementation, using \cref{rem:SymmetryRestrictedHierarchy}.
For example, this includes the three $\pi$-rotations
\begin{equation}\label{eq:KagomePiSpinRotations}
    \begin{aligned}
        R_x &: (X_i,Y_i,Z_i)\mapsto (X_i,-Y_i,-Z_i),\\
        R_y &: (X_i,Y_i,Z_i)\mapsto (-X_i,Y_i,-Z_i),\\
        R_z &: (X_i,Y_i,Z_i)\mapsto (-X_i,-Y_i,Z_i),
    \end{aligned}
\end{equation}
for every site $i$.
Invariance under these rotations enforces $\omega(X_i)=\omega(Y_i)=\omega(Z_i)=0$.
It also includes the cyclic permutations of the spin components,
\begin{align}\label{eq:KagomeCyclicSpinPermutations}
    (X_i,Y_i,Z_i)\mapsto(Y_i,Z_i,X_i),\qquad (X_i,Y_i,Z_i)\mapsto(Z_i,X_i,Y_i).
\end{align}
Together with $R_x,R_y,R_z$, these generate the tetrahedral group $T\subset SO(3)$, a finite subgroup of the spin-rotation action induced by $SU(2)$.

We also strengthen the relaxations with selected necessary stationarity constraints and positivity of local reduced density matrices, as described in~\cite[Sections~5.1 and~5.2]{wang2026scalable}.
At monomial degrees $d=2,3$, the relaxation hierarchy gives the certified bounds in \cref{tab:KagomeNumericalBounds}; see also \cref{fig:MainTextKagomeResult} in the main text.

\begin{table}[ht]
    \centering
    \caption{Certified upper bounds on the bulk spectral gap $\Delta_{\mathrm{bulk}}$ of the KLHM, obtained using full $SU(2)$ spin-rotation symmetry and the additional constraints described above.
    The values $L=1,2,3,4$ correspond to kagome patches of $5,13,27,45$ sites, respectively.}
    \label{tab:KagomeNumericalBounds}
    \begin{tabular}{c|cc}
        \toprule
        $L$ & $d=2$ & $d=3$ \\
        \midrule
        1 & $2.01J$ & $2.01J$ \\
        2 & $1.29J$ & $1.25J$ \\
        3 & $1.24J$ & $1.02J$ \\
        4 & $1.19J$ & $0.94J$ \\
        \bottomrule
    \end{tabular}
\end{table}

The best certified upper bound is obtained at $(L,d)=(4,3)$ and gives
\begin{align}\label{eq:KagomeBestCertifiedGapBound}
    \Delta_{\mathrm{bulk}}\leq0.94J.
\end{align}
This improves the elementary analytical bound $4J/3$ in \cref{rem:KagomeTrivialUpperBound} by $29.5\%$, and the approximately $1.214J$ bound from~\cite{misguich2005specific} by approximately $22.6\%$.
Thus, we certify that no KMS ground state of the KLHM can have a locally non-degenerate bulk spectral gap larger than $0.94J$.

% The necessary spin symmetries substantially improve the numerical tightness at the levels computed.
% Indeed, without imposing them, the hierarchy gives the certified upper bound $2.18J$ at $(L,d)=(2,2)$, compared with $1.31J$ after imposing them.

We also illustrate our algorithm by bounding observables for gapped KMS ground states.
Consider the scalar chirality on an oriented triangle with vertices $i,j,k$,
\begin{align}\label{eq:KagomeScalarChirality}
    \chi_{ijk}\coloneqq\boldsymbol{S}_i\cdot(\boldsymbol{S}_j\times\boldsymbol{S}_k)=\frac{1}{8}\bigl(X_iY_jZ_k+Y_iZ_jX_k+Z_iX_jY_k-X_iZ_jY_k-Z_iY_jX_k-Y_iX_jZ_k\bigr).
\end{align}
This observable is invariant under spin rotations, so the necessary symmetries above do not force its expectation to vanish.
Its eigenvalues are $0$ and $\pm\sqrt{3}/4$, giving the state-independent bound $\abs{\omega(\chi_{ijk})}\leq\sqrt{3}/4\approx0.433013$.
Since the interaction is time-reversal invariant, one directly checks that time reversal preserves both the KMS ground-state condition and the locally non-degenerate bulk-gap inequality with the same gap parameter.
It also changes the sign of $\chi_{ijk}$, so our upper bounds on $\omega(\chi_{ijk})$ apply equally to $\abs{\omega(\chi_{ijk})}$.

\begin{table}[ht]
    \centering
    \caption{Upper bounds on $\abs{\omega(\chi_{ijk})}$ for KLHM KMS ground states with locally non-degenerate bulk gap at least $\gamma$.
    The $SU(2)$ spin-rotation symmetry and the additional necessary constraints described above are imposed throughout.
    Values are reported to six decimal places.}
    \label{tab:KagomeChiralityBounds}
    \begin{tabular}{cc|cc}
        \toprule
        $L$ & $d$ & $\gamma=0.5J$ & $\gamma=0.9J$ \\
        \midrule
        1 & 2 & 0.433013 & 0.433013 \\
        1 & 3 & 0.433013 & 0.433013 \\
        2 & 2 & 0.433013 & 0.421291 \\
        2 & 3 & 0.377482 & 0.298236 \\
        3 & 2 & 0.432651 & 0.400294 \\
        3 & 3 & 0.355236 & 0.253990 \\
        \bottomrule
    \end{tabular}
\end{table}

The bounds in \cref{tab:KagomeChiralityBounds} are nonincreasing as $L$, $d$, or the prescribed gap $\gamma$ increases.
At $(L,d)=(3,3)$, we obtain
\begin{equation}\label{eq:KagomeChiralityBound}
    \begin{aligned}
        \abs{\omega(\chi_{ijk})}& \leq 0.355236
        \qquad(\gamma=0.5J),\\
        \abs{\omega(\chi_{ijk})}& \leq 0.253990
        \qquad(\gamma=0.9J).
    \end{aligned}
\end{equation}
These bounds are approximately $18.0\%$ and $41.3\%$ below the maximum magnitude $\sqrt{3}/4$, respectively, and apply to every KMS ground state with locally non-degenerate bulk gap at least the corresponding $\gamma$.

To conclude, the bound in \cref{eq:KagomeBestCertifiedGapBound} lies below both values that the energy density analysis of \cref{rem:KagomeTrivialUpperBound} can give, the $1.214J$ following from the rigorous cluster bound of~\cite{misguich2005specific} and the $1.17J$ following from the DMRG energy density of~\cite{depenbrock2012nature}, which is itself not certified.
It remains far above the gap values reported by variational studies.
These are statements of different kinds rather than the same quantity at different accuracies.
A variational study estimates the gap on finite systems and extrapolates, whereas \cref{eq:KagomeBestCertifiedGapBound} is a proof about the infinite system that no later calculation can overturn.
Sharper certified bounds would therefore be physically informative: driving the bound below the proposed singlet gap of~\cite{yan2011spin} or the spin gap of~\cite{depenbrock2012nature} would establish a bulk excitation scale below those estimates, and driving it to zero would certify that no KMS ground state of the KLHM has a locally non-degenerate bulk spectral gap.
The calculation is not optimized, and we leave an optimized numerical implementation to future work.

\section{Comparison with finite-size spectral-gap undecidability}\label{sec:DiscussionSemiDecidability}
In \cref{sec:CompleteSDPHierarchyMainSec,sec:BulkGapUndecidability}, we established complementary algorithmic results for the algebraic bulk spectral gap: the SDP hierarchy gives convergent upper bounds, whereas deciding whether the bulk gap is positive is undecidable and no general algorithm produces convergent lower bounds.
We now compare these results with the spectral-gap undecidability result of~\cite{cubitt2022undecidability}, formulated through finite-volume Hamiltonians under fixed boundary conditions.
The distinction concerns which boundary conditions are allowed and which excitations are counted: the bulk Hamiltonian $H_\omega$ sees only excitations reachable from $\omega$ by local observables, while a finite-volume spectrum includes every eigenstate, including states on the boundary.
The difference is not only definitional but also algorithmic.
Under a prescribed boundary condition the gap admits convergent bounds from neither side, whereas the bulk gap admits them from above.

In \cref{sec:FiniteVolumeSpectralGap}, we recall the finite-size gap notions and their relationship with KMS ground states.
In \cref{sec:CPWResultRevisit}, we apply the semidecision procedure of \cref{alg:SpectralGapSemi} to the construction of~\cite{cubitt2022undecidability}.
We show that, for some inputs, strong gaplessness under open boundary conditions coexists with a KMS ground state having a positive locally non-degenerate bulk gap.
This explains why the finite-size undecidability result of~\cite{cubitt2022undecidability} does not directly establish the bulk conclusion proved in \cref{thm:BulkGapUndecidability}, showcasing distinct algorithmic properties.

\subsection{Finite-volume spectral gaps and boundary conditions}\label{sec:FiniteVolumeSpectralGap}
We recall the finite-size notion of spectral gap used in~\cite{cubitt2022undecidability}, which is formulated in terms of the spectra of finite-volume Hamiltonians along a fixed choice of boundary condition.
This dependence will be central to the comparison with algebraic bulk notions.

We start by formalizing the boundary conditions.
\begin{definition}\label{def:BoundaryCondition}
Let $\{H^{\latticeL}\}$ be a family of local Hamiltonians on $\mathds{Z}^{D}$ associated with a finite-range interaction $\Phi$.
A \emph{boundary condition} on $\latticeL$ is a choice of a self-adjoint operator $B^{\latticeL} \in \cA$ whose support is contained in an $R$-neighbourhood of the boundary $\partial\latticeL$, for some fixed $R$ independent of $L$.
We set the family of Hamiltonians with the boundary condition $\{B^{\latticeL}\}$ to be
\begin{align}
    H^{\latticeL}_B \coloneqq H^{\latticeL} + B^{\latticeL},
\end{align}
which is again finitely supported for every $L$.

The \emph{open boundary condition} (OBC) corresponds to $B^{\latticeL}_{\OBC} \coloneqq 0$, that is, $H^{\latticeL}_{\OBC} = H^{\latticeL}$.
For the \emph{periodic boundary condition} (PBC), we replace $\latticeL$ with the discrete torus $\mathds{T}_L^D \coloneqq (\mathds{Z}/L\mathds{Z})^D$ and define $H^{\latticeL}_{\PBC}$ by the same interaction $\Phi$, with sites identified modulo $L$.
(In the nearest-neighbour case, this is equivalent to adding wrap-around couplings across opposite faces.)
\end{definition}

Fix a boundary condition $\{B^{\latticeL}\}_L$.
Let $\mathrm{Spec}(H^{\latticeL}_B) = \{\lambda_0(H^{\latticeL}_B), \lambda_1(H^{\latticeL}_B), \dots\}$ denote the spectrum with non-decreasing ordering $\lambda_0(H^{\latticeL}_B) \leq \lambda_1(H^{\latticeL}_B) \leq \cdots$.
Define the {ground state energy} as the smallest eigenvalue $\lambda_0(H^{\latticeL}_B)$ and call the associated normalized eigenvector a ground state.
The ground state energy is said to be {non-degenerate} if $\lambda_0(H^{\latticeL}_B) < \lambda_1(H^{\latticeL}_B)$, that is, the dimension of the eigenspace corresponding to the lowest eigenvalue is $1$.
There are two important classes of behaviour of the spectrum $\SpectrumNoBracket{(H^{\latticeL}_B)}$ at the thermodynamic limit as $L \to \infty$.

\begin{definition}[Non-degenerate uniform gap with a boundary condition]\label{def:GappedCPW15}
    Consider the family $\{H^{\latticeL}\}_L$ of Hamiltonians with the boundary condition $\{B^{\latticeL}\}$.
    We say that the system under the boundary condition $\{B^{\latticeL}\}$ is non-degenerate uniformly gapped if $\exists \gamma > 0, \exists L_0$, $\forall L > L_0$ the ground state energy $\lambda_0(H^{\latticeL}_B)$ is non-degenerate, and
    \begin{align}
        \Spectrum{H^{\latticeL}_B} \cap \left(\lambda_0(H^{\latticeL}_B), \lambda_0(H^{\latticeL}_B) + \gamma\right) = \emptyset.
    \end{align}
    In this case, we say that the system under the boundary condition $\{B^{\latticeL}\}$ has a non-degenerate uniform spectral gap of at least $\gamma$.
\end{definition}

\begin{remark}\label{rem:UniformGapVSAlgebraicGap}
    \cref{def:GappedCPW15} is stricter than the conventional definition in the usual physics literature for gapped systems.
    In particular, this notion is stronger than the locally non-degenerate bulk spectral gap on KMS ground states (\cref{def:GappedAlgebraic}) since it requires a uniform bound for excitations in finite systems.
    Indeed, a finite-volume uniform gap for some chosen boundary condition is a lower bound for the bulk gap~\cite[Proposition~5.4]{bachmann2016lieb}.

    The converse implication is false, that is, \cref{def:GappedAlgebraic} does not imply \cref{def:GappedCPW15}.
    First, the local non-degeneracy of \cref{def:GappedAlgebraic} does not necessarily imply the global non-degeneracy of \cref{def:GappedCPW15}.
    In addition, finite-size gaps may close due to boundary (edge) excitations even when the bulk is gapped (e.g., the AKLT model with OBC~\cite{affleck1987rigorous}).
    
    Moreover, as in \cref{rem:NonDegeneracyDifferent}, we remark again that ``non-degeneracy'' in \cref{def:GappedCPW15} refers to the global uniqueness of ground states, whereas the algebraic \cref{def:GappedAlgebraic} refers to non-degeneracy in the GNS representation of KMS ground states.
\end{remark}

We next state the conventional gapless condition.
\begin{definition}[Gapless with a boundary condition]\label{def:GaplessConventional}
    Consider the family $\{H^{\latticeL}\}_L$ of Hamiltonians with the boundary condition $\{B^{\latticeL}\}$.
 We say that the system under the boundary condition $\{B^{\latticeL}\}$ is gapless if $\forall \epsilon >0$, $\exists L_0$, $\forall L > L_0$
    \begin{align}
        \Delta(H^{\latticeL}_B) = \lambda_1(H^{\latticeL}_B) - \lambda_0(H^{\latticeL}_B) < \epsilon.
    \end{align}
\end{definition}
Importantly, being gapless as in \cref{def:GaplessConventional} does not imply that the system is still bulk-gapless (\cref{def:GaplessAlgebraic}), due to the potential existence of edge states associated with the given boundary condition; again, see the AKLT model with open boundary condition.
We remark that the gapless \cref{def:GaplessConventional} is not as stringent as the gapless condition considered in~\cite{cubitt2022undecidability}, where they instead adopt the continuum/uniform condition near the ground state energy.

\begin{definition}[Strongly gapless with a boundary condition]\label{def:StrongGaplessCPW15}
    Consider the family $\{H^{\latticeL}\}_L$ of Hamiltonians with the boundary condition $\{B^{\latticeL}\}$.
 We say that the system under the boundary condition $\{B^{\latticeL}\}$ is strongly gapless if $\exists c>0$, $\forall \epsilon >0$, $\exists L_0$, $\forall L > L_0$ the spectrum $\SpectrumNoBracket{(H^{\latticeL}_B)}$ is $\epsilon$-dense in the interval
    \begin{align}
        [\lambda_0(H^{\latticeL}_B), \lambda_0(H^{\latticeL}_B) + c].
    \end{align}
    That is, every point in the interval lies within a distance $\epsilon$ of some eigenvalue of $H^{\latticeL}_B$.
\end{definition}
This is the finite-volume analogue of the strong bulk-gaplessness in \cref{def:GaplessAlgebraic}.

The strongly gapless condition is strictly stronger than \cref{def:GaplessConventional}:
For example, $H^{\latticeL}_B$ with eigenvalues $\lambda_0, \lambda_k = \lambda_0 + 1/(1+k)$ for $k=1, \dots, L$ is gapless but not strongly gapless, since the spectrum near $\lambda_0$ is too discrete to be $\epsilon$-dense in any nontrivial interval.
Another example is a system with a degenerate ground space but a spectral gap above it (e.g., $\lambda_0=\lambda_1<\lambda_2$): it is gapless in the sense of \cref{def:GaplessConventional} but not strongly gapless.
The AKLT chain with open boundary condition is an example of this type.
On the other hand, $H^{\latticeL}_B$ with $\lambda_0, \lambda_k = \lambda_0 + k/L$ is both gapless and strongly gapless.

Finally, we recall the relationship between weak-$*$ limit points of finite-volume ground states (as the volume grows and boundary conditions vary) and KMS ground states of the associated infinite-volume dynamics.
\begin{remark}\label{rem:RelationLimitGroundAndKMS}
    Consider a family of finite-volume Hamiltonians $\familyH=\{H^{\latticeL}\}_L$ induced by a finite-range interaction $\Phi$, and let $(\cA,\tau_{\familyH},\delta_{\familyH})$ be the associated dynamical system.
    For any boundary condition $\{B^{\latticeL}\}_L$ as in \cref{def:BoundaryCondition}, let $\ket{\psi_B^{\latticeL}}$ be a normalized ground state of $H_B^{\latticeL}$ with eigenvalue $\lambda_0(H_B^{\latticeL})$, and let $\omega_B^{\latticeL}$ denote the corresponding finite-volume vector state,
    \begin{align}\label{eq:BoundaryGroundVectorState}
        \omega_B^{\latticeL}(\cdot)\coloneqq\sandwich{\psi_B^{\latticeL}}{\cdot}{\psi_B^{\latticeL}}.
    \end{align}
    After extending these states to $\cA$, any weak-$*$ limit point of $\{\omega_B^{\latticeL}\}_L$ is a KMS ground state $\omega_B$ on $(\cA,\tau_{\familyH},\delta_{\familyH})$, by the same argument as in \cref{rem:ExistenceKMSGroundState}.
    Indeed, for any fixed local observable $a$ and sufficiently large $L$, the boundary term has disjoint support from $a$, and $H^{\latticeL}$ contains every interaction term meeting its support.
    Thus $[B^{\latticeL},a]=0$ and $\delta_{\familyH}(a)=\mathrm{i}[H_B^{\latticeL},a]$, so the finite-volume ground-state condition gives
    \begin{align}\label{eq:BoundaryGroundStateInequality}
        -\mathrm{i}\omega_B^{\latticeL}(a^*\delta_{\familyH}(a))
        =\omega_B^{\latticeL}(a^*[H_B^{\latticeL},a])\geq0.
    \end{align}
    Passing to the weak-$*$ limit and extending from the local core proves the claim.
    We note that the construction in \cref{rem:ExistenceKMSGroundState} corresponds to the open boundary condition.

    The converse direction requires a more general notion of boundary conditions beyond~\cref{def:BoundaryCondition} from a variational perspective.
    Consider the conditional Hamiltonian on $\latticeL$
    \begin{align}
        \widetilde H^{\latticeL} \coloneqq H^{\latticeL} + W^{\latticeL} ,
        \quad
        W^{\latticeL} \coloneqq \sum_{\substack{S\cap\latticeL\neq\emptyset\\[2pt]
        S\cap\latticeL^{c}\neq\emptyset}}\Phi(S)\,,
    \end{align}
    where the surface term $W^{\latticeL}$ is well-defined since $\Phi$ is assumed to be of finite-range.
    By \cite[Theorem~6.2.52]{BratteliRobinsonVol2}, a state $\omega$ is a KMS ground state if and only if, for every finite region $\latticeL$,
    \begin{align}
        \omega(\widetilde{H}^{\latticeL}) = \inf_{\omega'\in C^\omega_{\latticeL}} \omega'(\widetilde H^{\latticeL}),
    \end{align}
    where
    \begin{align}
        C^\omega_{\latticeL} \coloneqq \{\omega' \text{ is a state on } \cA \mid \omega'|_{\cA_{\latticeL^c}} = \omega|_{\cA_{\latticeL^c}}\}.
    \end{align}
    In this sense, the exterior restriction $\omega|_{\cA_{\latticeL^c}}$ can be regarded as a boundary condition of the finite subsystem $\latticeL$.

    Thus every KMS ground state minimizes the energy in each finite region when its exterior restriction is fixed.
    We interpret this exterior restriction as a generalized boundary condition, in the conditional sense above.
    This is more general than the boundary conditions of \cref{def:BoundaryCondition}.
\end{remark}

\subsection{Revisiting 2D spectral gap undecidability through bulk-gap SDP hierarchy}\label{sec:CPWResultRevisit}
The authors of~\cite{cubitt2022undecidability} construct, for each input $m$, a family of Hamiltonians $\{H^{\latticeL}_{\mathrm{CPW}}(m)\}_L$ on the two-dimensional lattice $\mathds{Z}^2$ with the following properties: the interaction is translationally invariant, nearest-neighbour, and has uniformly bounded coefficients, and the Hilbert space dimension on each site is uniformly bounded.
For open boundary conditions, each such family falls into one of two promised cases: it is either non-degenerate uniformly gapped in the sense of \cref{def:GappedCPW15} with the gap $\gamma_{\mathrm{CPW}}(m) \geq 1$, or strongly gapless in the sense of \cref{def:StrongGaplessCPW15}.

Crucially, the spectral behaviour of $\{H^{\latticeL}_{\mathrm{CPW}}(m)\}_L$ under open boundary conditions encodes the halting problem for a fixed universal Turing machine (UTM):
\begin{equation}\label{eq:HaltingToGapCPW15}
    \begin{cases}
        \textbf{YES} \text{: UTM halts on } m \implies 
        \begin{aligned}[t]
            & \{H^{\latticeL}_{\mathrm{CPW}}(m)\}_L \text{ is non-degenerate uniformly gapped} \\
            & \text{under OBC (\cref{def:GappedCPW15}) with } \gamma_{\mathrm{CPW}}(m) \geq 1,
        \end{aligned}
        \\[2em]
        \textbf{NO} \text{: UTM does not halt on } m \implies 
        \begin{aligned}[t]
             & \{H^{\latticeL}_{\mathrm{CPW}}(m)\}_L \text{ is strongly gapless, and thus gapless under OBC} \\
             & \text{(\cref{def:StrongGaplessCPW15}, and thus \cref{def:GaplessConventional}).}
        \end{aligned}
    \end{cases}
\end{equation}
In particular, deciding whether $\{H^{\latticeL}_{\mathrm{CPW}}(m)\}_L$ is gapped or gapless under open
boundary conditions is equivalent to deciding whether the UTM halts on $m$, and is therefore
undecidable.

We now apply \cref{alg:SpectralGapSemi} based on our SDP hierarchy to the dynamics induced by $\familyH_{\mathrm{CPW}}(m)=\{H^{\latticeL}_{\mathrm{CPW}}(m)\}_L$, with a fixed threshold $\gamma_0\in(0,1)$.
The completeness of our SDP hierarchy (\cref{thm:CompletenessSDPHierarchy}) has the following implications:
\begin{enumerate}
    \item If the input $\{H^{\latticeL}_{\mathrm{CPW}}(m)\}_L$ corresponds to a \textbf{YES}-instance of the UTM, then with the open boundary condition, the system admits a non-degenerate uniform gap $\geq 1$ (in the sense of \cref{def:GappedCPW15}).
    By \cref{rem:UniformGapVSAlgebraicGap}, there exists a KMS ground state with a locally non-degenerate bulk spectral gap $\geq 1$ (in the sense of \cref{def:GappedAlgebraic}) due to~\cite[Proposition~5.4]{bachmann2016lieb}.
    Consequently, every relaxation at threshold $\gamma_0$ is feasible, and \cref{alg:SpectralGapSemi} runs indefinitely.
    \item If \cref{alg:SpectralGapSemi} returns \textbf{excluded} at some $(n_0, d_0)$, then we have a proof that no KMS ground state with a gap $\geq \gamma_0$ exists, including the limiting state $\omega_{\OBC}$ from the ground states due to the open boundary condition (\cref{rem:RelationLimitGroundAndKMS}).
    By the spectral property of $\{H^{\latticeL}_{\mathrm{CPW}}(m)\}_L$, it follows that the system must be gapless with the open boundary condition (\cref{def:GaplessConventional}).
\end{enumerate}
Thus, whenever \cref{alg:SpectralGapSemi} terminates, it certifies that the input is a nonhalting instance and that the corresponding family is strongly gapless under open boundary conditions.

This appears to be in tension with the well-known undecidability results of~\cite{cubitt2022undecidability} by the following tempting but incorrect argument:
\begin{quote}
    (A) Run the UTM and \cref{alg:SpectralGapSemi} simultaneously.
    In the \textbf{YES}-instance of \cref{eq:HaltingToGapCPW15}, the UTM halts in finite time and \cref{alg:SpectralGapSemi} runs indefinitely.
    In the \textbf{NO}-instance of \cref{eq:HaltingToGapCPW15}, the UTM runs indefinitely.
    But, since the system is gapless, \cref{alg:SpectralGapSemi} halts in finite time.
    Therefore, one can fully decide the halting problem in finite time with \cref{alg:SpectralGapSemi}, which is a contradiction.
\end{quote}
It is \emph{false} due to a mix-up of various definitions modeling spectral gaps in the thermodynamic limit and dependence on boundary conditions.
We explain this in the following remark.

\begin{remark}[Compatibility with~\cite{cubitt2022undecidability}]\label{rem:CompatibilityWithCPW15}
Let us analyze two cases of the given UTM on input $m$.
\begin{enumerate}[label=(\roman*)]
    \item \emph{If the UTM halts:}
    This is already discussed by bullet point 1.
    Our SDP hierarchy is feasible at every level and \cref{alg:SpectralGapSemi} runs indefinitely.
    
    \item \emph{If the UTM does not halt:}
    The system is proven to be gapless with open boundary condition in the sense of \cref{def:GaplessConventional}.
    Let $\omega_{\OBC}$ be a weak-$*$ limit point of finite-volume ground states of $\{H^{\latticeL}_{\mathrm{CPW}}(m)\}$ with open boundary conditions.
    However, it is not immediate that this open-boundary gaplessness implies that $\omega_{\OBC}$ is bulk-gapless in the GNS sense.
    Moreover, this does not exclude the possibility that there exists another KMS ground state which admits a locally non-degenerate bulk spectral gap of at least $1$.
    (The existence of another KMS ground state with a locally non-degenerate bulk gap at least $\gamma_0$ would prevent the algorithm from terminating.)
    If this is the case, \cref{alg:SpectralGapSemi} runs indefinitely, and therefore falsifies the Statement (A).
\end{enumerate}
To sum up, the \textbf{NO}-instance of the decision problem \cref{eq:HaltingToGapCPW15} does not mean that \cref{alg:SpectralGapSemi} necessarily halts in finite time, due to, e.g., the possibility of an alternative boundary condition for which the spectral gap is gapped.
\end{remark}

We conclude this section by an observation on strong gaplessness under a fixed boundary condition vs KMS ground states that are locally non-degenerate bulk-gapped.
\begin{proposition}\label{prop:CPW15GappedInDiffBoundaryCondition}
    Consider the family $\{H^{\latticeL}_{\mathrm{CPW}}(m)\}_L$ associated with the fixed UTM in~\cite{cubitt2022undecidability}.
    Fix any threshold $\gamma_0 \in (0, 1)$.
    Then there exists a non-halting input $m'$ and a KMS ground state $\widetilde{\omega}$ with a locally non-degenerate bulk spectral gap at least $\gamma_0$.
\end{proposition}
\begin{proof}
    Suppose not.
    Then, for every nonhalting input $m'$, there is no KMS ground state of the associated dynamical system with a locally non-degenerate bulk gap at least $\gamma_0$.
    By the completeness of the SDP hierarchy (\cref{thm:CompletenessSDPHierarchy}), it follows that running \cref{alg:SpectralGapSemi} with $\gamma_0$ on the instance $\{H^{\latticeL}_{\mathrm{CPW}}(m')\}_L$ must terminate in finite time for every such $m'$.
    Consequently, Statement (A) can be used to contradict the undecidability of the halting problem.
\end{proof}

The proposition shows that strong gaplessness under open boundary conditions can coexist with a locally non-degenerate bulk-gapped KMS ground state.
Thus the open-boundary gaplessness conclusion of~\cite{cubitt2022undecidability} does not by itself imply strong bulk gaplessness.
In contrast, the halting case of \cref{thm:BulkGapUndecidability} gives bulk spectrum $[0,\infty)$ in every KMS ground-state representation, as discussed at the beginning of \cref{sec:BulkGapUndecidability}.

The two reductions also run in opposite directions.
In~\cite{cubitt2022undecidability} a halting input gives a gapped family under open boundary conditions, so an algorithm producing convergent upper bounds on the finite-size gap would semidecide non-halting.
The reduction of~\cite{bausch2020undecidability} sends halting inputs to gapless families, so convergent lower bounds would do the same.
Neither is possible, and the finite-size gap under the open boundary condition is therefore uncomputable from both sides.
By contrast, \cref{thm:CompletenessSDPHierarchy} certifies the bulk gap from above and \cref{thm:BulkGapUndecidability} shows that it is uncomputable from below.
Together, they make the bulk-gaplessness promise problem $\mathsf{RE}$-complete, and the complementary bulk-gap promise problem $\mathsf{coRE}$-complete (\cref{cor:BulkGapNoConvergenceRate}).
Thus, two notions of spectral gap in the thermodynamic limit differ in their algorithmic nature, not only in their definitions.

\end{document}